\documentclass[12pt]{article}
\usepackage[a4paper]{geometry}
\usepackage[utf8]{inputenc}
\usepackage[T1]{fontenc}
\usepackage{lmodern}
\usepackage[british]{babel}
\usepackage{csquotes}
\usepackage[style=alphabetic]{biblatex}
\bibliography{References.bib}
\usepackage{enumitem}
\usepackage{ebproof}
\usepackage{hhline}
\usepackage{ circuitikz }

\usepackage{amsmath}
\usepackage{amsfonts}
\usepackage{amssymb}
\usepackage{ntheorem}
\usepackage{mathrsfs}
\usepackage{stmaryrd}
\usepackage{proof}
\usepackage{phonetic}
\usepackage{upgreek}
\usepackage{tipa}
\usepackage{graphicx}
\usepackage{yhmath}
\usepackage{mathdots}
\usepackage{MnSymbol}
\usepackage{calc}
\usepackage{layout}
\usepackage{thmtools}
\usepackage{fancyhdr}
\usepackage{longtable}
\usepackage[dvipsnames]{xcolor}

\usepackage{tikz}
\usetikzlibrary{shapes, arrows.meta}
\usepackage{tabularx,colortbl}
\usepackage[twopage]{rotating}
\usepackage{musicography}
\usepackage{multirow}
\usepackage{diagbox}
\definecolor{oxfordblue}{RGB}{4,30,66}
 
\allowdisplaybreaks

\theorembodyfont{\fontshape{rm}}

\newtheorem{definition}{Definition}
\newtheorem{theorem}{Theorem}
\newtheorem{lemma}{Lemma}
\newtheorem{corollary}{Corollary}
\newtheorem{conjecture}{Conjecture}

\newtheorem{example}{Example}

\renewtheorem{definition}{Definition}
\newtheorem*{proof}{Proof}

\newcommand{\F}{\mathbf{\mathsf{F}}}

\makeatletter
\providecommand{\leftsquigarrow}{%
  \mathrel{\mathpalette\reflect@squig\relax}%
}
\newcommand{\reflect@squig}[2]{%
  \reflectbox{$\m@th#1\rightsquigarrow$}%
}
\makeatother

\usepackage[flushleft]{threeparttable}

\usepackage{setspace}
\usepackage{tcolorbox}
\usepackage{turnstile}

\usepackage{hyperref}
\title{Inference-Behaviour Semantics for All$^\ast$ Connectives in Two-Dimensional Sequent Calculi}
\author{Sophie Nagler\footnote{Institute for Logic, Language and Computation, University of Amsterdam, and Arché Research Centre, University of St Andrews}}
\date{
\today}

\begin{document}
\maketitle
\begin{abstract}
Inference-behaviour semantics (I-bS) is a new approach to proof-theoretic semantics, grounded in two inferentialist principles: (1) the use of an expression in reasoning determines its meaning, and (2) a connective is defined by its operational rules. I-bS operationalises these ideas by measuring the syntactic use of a connective in the proof of its definability, against a substructurally minimal derivability relation. I-bS thereby gives the meaning of a connective in terms of its semantic clause, i.e. minimal substructural rule pair.

This paper validates and verifies I-bS by applying it to all $10,816$ connective rule pairs that can be formulated in two-dimensional sequent calculi using at most two premiss sequents and at most two active formulae. As a result, we find semantic clauses for exactly $21$ meaningful connectives, namely bottom, top, two negations (intuitionistic and dual-intuitionistic), group and lattice conjunction, disjunction and implication, as well as their converses and inverses.

We use these results to precisely map the semantic interrelations among the connectives, across linear, classical, intuitionistic, dual-intuitionistic, minimal, and lattice logic. Most notably, we find that intuitionistic negation, disjunction and implication each capture half of the meaning of their classical counterparts.
\end{abstract}

\section{Research aim: formalising, validating and verifying inference-behaviour semantics}

Proof-theoretic semantics (P-tS) is a family of methodologies in formal semantics \cite[e.g.][]{wansing2000idea, francez2015proof, sep-proof-theoretic-semantics}. It rests on three main pillars. First, P-tS conceptualises `the meaning of a word [as] its use in the language' \cite{wittgenstein1953philosophische}, thus connecting to the wider programme of \textit{inferentialism} \cite[e.g.][]{ dummett1991logical, brandom1994making, incurvati2023reasoning}. Second, it follows Gentzen in viewing the formal rules of a connective as defining this use \cite{gentzen1935auntersuchungen}. Third, P-tS assumes that formal logical proofs capture (possibly idealised or normative) actual reasoning. Thus, P-tS investigates the meaning of a logical connective by turning to its use in proofs as encoded by its formal rules. 

In this paper, we focus on \textit{inference-behaviour semantics (I-bS)}, a recent approach in the P-tS paradigm \cite{nagler2026measuring}. I-bS operationalises the inferential use of a connective by measuring its \textit{minimal inference behaviour}, i.e. the occurrence pattern the connective exhibits when we prove its \textit{definability} using only minimal structural resources and no other connectives. Although initial results for \textbf{LK}-style sequent calculi are promising \cite{nagler2026measuring}, I-bS currently lacks both a systematic formal implementation and a detailed technical assessment of its broader potential.

This paper aims to remedy these shortcomings by applying I-bS to all $10,816$ connectives whose operational rules can be formulated in two-dimensional sequent calculi with at most two premiss sequents and at most two active formulae. In this way, we seek to minimise design bias in favour of any specific logic. Our testing environment, the \textit{minimal derivability relation}, consists of sequent rules for \textsc{Identity} over primitive formulae and context-differentiating \textsc{Cut}, both defined over multisets. \textit{Conservativity} and \textit{uniqueness} serve as our notion of definability, in line with \cite{belnap1962tonk,nagler2026measuring}.

Among the $150$ connectives that are definable in this minimal setting, we identify $21$ as minimally meaningful, in the sense that they have distinct minimal inference behaviour and minimal premiss types and frameworks. Each of them corresponds to a familiar connective, specifically, top $\top$, bottom $\bot$, two forms of negation $\neg/\sim$---each filling half of the inferential role of classical negation---converse negation $\neg^{cv}$ (`the trivial connective'), additive conjunction $\sqcap$, multiplicative/group conjunction $\otimes$, additive/lattice disjunction $\sqcup$, multiplicative/group disjunction $\oplus$ and each of their converses $^{cv}$, as well as additive right-implication $\rightsquigarrow$, multiplicative right-implication $\rightarrow$, additive left-implication $\leftsquigarrow$, multiplicative left-implication $\leftarrow$, and each of their inverses $^{iv}$ (see Table 1).
\begin{table}[h]
    \centering
    \begin{threeparttable}
    \begin{tabular}{| c | c | c | c | c | c | c | c |}
    \hline
        \multicolumn{4}{|c|}{$\bot$} & \multicolumn{4}{|c|}{$\top$}  \\ \hline
        \multicolumn{4}{|c|}{$\neg$/$\sim$\tnote{\textcolor{gray}{\ddag}}} & \multicolumn{4}{|c|}{$\neg ^{cv}$} \\ \hline
        $\otimes$ &  $\otimes ^{cv}$ & $\oplus $ & $\oplus ^{cv}$&
        $\rightarrow$ &  $\rightarrow ^{iv}$ & $\leftarrow $ & $\leftarrow ^{iv}$ \\ \hline
        $\sqcap$ &  $\sqcap ^{cv}$ & $\sqcup $ & $\sqcup ^{cv}$ &
        $\rightsquigarrow$ &  $\rightsquigarrow ^{iv}$ & $\leftsquigarrow $ & $\leftsquigarrow ^{iv}$ \\ \hline
    \end{tabular}
    \begin{tablenotes}
    \footnotesize
        \item[\textcolor{gray}{\ddag}]            \textcolor{gray}{two \textit{prima facie} identical rule pairs with different frameworks}
    \end{tablenotes}
    \end{threeparttable}
    \caption{results overview, $21$ minimally meaningful connectives}
\end{table}

Considering the vast and unstructured input set of merely formulable operational rules, with no reference to specific calculi, our results \textit{verify} I-bS as a formal theory of connective meanings: relative to our notion of definability and our minimal derivability relation, it yields primitive meanings for \textit{exactly} and \textit{only} these standard connectives. Moreover, the resulting picture leads to highly parsimonious semantic explanations for use-differences across logics: for instance, we can explain the differences in connective use between classical and (classical) linear logic in terms of \textit{conflating} semantically distinct connectives such as $\otimes$ and $\sqcap$. The same holds for differences between classical and intuitionistic logic: e.g. classical `and' combines the meanings of $\otimes$ and $\sqcap$, while intuitionistic `and' only has the latter meaning. Hence, the results of this exploratory study \textit{validate} I-bS by showing that it meets its main conceptual promise: logic-independent semantic explanations for the interrelations of the connectives \cite{nagler2026measuring}.

In \S 2, we justify our choice of I-bS by elaborating its advantages over the main alternatives in P-tS. In \S 3, we then specify our object of inquiry, i.e. all$^*$ connectives formulable in two-dimensional sequent calculi. In \S 4, we introduce the methodology of I-bS in general and our design choices in particular. In \S5, we generate the relevant data in the form of \cite{belnap1962tonk}-style definability proofs---or undefinability proofs---for each connective. In \S6, we analyse these data, taking note of the minimal inferential resources needed in each proof such that it succeeds. Finally, we present our results in \S7 and discuss their semantic significance in \S8.

\section{Motivation: a substructural approach to proof-theoretic semantics}

We first introduce standard approaches to P-tS in relation to \textit{model-theoretic semantics (M-tS)}, the dominant methodology in formal semantics. We then justify our alternative choice of I-bS by outlining its overarching ideas and relative benefits.

M-tS gives the truth values of the non-logical components of a fixed language relative to a \textit{model} $\mathfrak{M}$. This semantic base is then extended by specifying how the logical vocabulary contributes to the truth conditions of complex formulae. For example, via a \textit{satisfaction-in-$\mathfrak{M}$ relation} $\Vdash_\mathfrak{M}$, we give \textit{model-independent semantic clauses} for the connectives and recursively build up the semantic values of connective-containing expressions. This results in an \textit{entailment relation} $\models$, defined as \textit{truth-preservation over all models}, which we prove to be \textit{sound} and \textit{complete} relative to the derivability relation of a proof system $\vdash$ (see Figure 1).
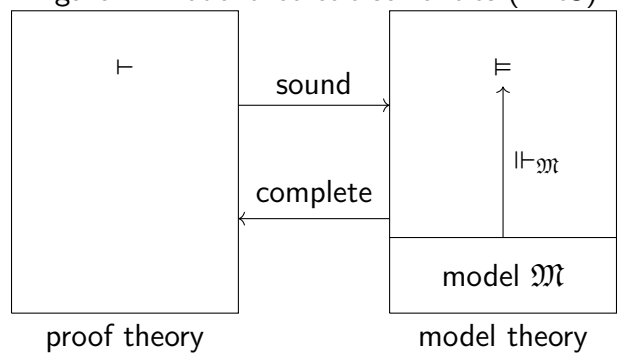
\begin{figure}[h]
    \centering
    \caption{model-theoretic semantics (M-tS)}
    \begin{tikzpicture}
        \draw (1.5,0) node[anchor=north]{proof theory};
        \draw (1.5,3) node[anchor=south]{$\vdash$};
        \draw (0,0) -- (3,0) -- (3,4) -- (0,4) -- (0,0);
        \draw[->] (3,2.75) -- (5,2.75);
        \draw (4,2.75) node[anchor=south]{sound};
        \draw[<-] (3,1.25) -- (5,1.25);
        \draw (4,1.25) node[anchor=south]{complete};
        \draw (6.5,0) node[anchor=north]{model theory};
        \draw (5,0) -- (8,0) -- (8,4) -- (5,4) -- (5,0);
        \draw[] (5,1) -- (8,1);
        \draw[->] (6.5,1) -- (6.5,3);
        \draw (6.5,2) node[anchor=west]{$\Vdash_\mathfrak{M}$};
        \draw (6.5,3) node[anchor=south]{$\models$};
        \draw (6.5,0.5) node{model $\mathfrak{M}$};
    \end{tikzpicture}
\end{figure}

This general schema is mirrored by the two most prominent current approaches to P-tS, \textit{proof-theoretic validity (P-tV)} \cite[e.g.][]{prawitz1971ideas, prawitz1973towards, prawitz1974idea, schroeder2006validity} and \textit{base-extension semantics (B-eS)} \cite[e.g.][]{sandqvist2015base, gheorghiu2024proof}.\footnote{We bracket more philosophically rooted bilateral approaches in this paper \cite[e.g.][]{incurvati2023reasoning, ayhan2023introductiona, ayhan2023introductionb}.} For instance, B-eS starts from a (semantic) \textit{base} $\mathfrak{B}$, i.e. a set of rules representing assertion conditions for the non-logical part of the vocabulary. This base defines a calculus for atomic sentences and the corresponding basic derivability relation $\vdash_{\mathfrak{B}}$. Using an \textit{assertability-in-$\mathfrak{B}$ relation} $\Vdash_\mathfrak{B}$, we extend $\vdash_{\mathfrak{B}}$ to the full language. Once again, this is built up recursively via \textit{base-independent semantic clauses}, one for each logical operator. Validity is defined as \textit{assertability-preservation across all bases}: we can read `$\Gamma \models \Delta$' as `given $\Gamma$, one can assert $\Delta$, in every base $\mathfrak{B}$' ($\Gamma\Vdash_\mathfrak{B} \Delta$) (see Figure 2).\footnote{The story for P-tV is somewhat more complex, yet some direct parallels remain, e.g. between models $\mathfrak{M}$/bases $\mathfrak{B}$ and atomic systems $\mathfrak{S}$. The following critiques of B-eS largely apply here as well.}
\begin{figure}
    \centering
    \caption{base-extension semantics (B-eS)}
    \begin{tikzpicture}
        \draw (1.5,0) node[anchor=north]{proof theory};
        \draw (1.5,3) node[anchor=south]{$\vdash$};
        \draw (0,0) -- (3,0) -- (3,4) -- (0,4) -- (0,0);
        \draw[->] (3,2.75) -- (5,2.75);
        \draw (4,2.75) node[anchor=south]{sound};
        \draw[<-] (3,1.25) -- (5,1.25);
        \draw (4,1.25) node[anchor=south]{complete};
        \draw (6.5,0) node[anchor=north]{proof theory};
        \draw (5,0) -- (8,0) -- (8,4) -- (5,4) -- (5,0);
        \draw[] (5,1) -- (8,1);
        \draw[->] (6.5,1) -- (6.5,3);
        \draw (6.5,2) node[anchor=west]{$\Vdash_{\mathfrak{B}}$};
        \draw (6.5,3) node[anchor=south]{$\models$};
        \draw (6.5,0.5) node{base $\mathfrak{B}$};
    \end{tikzpicture}
\end{figure}
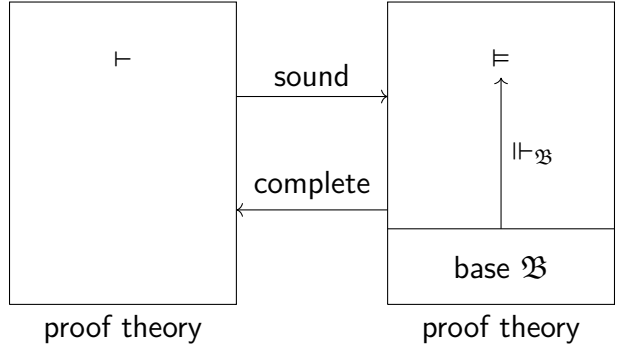

M-tS and B-eS share a form of \textit{globalism}: they generally do not allow us to compare the meanings of logical connectives \textit{across different logics}.
\begin{example}\

\noindent The B-eS clause for fusion $\otimes$  in intuitionistic (multiplicative linear) logic is
\[
\Vdash^{\Gamma_{At}}_\mathfrak{B} A\otimes B \text{ iff, for all }\mathfrak{C}\supseteq \mathfrak{B}, p\text{ atomic and }\Delta_{At}, \text{ if }A,B\Vdash^{\Delta_{At}}_\mathfrak{C}p\text{ then }\Vdash^{\Gamma_{At}, \Delta_{At}}_\mathfrak{C}p
\]
\cite{gheorghiu2024proof}, whereas its counterpart in classical (linear) logic  is:
\[
\Vdash^{\Gamma_{At}}_\mathfrak{B} A\otimes B \text{ iff, for all }\mathfrak{C}\supseteq \mathfrak{B}, \Delta_{At}, \text{ if }A,B\Vdash^{\Delta_{At}}_\mathfrak{C}\bot\text{ then }\Vdash^{\Gamma_{At}, \Delta_{At}}_\mathfrak{C}\bot
\]
\cite{barroso2025proof}, where $\Gamma_{At}, \Delta_{At}$ are fixed multisets of atoms.\footnote{These multisets help keep book of resource use in the linear setting. Their presence does not bear on our discussion.} Crucially, the latter takes a detour via $\bot$, which one must include in one's base. The former does not require this constraint, allowing for a more direct definition.
\end{example}

Clearly, these semantic clauses are different: they give incompatible use instructions for $\otimes$. Given compositionality and the recursive definition of $\Vdash_\mathfrak{M}/\Vdash_\mathfrak{B}$, the use of any complex expression containing the connective differs between the calculi. In Quine's words, if we \textit{change logic}, we also \textit{change meaning} \cite{quine1986philosophy}.

Consequently, we find ourselves struggling to give semantic explanations of the differences in connective use between the two calculi---beyond stating \textit{that} the meaning is different. We are unable to distinguish which use-phenomena result from semantic features of a connective itself, which derive from interactions between different connectives, and which derive from other structural features of the derivability relation.

This stands in tension with the inferentialist underpinnings of P-tS. In Priest's words:
\begin{quote}
    `If we give different truth conditions for the connectives, we are giving the formal connectives different meanings. When we apply the logics to vernacular reasoning we are, therefore, giving different theories of the meanings of the vernacular connectives.' \cite[204]{priest2005doubt}
\end{quote}
Originally levelled at Beall and Restall's supposedly-meaning-invariant logical pluralism \cite{beall2000logical, beall2001defending, beall2006logical}, Priest's critique is no less relevant here. If connectives have different proof-theoretic semantic clauses, they have different assertability conditions. Yet, P-tS assumes that formal proofs represent actual reasoning and that connective rules define connective use (see \S1). Hence, a connective's assertability conditions---the rule guiding its use---is an (idealised) model of how we use the corresponding vernacular connective in reasoning. If we were to use B-eS to model argumentative exchanges, e.g. between classical and intuitionistic (linear) reasoners, we would find that their logical disagreements would be reduced to different meanings of `and', `or', `if', etc. in their arguments. By providing different assertability conditions, B-eS renders logical disagreements primarily into communication breakdowns. Although some philosophers argue that meaningful disagreement can be possible even if word meanings differ \cite{chalmers2011verbal, hjortland2022disagreement}, this at least casts doubt on whether B-eS can be a suitable choice for the aim of this paper: modelling connective-meanings across logics.

In \cite{nagler2026measuring}, Nagler responds with an alternative P-tS framework, dubbed \textit{inference-behaviour semantics (I-bS)}, which avoids these concerns by giving \textit{calculus-independent} semantic clauses. Instead of presenting them in a proof-theoretic construction designed to parallel the proof system under investigation, she proposes to situate them \textit{within} the proof system. This takes the form of a \textit{substructural kernel} of a sequent calculus, which contains one set of substructural operational rules per semantically distinct connective in the system, its \textit{semantic clause}. To obtain this semantic clause, one must isolate the inferential resources needed to define the corresponding connective in a suitable connective-free and structurally minimal fragment of the proof system (more on these fragments in \S4.3). To re-emerge the full calculus, one enriches the kernel with the structural resources removed during this reduction. This process of reduction and emergence of structural properties serves as the core `sanity check' of I-bS in lieu of the traditional proof of soundness and completeness, yet playing an analogous role (Figure 3).

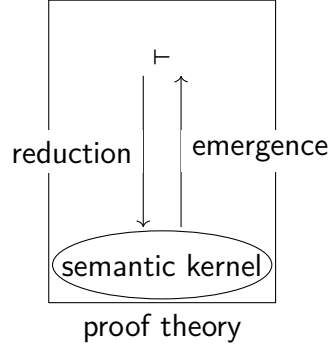
\begin{figure}
    \centering
    \caption{inference-behaviour semantics (I-bS)}
    \begin{tikzpicture}
        \draw (6.5,0) node[anchor=north]{proof theory};
        \draw (5,0) -- (8,0) -- (8,4) -- (5,4) -- (5,0);
        \draw[<-] (6.25,1) -- (6.25,3);
        \draw (6.25,2) node[anchor=east, fill=white]{reduction};
        \draw[->] (6.75,1) -- (6.75,3);
        \draw (6.75,2) node[anchor=west, fill=white]{emergence};
        \draw (6.5,3) node[anchor=south]{$\vdash$};
        \draw (6.5,0.5) ellipse (1.45cm and 0.45cm);
        \draw (6.5,0.5) node{semantic kernel};
    \end{tikzpicture}
\end{figure}

For instance, the semantic clause for $\otimes$ takes the form of the sequent rule pair:
\[
\Bigg \langle
\begin{prooftree}
    \hypo{A,B\vdash C}
    \infer1[L$\otimes$]{A\otimes B\vdash C}
\end{prooftree},
\begin{prooftree}
    \hypo{C\vdash A}
    \hypo{D\vdash B}
    \infer2[R$\otimes$]{C,D\vdash A\otimes B}
\end{prooftree}
\Bigg \rangle.
\]
We may read semantic clauses using a Sandqvist-style hypothetical reading of sequents \cite{sandqvist2015base}. For instance, one may read R$\otimes$ as: `If we can assert $A$, given $C$, and if we can assert $B$ given $D$, we may assert $A\otimes B$, given both $C$ and $D$.'

Crucially, this semantic clause gives the meaning of $\otimes$ across a wide range of logics, including classical (\textbf{LK}) and intuitionistic (\textbf{LJ}) logic as well as their linear logic counterparts. Each calculus emerges from the same semantic kernel of the $\otimes$ clause (together with clauses for the other connectives) by adding the relevant structural framework. For instance, to obtain \textbf{LJ}, we drop set-theoretic cardinality constraints on the left of `$\vdash$', and for \textbf{LK}, we additionally remove constraints on the right. In either case, we add structural rules for a fully Tarskian derivability relation (\textsc{Identity}, \textsc{Cut}, \textsc{Weakening}, \textsc{Contraction}), or just for a linear derivability relation if we want to reach the corresponding linear calculi (\textsc{Identity}, \textsc{Cut}). Technical narratives of this kind can already be found in \cite{gentzen1935auntersuchungen, gentzen1935buntersuchungen}. However, Nagler builds on \cite{restall2014pluralism, dicher2016proof} and the ability of the sequent calculus to make structural resources explicit to tell a new \textit{semantic} story \cite{nagler2026measuring}. Equating a connective with its definitional core---a move in Gentzenian spirit---she avoids the \textit{change of logic, change of meaning} problem. 

However, this is not the only benefit I-bS has over B-eS. First, B-eS appears to favour intuitionistic logic by design, posing a potential problem for logicians with pluralistic or non-constructive leanings. This is largely due to its reliance on natural deduction proof systems. For instance, the semantic clauses for the connectives in \cite{sandqvist2015base}---the basis of much current B-eS research---largely follow the corresponding standard natural deduction elimination clauses. As we would expect in the absence of modifications such as double negation elimination, the resulting validity relation is that of sentential intuitionistic logic. Classical modifications are more difficult to implement, requiring more resources or yielding less stable results \cite[e.g.][]{sandqvist2009classical, makinson2014inferential, barroso2025proof}.\footnote{There are recent advances towards addressing this issue, intriguingly shifting towards more sequent-oriented technology \cite{barroso2025proof, barroso2025sequent}.}

In contrast, I-bS does not display such a bias. Whether we take a classical, intuitionistic, or substructural starting point, the reduction follows the same procedure and yields results specific to the proof system under study. This shows the ecumenical potential of sequent calculus formalisations underlying I-bS and their aptitude for tracking structural resource expenditure. This makes I-bS a particularly suitable candidate for modelling logical disagreements and, relatedly, meaning-invariant logical pluralisms, in line with its predecessors \cite{dicher2016proof, ferrari2021proof}. At the same time, left- and right-rules are treated on par and considered in tandem. As a result, debates about the priority of introduction or elimination rules are largely superfluous in this context, as are qualms about the `hypothetical character' of disjunction rules \cite[cf.][]{sandqvist2015base}.

Finally, I-bS is notably more ontologically frugal than B-eS, P-tV or M-tS, giving semantics with resources present in the (sub-)structure of the calculus under investigation rather than requiring any model- or proof-theoretic parallel constructions.

These merits provide sufficient motivation to put I-bS to the test. Its core strength lies in its modular, logic-independent semantic clauses. What happens, then, if we remove entirely the ties between a connective and its calculus? For the remainder of this paper, we give I-bS-style semantic clauses for \textit{any} connective whose rules can be formulated in standard two-dimensional sequent calculi (using $\leq 2$ premiss sequents, $\leq 2$ active formulae).

\section{Object of inquiry: all$^*$ connectives formulable in two-dimensional sequent calculi}

\subsection{All$^*$ connectives}

In this section, we define the object of investigation: the logical connectives. Specifically, we study all connectives whose operational rules are \textit{two-dimensional} and contain at most \textit{two premiss sequents} and \textit{two active formulae}. We consider rules \textit{operational} iff they are \textit{separate}, \textit{explicit}, and \textit{symmetrical}.

\begin{definition}[two-dimensional sequent]\

\noindent
\textit{Two-dimensional sequents} are of the shape $\Phi \vdash \Psi$, where $\Phi, \Psi$ are finite multisets of formulae of a sentential language $\mathcal{L}$ and `$\vdash$' denotes a binary inner derivability relation between them, also known as the `sequent symbol'. We call $\Phi$ the component in the first dimension (or: \textit{first component}) and $\Psi$ the component in the second dimension (\textit{second component}) of the sequent.
\end{definition}
We can read `$\Phi \vdash \Psi$' as `From all formulae in $\Phi$, derive at least one formula in $\Psi$.'
\begin{definition}[sequent rule, calculus]\

\noindent
A \textit{sequent rule} expresses a relation between a multiset of sequents and a sequent. We will write sequent rules as 
\begin{prooftree}
    \hypo{\mathfrak{s}_1}
    \hypo{\mathfrak{s}_2}
    \hypo{\dots}
    \hypo{\mathfrak{s}_n}
        \infer4{\mathfrak{s}_i}
\end{prooftree}
, where $\mathfrak{s}_1, \mathfrak{s}_2, \dots, \mathfrak{s}_n, \mathfrak{s}_i$ are sequents, the sequents above the inference line form a multiset, and the inference line expresses the external derivability relation. We call $\mathfrak{s}_1, \mathfrak{s}_2, \dots, \mathfrak{s}_n$ \textit{premiss sequents} and $\mathfrak{s}_i$ the \textit{conclusion sequent}. 

A \textit{(sequent) calculus} is a set of sequent rules.
\end{definition}

Thus, a connective is \textit{two-dimensional} iff its operational rules are defined on two-dimensional sequents. While this leaves higher-dimensional connectives \cite[e.g.][]{baaz1993elimination} to future work, the majority of sequent calculi are two-dimensional, hence this restriction does not unduly limit the significance of the investigation.\footnote{Note that this requirement also excludes most model-theoretically (`semantically') polluted calculi such as labelled sequents, see for instance \cite{martinot2025formalising}.} This restriction is primarily a matter of containing the scope of the study.

For the same reason, we restrict ourselves to operational rules with at most \textit{two premiss sequents} and at most \textit{two active formulae}. Again, most connectives in the literature, including the Boolean connectives, satisfy this restriction. However, unlike two-dimensionality, these constraints can be relaxed without difficulty and are even redundant in the case of repeated premiss sequents as we will see in part in \S 8.

What exactly are operational rules? Following \cite{gentzen1935auntersuchungen}, we can think of them as the definition of a connective.

\begin{example}[$\neg_{\textbf{LK}/\textbf{LJ}}$]\

\noindent
For example, consider Gentzen's original sequent calculi for classical logic, \textbf{LK}, and for intuitionistic logic, \textbf{LJ}, \cite{gentzen1935auntersuchungen, gentzen1935buntersuchungen}. The following rule pairs define the unary connective $\neg$:
\begin{itemize}
    \item in \textbf{LK}: $\Biggl\{$ \begin{prooftree}
\hypo{\Gamma &\vdash \Delta, A}
\infer1[L$\neg_{\textbf{LK}}$]{\neg A, \Gamma  &\vdash \Delta}
\end{prooftree}, 
\begin{prooftree}
\hypo{A, \Gamma &\vdash \Delta}
\infer1[R$\neg_{\textbf{LK}}$]{\Gamma &\vdash \Delta, \neg A}
\end{prooftree} $\Biggl\}$;
\item in \textbf{LJ}: $\Biggl\{$ \begin{prooftree}
\hypo{\Gamma &\vdash A}
\infer1[L$\neg_{\textbf{LJ}}$]{\neg A, \Gamma  &\vdash}
\end{prooftree} ,
\begin{prooftree}
\hypo{A, \Gamma &\vdash }
\infer1[R$\neg_{\textbf{LJ}}$]{\Gamma &\vdash \neg A}
\end{prooftree} $\Biggl\}$.
\end{itemize}
\end{example}
We use Roman capitals for arbitrary formulae, and Greek capitals for finite multisets of formulae of $\mathcal{L}$. We call a formula \textit{active} iff it only occurs above the inference line, \textit{principal} iff it only occurs below the inference line, and \textit{parametric} iff it occurs above as well as below the inference line. Our rule pairs specify the conditions for introducing the principal formula $\neg A$, given in terms of the active formula $A$ and the parametric formulae $\Gamma$ and $\Delta$. 
If we delete all active and principal formulae of an operational rule, we obtain its \textit{context} (sequent), e.g. $\Gamma \vdash \Delta$ for $\neg_{\textbf{LK}}$ and $\Gamma \vdash$ for $\neg_{\textbf{LJ}}$. As a general definition:

\begin{definition}[operational rules]\

    \noindent A set of rules \textbf{R} is \textit{operational} for the connective $\#$ (or simply: the `$\#$-rules') iff each $\mathfrak{r}\in\textbf{R}$ is
    \begin{enumerate}
        \item \textit{separate}, i.e. no connective other than $\#$ occurs in $\mathfrak{r}$ \cite{zucker1978adequacy};\footnote{Zucker and Tragesser attribute the term to Feferman \cite{zucker1978adequacy}.}
        \item \textit{explicit}, i.e. $\#$ only occurs once in $\mathfrak{r}$, namely in the principal formula \cite{zucker1978adequacy, prawitz1979proofs}. If $\#$ occurs in the:
        \begin{enumerate}
            \item first component, $\mathfrak{r}$ is a L$\#$-rule; 
            \item second component, $\mathfrak{r}$ is a R$\#$-rule; 
        \end{enumerate}
        \item \textit{symmetrical}, i.e. $\mathfrak{r}$ is either a L$\#$- or a R$\#$-rule, and there is at least one L$\#$- and at least one R$\#$-rule in \textbf{R} \cite{wansing2000idea}.
    \end{enumerate}
\end{definition}
Indeed, we will give \textit{exactly one} L$\#$- and \textit{exactly one} R$\#$-rule. However, this is a notational matter, depending on one's reading of sequent rules (cf. Definition 2). One can understand a rule either as expressing
\begin{itemize}
    \item `If you have a derivation for \textit{each} sequent in $\{\mathfrak{s_1}, \mathfrak{s_2}, \dots \mathfrak{s_n}\}$, derive $\mathfrak{s_i}$', or as
    \item `If you have a derivation for \textit{at least one} sequent in $\{\mathfrak{s_1}, \mathfrak{s_2}, \dots \mathfrak{s_n}\}$, derive $\mathfrak{s_i}$'.
\end{itemize} 
We write `$s_1 \_ s_2 \_ \dots \_ s_n $' to indicate the former reading, and `$s_1 + s_2 + \dots + s_n $' to indicate the latter reading.

Note that \textit{explicitness} excludes several connective rules with absorbed structural properties such as $L\supset$ in Negri and van Plato's \textbf{G3ip} \cite{negri2001structural} while others are unaffected such as the \textsc{Weakening}-absorbent version of our L$\neg_{\textbf{LK}}$,
\[
\begin{prooftree}
    \hypo{A, \Gamma \vdash\Delta}
        \infer1[L${\neg_{\textbf{LK}}}^{\prime}$]{\Sigma, \Gamma \vdash \neg A, \Delta, \Pi}
\end{prooftree}.
\]
Rules of the latter type will not affect the results of this investigation as they will turn out to be structurally more-than-minimal. We will come back to this in \S7. If a connective can be explicitly defined in a sequent calculus, then it also has rules satisfying these three restrictions \cite{wansing2000idea}.

We can now make explicit what is meant by the asterisk `$^\ast$' in the title of this paper: `all$^\ast$ connectives' are defined by a rule pair that is \textit{operational} (in the sense of \textit{separation}, \textit{explicitness} and \textit{symmetry}) and has \textit{$\leq 2$ premiss sequents} and \textit{$\leq 2$ active formulae}.

\subsection{Notation}

What are all$^{\ast}$ connectives in two-dimensional sequent calculi, and how many are there? To answer these questions, we must look at the syntactically unique rule pairs we can formulate given our restrictions, in terms of their \textit{framework}, \textit{premiss shape} and \textit{arity}.

\begin{definition}[sequent framework]\

\begin{itemize}
\item A sequent $\Phi \vdash \Psi$ is of the \textit{framework} $l:r$, iff $|\Phi|=l$ and $|\Psi|=r$, where $l,r\in \mathbb{N}^0$. 
\item A sequent rule/calculus is of the \textit{framework} $ L: R$, iff for all sequents $\Phi \vdash \Psi$ derivable from the rule/calculus, $|\Phi|\in L$ and $|\Psi|\in R$, where $L,R\subseteq \mathbb{N}^0$.\footnote{This notion of frameworks is inspired by \cite{humberstone2011connectives}.}
\end{itemize}
\end{definition}

For instance in Example 2, $\neg_{\textbf{LK}}$ has the framework $\mathbb{N}^0:\mathbb{N}^0$, allowing for finite multisets of arbitrary cardinality on each side of $\vdash$. In contrast, $\neg_{\textbf{LJ}}$ restricts the second component of sequents to at most one formula occurrence, resulting in the framework $\mathbb{N}^0:\{0,1\}$.

To account for unbounded sets, we use standard ordinal arithmetic in the following, employing Hessenberg (natural) sums $\oplus$ and identifying the cardinality of set $A$ with its initial ordinal, denoted as $|A|$.
\begin{definition}[framework size]\

\noindent 
The \textit{size} of the framework of a sequent rule/calculus $L:R$, $size(L:R)$, is the ordinal pair
$\big\langle sup(L)\oplus sup(R),|L|\oplus |R|\big\rangle$, ordered lexicographically, such that for any two frameworks $L_1:R_1$ and $L_2:R_2$: $size(L_1:R_1) < size(L_2:R_2)$ iff
\begin{enumerate}
    \item $sup(L_1)\oplus sup(R_1)<sup(L_2)\oplus sup(R_2)$, or
    \item $sup(L_1)\oplus sup(R_1)=sup(L_2)\oplus sup(R_2)$ and $|L_1|\oplus |R_1| < |L_2|\oplus |R_2|$.
\end{enumerate}
\end{definition}
For instance, the framework of classical \textbf{LK}, $\mathbb{N}^0:\mathbb{N}^0$, is larger than that of intuitionistic \textbf{LJ}, $\mathbb{N}^0:\{0,1\}$ (since $2\omega>\omega+1$). The latter is of equal size to that of \cite{urbas1996dual}'s dual-intuitionistic \textbf{LDJ}, $\{0,1\}:\mathbb{N}^0$, but larger than that of \cite{johansson1937minimalkalkul}'s minimal \textbf{LM}, $\mathbb{N}^0:\{1\}$ (as $\omega+1=\omega+1$ but $\omega+1<\omega+2$). The latter, in turn, is larger than that of \cite{restall2005geometry}'s lattice logic \textbf{LaL}, $\{0,1\}:\{0,1\}$ (since $\omega+1>2$). This notion of framework size will be especially helpful in \S 6 when we analyse our results.

For the remainder of this paper, we will use a shorthand notation for sequents, inspired by \cite{baaz1993elimination}. It allows us to quantify into component position and thereby to consider several rule configurations at once. This enables a more compact presentation of formulable rule combinations and helps limit the length of proofs in this study.

\begin{definition}[shorthand]\

\noindent
\begin{enumerate}[label=(\roman*)]
\item $[1:\Phi; 2:\Psi]$ is a sequent whose first component is $\Phi$ and whose second component is $\Psi$. If a component is unspecified, it defaults to $\emptyset$.
\item We write `$s$', `$t$' for sequents whose components are any finite multisets of $\mathcal{L}$-formulae. We will use this to express
contexts.
\item We write `$s \times t$' or, simply, `$st$' for a sequent whose $i$th component is the union of the $i$th component of $s$ and the $i$th component of $t$, for all $i \in \{1,2\}$. We will drop `$\times$' in most cases.
\end{enumerate}
\end{definition}
\begin{table}[htbp]
\centering
\caption{example transcriptions}
\begin{threeparttable}
\begin{tabular}{|l|l|l|}
\hline
   shorthand  & traditional notation & \color{gray} transcription note \\ \hline \hline
   $[1: \Gamma; 2:A, B]$ & $\Gamma \vdash A, B$ & \color{gray} immediate\\ \hline
   $[2:A]$ & $\emptyset \vdash A$ & \color{gray} $[1:\emptyset; 2:A]$ \\ \hline
   $s$ & $\Gamma \vdash \Delta$ & \color{gray} $[1: \Gamma; 2: \Delta]^\dag$ \\ \hline
   $[1:A]st$ & $\{A, \Gamma, \Sigma \vdash \Delta, \Pi\}$ & \color{gray} $[1:A]\times s\times t$\\  & &   \color{gray} $[1:A]\times[1:\Gamma; 2:\Delta]\times [1:\Sigma; 2:\Pi]$\\ & & \color{gray} $[1:A\cup \Gamma\cup \Sigma; 2: \Delta\cup \Pi]$\\ & & \color{gray} $[1:A, \Gamma, \Sigma; 2: \Delta, \Pi]$ \\ \hline
   $[i:A; j:B]s$ & $A,B, \Gamma \vdash \Delta$ or & \color{gray} if $i=j=1$\\
   & $A, \Gamma \vdash \Delta, B$ or & \color{gray} if $i=1$ and $j=2$\\
   & $B, \Gamma \vdash \Delta, A$ or & \color{gray} if $i=2$ and $j=1$\\
   & $ \Gamma \vdash \Delta, A, B$ & \color{gray} if $i=j=2$\\\hline
\end{tabular}
\begin{tablenotes}
    \footnotesize
    \item Note: $\Gamma, \Delta, \Sigma, \Pi$ are arbitrary and $i,j\in \{1,2\}$.
\end{tablenotes}
\end{threeparttable}
\end{table}
We can now re-state the definition of $\neg$ in \textbf{LK/LJ} using our shorthand:
\begin{example}[$\neg_{\textbf{LK}/\textbf{LJ}}$, again]\

\noindent
\[
\begin{prooftree}
\hypo{[i:A]s}
\infer1[j$\neg$]{[j:\neg A]s}
\end{prooftree},
\]
where $i\neq j$ and $i,j\in \{1,2\}$.
\end{example}
The benefit of this notation is clear: a single rule gives the operational rule for negation across several calculi, including \textbf{LK} and \textbf{LJ}. The difference between them results from the framework of context $s$: e.g. in \textbf{LK}, the framework of $s$ is unconstrained; in \textbf{LJ}, $s$ must be of $n:m$, for $n\in \mathbb{N}^0$ and $m\in\{0, 1\}$.

\subsection{Formulability}

Using this shorthand, we now identify which operational rule pairs can be formulated in a syntactically unique fashion, beginning with their premiss sequents.
\begin{definition}[premiss type $\mathtt{SP}$]\

\noindent
An \textit{operational rule} $k\#$, with $k\in \{1,2\}$, is of \textit{premiss type} $\mathtt{SP}(k\#)$ iff its premiss sequents are specified by the ordered pair $\langle \mathtt{S}, \mathtt{P} \rangle$, where 
\begin{enumerate}
\item $\mathtt{S}$ is the \textit{shape} of the premises, and
\item $\mathtt{P}$ is the \textit{position} of its active formulae.
\end{enumerate}
A \textit{connective} $\#$  is of premiss type $\mathtt{SP}(\#)$ iff  $\mathtt{SP}(\#)=\langle \mathtt{SP}(1\#), \mathtt{SP}(2\#)\rangle$.
\end{definition}
\begin{definition}[premiss shape $\mathtt{S}$]\

\noindent Let $k\#$ be an operational rule and $\mathfrak{S}$ the multiset of its premiss sequents. \\
Then, $\mathtt{S}(k\#)=$
\begin{itemize}
    \item $\alpha$ iff $\mathfrak{S}=\emptyset$;
    \item $\beta$ iff $\mathfrak{S}=\{\emptyset + s_1\}$;
    \item $\gamma$ iff $\mathfrak{S}=\{s_1\}$;
    \item $\delta cs$ iff $\mathfrak{S}=\{s_1 \_ s_2\}$ and $s_1, s_2$ share contexts;
    \item $\delta cd$ iff $\mathfrak{S}=\{s_1 \_ s_2\}$ and $s_1, s_2$ have different contexts;
    \item $\epsilon$ iff $\mathfrak{S}=\{s_1+ s_2\}$.
\end{itemize}
\end{definition}
Rules of shape $\alpha$ are \textit{axiomatic}, whereas those of shape $\beta$ are \textit{potentially axiomatic}. The latter can be understood as split rules:
\[
\begin{prooftree}
    \hypo{\emptyset}
        \infer1[k$\#$1]{\mathfrak{s_i}}
\end{prooftree}\, ,
\qquad\qquad 
\begin{prooftree}
    \hypo{\mathfrak{s_1}}
        \infer1[k$\#$2]{\mathfrak{s_i}}
\end{prooftree}\, ,
\]
where $\mathfrak{s_1}$ only contains active formulae.

The same holds for rules of shape $\epsilon$: deriving one of the two premiss sequents is sufficient to derive the conclusion. We call such rules \textit{additive}.

In contrast, rules of shape $\gamma$ are \textit{multiplicative}. They can be understood as products under sequent multiplication $\times$. For instance, the premiss sequent $[1:A,B]s\  =\ [1:A] \times [1:B] \times s$.

In the case of $\delta$, derivations of both premiss sequents are needed to derive the conclusion sequent $\mathfrak{s_i}$; this is expressed by a \textit{single-lined} rule schema: \[\begin{prooftree} \hypo{\mathfrak{s_1}} \hypo{\mathfrak{s_2}} \infer2[k$\#$]{\mathfrak{s_i}} \end{prooftree}\, .
\] 

The context $c$ of the conclusion sequent is the set of contexts of its premiss sequents. Thus, in rules of shape $\delta cs$, $c=s$ (both premises have context $s$); in rules of shape $\delta cd$, $c=s \times t$ (one premiss has context $s$, the other $t$). Note that \textit{context-differentiation} and \textit{multiplicativity}  are frequently conflated, as are \textit{context-sharing} and \textit{additivity}. However, we keep all four properties separate here as rule pairs crossing those categories are just as formulable.

\begin{definition}[active formula position $\mathtt{P}$]\

\noindent Let $i,j\in \{1,2\}$. In the operational rule $k\#$, the active formula position $\mathtt{P}(k\#)=$
\begin{itemize}
    \item $0$ iff there are no active formulae;
    \item $iF$ iff one active formula occurs in the $i^{th}$ component of a premiss sequent;
    \item $iFjG$ iff one active formula occurs in the $i^{th}$ and one in the $j^{th}$ component of the \textit{same} premiss sequent;
    \item $iF-jG$ iff one active formula occurs in the $i^{th}$ and one in the $j^{th}$ component of \textit{different} premiss sequents.
\end{itemize}
\end{definition}
\begin{example}\

    \begin{itemize}
        \item $\mathtt{SP(\neg_\textbf{LK})}=\mathtt{SP(\neg_\textbf{LJ})}=\langle \gamma 2A, \gamma 1A \rangle$;
        \item $\mathtt{SP}(\textsc{tonk})=\langle \gamma 1B,\gamma 2A \rangle$, for \\ $\Bigg\langle 
        \begin{prooftree}
            \hypo{[1:B]s}
                \infer1[1\textsc{tonk}]{[1:A\textsc{tonk}B]s}
        \end{prooftree}, \
        \begin{prooftree}
            \hypo{[2:B]s}
                \infer1[2\textsc{tonk}]{[2:A\textsc{tonk}B]s}
        \end{prooftree}
        \Bigg \rangle$;
        \item $\mathtt{SP}(\rightarrow)=\langle \delta cd 1B-2A, \gamma 1A2B\rangle$ for\\ $\Bigg\langle 
        \begin{prooftree}
            \hypo{[1:B]s}
            \hypo{[2:A]t}
                \infer2[1$\rightarrow$]{[1:A\rightarrow B]st}
        \end{prooftree}, \
        \begin{prooftree}
            \hypo{[1:A; 2:B]s}
                \infer1[2$\rightarrow$]{[2:A\rightarrow B]s}
        \end{prooftree}
        \Bigg \rangle$;
        \item $\mathtt{SP}(\sqcup)=\langle \delta cs 1B-1A, \epsilon 1A2B \rangle $ for \\ $\Bigg\langle 
        \begin{prooftree}
            \hypo{[1:B]s}
            \hypo{[1:A]s}
                \infer2[1$\sqcup$]{[1:A\sqcup B]s}
        \end{prooftree}, \
        \begin{prooftree}
            \hypo{[1:A]s\ +\ [2:B]s}
                \infer1[2$\sqcup$]{[2:A\sqcup B]s}
        \end{prooftree}
        \Bigg \rangle$;
        \item $\mathtt{SP}(\star)=\langle \beta 2A, \delta cd 2A2B \rangle $ for \\ $\Bigg\langle 
        \begin{prooftree}
            \hypo{[2:A]\ + \ \emptyset}
                \infer1[1$\star$]{[1:\star]}
        \end{prooftree}, \
        \begin{prooftree}
            \hypo{[2:A, B]s}
            \hypo{t}
                \infer2[2$\star$]{[2:\star]st}
        \end{prooftree}
        \Bigg \rangle$
    \end{itemize}
\end{example}
Using this notion of premiss type, we can give a general definition of a connective.
\begin{definition}[connective $\#$]\

\noindent A connective $\#$ is an ordered triple $\langle \mathtt{SP}, n, L:R \rangle$, where `$n\in \mathbb{N}^0$' denotes the \textit{arity} of $\#$ in the principal formula, and `$L:R$' its \textit{framework}.
\end{definition}
\begin{example}[$\neg_{\textbf{LK}/\textbf{LJ}}$, once more]\

\noindent
$\neg_{\textbf{LK}}=\langle \gamma 2A, \gamma 1A, 1, \mathbb{N}^0:\mathbb{N}^0 \rangle$; $\neg_{\textbf{LJ}}=\langle \gamma 2A, \gamma 1A, 1, \mathbb{N}^0:0 \rangle$.
\end{example}

Table 3 gives all premiss types which can be formulated in a syntactically unique way, given our constraints. Columns indicate shapes and rows indicate positions.
\begin{sidewaysfigure}
\begin{threeparttable}[h]
\begin{longtable}{c!{\vrule width 1pt}c|c|c|c|c|c}
\caption{premiss types of two-premiss operational rules, for $i,j\in \{1,2\}$ and $F, G\in \{A,B\}$}
\endfirsthead
\endhead
\multirow{2}*{\diagbox{\texttt{P}}{\texttt{S}}} &$\alpha$ & $\beta $ & $\gamma $ & \multicolumn{2}{c|}{$\delta$} & $\epsilon$\\
\cline{5-6}
&&&&$cs$&$cd$&\\
\hline\hline
$0$&$\emptyset$&\cellcolor{lightgray} (inexpressible)&$s$&$s\ \_ \ s$&$s\ \_ \ t$ & \cellcolor{lightgray} (cf. $\gamma 0$)\\
\hline
$iF$&\cellcolor{lightgray} (cf. $\alpha 0$)&$[i:F]\ +\ \emptyset$&$[i:F]s$&$[i:F]s\ \_ \ s$&$[i:F]s\ \_ \ t$ & $[i:F]s\ + \ s$
\\
\hline
$iFjG$&\cellcolor{lightgray} (cf. $\alpha 0$)& $[i:F;j:G]\ + \ \emptyset$ & $[i:F;j:G]s$& $[i:F; j:G]s\ \_ \ s$ & $[i:F; j:G]s\ \_ \ t$ & $[i:F; j:G]s\ + \ s$ \\ \hline
$iF-jG$& \cellcolor{lightgray} (inexpressible) & \cellcolor{lightgray} (inexpressible) & \cellcolor{lightgray} (inexpressible) & $[i:F]s\ \_ \ [j:G]s$ & $[i:F]s\ \_ \ [j:G]t$ & $[i:F]s\ + \ [j:G]s$  
\\\end{longtable}
\end{threeparttable}
\end{sidewaysfigure} 

Why is Table 3 exhaustive? For instance, one might also expect a `con\-text-differenti\-at\-ing' version of $\epsilon$. However, this is not formulable. Since only one of the premiss sequents is active, the parametric formulae of the other premiss sequent would automatically become principal. For example:
\[
\begin{prooftree}
    \hypo{[1:A]s}
    \infer1[1\#1]{[1:\#A]s\boxed{t}}
\end{prooftree}
\qquad
\begin{prooftree}
    \hypo{[1:A]t}
    \infer1[1\#2]{[1:\#A]\boxed{s}t}
\end{prooftree}\, .
\]
Similarly, we cannot formulate a single-lined version of $\beta$. The premiss multiset $\emptyset\ \_\ \mathfrak{s}$ is indistinguishable from $\mathfrak{s}$ given that $\emptyset$ holds vacuously, irrespective of $\mathfrak{s}$. For the same reason, we do not list $\epsilon 0$ as a separate premiss type as it is the same as $\alpha$. Naturally, no premiss shape with fewer than two premiss sequents can take the position $iF-jG$ as its definition requires at least two premiss sequents. 

Using Table 3, we can formulate $$\underbrace{\overbrace{4}^{\texttt{S}=\alpha, \dots, \delta cd}}_{\mathtt{P}=0}+\underbrace{\overbrace{5}^{\beta, \dots, \epsilon}\cdot \overbrace{2}^{i}\cdot \overbrace{2}^{F}}_{iF}+(\underbrace{\overbrace{5}^{\beta, \dots, \epsilon}\cdot \overbrace{3}^{i, j}\cdot \overbrace{2}^{F}}_{iFjF}+\underbrace{\overbrace{5}^{\beta, \dots, \epsilon}\cdot \overbrace{4}^{i, j}}_{iAjB})+(\underbrace{\overbrace{3}^{\delta cs, \dots, \epsilon}\cdot \overbrace{3}^{i, j}\cdot \overbrace{2}^{F}}_{iF-jF}+\underbrace{\overbrace{3}^{\delta cs, \dots, \epsilon}\cdot \overbrace{4}^{i, j}}_{iA-jB})=104$$ syntactically unique premiss types given our discussed constraints. Each summand (underbrace) of the calculation represents the number of unique rules we can formulate relative to a possible active formula position. This is straightforward combinatorics, although we need to show some care in the case of two-premiss rules using the same active formula twice. To avoid double-counting identical cases, we considered this separately from the case with different active formulae. As we define connectives on rule \textit{pairs}, we must consider $104^2=10,816$ unique cases (for each framework and arity) in this study.

Having specified our object of study, the $10,816$ operational rule pairs formulable in two-dimensional sequent calculi ($\leq2$ premiss sequents; $\leq 2$ active formulae), we now turn to measuring their meaning.

\section{Methodology: measuring connective meaning}

\subsection{Measurement device: inference behaviour}

With candidate connectives delineated, we now specify how to give I-bS clauses for them, largely following \cite{nagler2026measuring}.

At the core of I-bS lies the \textit{minimal inference behaviour} of a connective. This notion serves as an identity criterion for connective use. If and only if two connectives share minimal inference behaviour, they share inferential use and, given inferentialism, meaning. Let us start with inference behaviour \textit{simpliciter}.

Inference behaviour is defined relative to a proof \cite[e.g.][]{restall2023logical}:
\begin{definition}[derivation/proof tree]\

    \noindent
    A \textit{derivation} (tree) of a conclusion sequent $\mathfrak{s}$ from the set of premiss sequents $\mathfrak{S}$ in a calculus \textbf{Cal} is a non-downward branching partial order on a finite set of sequents whose root is $\mathfrak{s}$, each of whose leaves is $\emptyset$ or some $\mathfrak{s}^\prime \in \mathfrak{S}$, and each step is an application of some \textbf{Cal}-rule.
    
    If $\mathfrak{S}=\emptyset$, we call the derivation a \textit{proof} (tree).
\end{definition}
Intuitively, the \textit{inference behaviour} of a given connective \# tracks (1) \textit{how} we use \# (2) \textit{whenever} we use \# in a given proof. Figure 4 illustrates this idea.

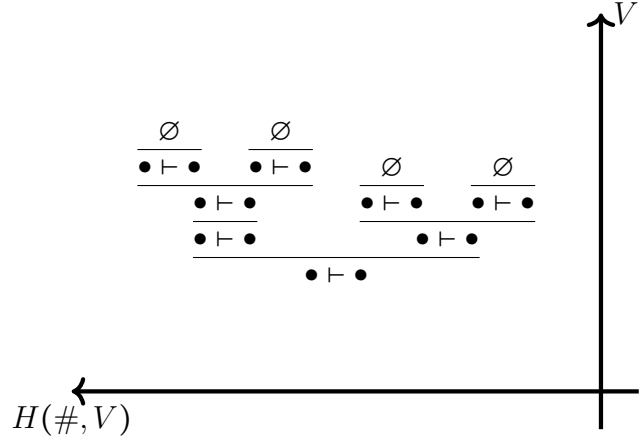
\begin{figure}[h]
    \centering
    \caption{schematic illustration of inference behaviour $\texttt{ib}$}
    \begin{tikzpicture}
        \draw[->,ultra thick] (7.5,0)--(0,0) node[below]{$H(\#,V)$};
        \draw (3.5,2.5) node{\begin{prooftree}
            \hypo{\emptyset}
                \infer1{\bullet \vdash \bullet}
            \hypo{\emptyset}
                \infer1{\bullet \vdash \bullet}
                \infer2{\bullet \vdash \bullet}
                \infer1{\bullet \vdash \bullet}
            \hypo{\emptyset}
                \infer1{\bullet \vdash \bullet}
            \hypo{\emptyset}
                \infer1{\bullet \vdash \bullet}
                \infer2{\bullet \vdash \bullet}
                \infer2{\bullet \vdash \bullet}
            \end{prooftree}};
        \draw[->,ultra thick] (7,-0.5)--(7,5) node[right]{$V$}; 
    \end{tikzpicture}
\end{figure}

(2) is recorded in terms of the current proof step, called a \textit{vertical location} $V$ \cite{nagler2026measuring}. As we are operating in sequent calculi, we give a proof in the form of a proof tree. Thus, vertical locations represent tree nodes (`vertical $x$-axis' in Figure 4).

(1) tracks whether \# is used as a premiss or a conclusion in each proof step. Following \cite{nagler2026measuring}, we call this the \textit{horizontal location(s)} $H$ of \# in $V$. The function $H(\#, V)$ tells us which component \#-formulae and their immediate precursor formulae occupy in $V$ (`horizontal $y$-axis' in Figure 4). Let us capture this idea formally.
\begin{definition}[vertical and horizontal location $V$/$H$]\

    \noindent
    Let $\mathfrak{p}$ be a list of nodes forming a proof tree.\footnote{We call $\mathfrak{p}$ a `proof tree' although latter technically refers to a more fine-grained structure (cf. Definition 11). However, this level of description suffices for our present purposes.}\\
    A \textit{horizontal location} $H$ is a function, $H: \langle \#,V \rangle \mapsto d$, such that:
    \begin{enumerate}
        \item \# is a fixed connective,
        \item $V\in \mathfrak{p}$ is an occupied node in a proof tree (\textit{vertical location}), and
        \item $d$ is a multiset over $\{1,2\}$. Each element of $d$ corresponds to the component in which a \#-formula or an active formula of a \#-rule application occurs, for any such occurrence in the sequent at $V$. 
    \end{enumerate} 
\end{definition}
\begin{definition}[inference behaviour $\texttt{ib}$]\

    \noindent 
        The \textit{inference behaviour} $\texttt{ib}^\#_\mathfrak{p}$ of connective $\#$ in a proof (tree) $\mathfrak{p}$ is the list of function values of $H(\#,V)$, for each $V\in \mathfrak{p}$. We call the set of all elements of $\texttt{ib}^\#_\mathfrak{p}$ the `inference behaviour profile' $\texttt{ibp}^\#_\mathfrak{p}$. 
\end{definition}
We use `$\{\dots\}$' to denote sets, `$[\dots]$' for multisets, and `$\langle \dots \rangle$' for tuples/lists.
\begin{example}\

\noindent
Consider the proof tree
\[
    \mathfrak{p} =
    \begin{prooftree}
        \hypo{\emptyset^0}
        \infer1[\textsc{Id}]{[1:A;2:A]^1}
        \hypo{\emptyset^0}
        \infer1[\textsc{Id}]{[1:B;2:B]^2}\
        \infer2[2$\wedge$]{[1:A,B; 2: A\wedge B]^3}
        \infer1[2$\neg$]{[1:A; 2: \neg B, A\wedge B]^4}
        \infer1[2$\neg$]{[2: \neg A, \neg B, A\wedge B]^5}
        \infer1[1$\neg$]{[1:\neg(A\wedge B); 2: \neg A, \neg B]^6}
        \infer1[2$\vee$]{[1:\neg(A\wedge B); 2: \neg A\vee \neg B]^7}
    \end{prooftree}\, .
\]
We label each node $V \in \mathfrak{p}$ with a distinct Arabic numeral, displayed as a superscript. For example, $V=1$ in the top-left-most occupied node; $V=2$ in the top-right-most occupied node; etc. We tag empty nodes (\textit{null-nodes}) with `$V=0$' to indicate that $H$ is undefined for them, i.e. $H(\#, 0)=\texttt{null}$, for any connective \#. We will ignore null-nodes henceforth.

First, consider conjunction. Although `$\wedge$' does not occur in the first vertical location, we use the $A$ in the second component as an active formula when we apply $2\wedge$ in the following step. Hence, we track $H(\wedge,1)=[2]$. The same reasoning holds for $B$ in the second component at $V=2$, yielding $H(\wedge,2)=[2]$. At $V=3$, $A\wedge B$ occurs in the second component and, thus, $H(\wedge,3)=[2]$. Continuing this reasoning, $H(\wedge,4)=[2]$ and $H(\wedge,5)=[2]$. Now, although `$\wedge$' also occurs in `$\neg (A\wedge B)$', it is not the main connective. Accordingly, $\neg (A\wedge B)$ is not a $\wedge$-formula but a $\neg$-formula. Therefore, $H(\wedge,6)=\emptyset$ and $H(\wedge,7)=\emptyset$.\footnote{For technical simplicity, we do not track $\#$-\textit{sub}formulae as part of the inference behaviour of connective \#. We can think of the main connective as being \textit{explicitly} used in our reasoning (e.g. we \textit{introduce} $\neg$ in the sixth step of this proof via 1$\neg$). In contrast, connectives that are subformula constituents are not the main focus of the proof, only reasoned with \textit{implicitly}. However, this choice does affect our results as we will only consider single-connective calculi here.} In summary, we find that $\mathtt{ib}^\wedge_\mathfrak{p}=\langle [2], [2], [2], [2], [2], \emptyset, \emptyset \rangle$, and $\mathtt{ibp}^\wedge_\mathfrak{p}=\{\emptyset, [2]\}$.

We now apply the same process to negation, this time going through the proof tree bottom-up. As only one $\neg$-formula, $\neg (A\wedge B)$, occurs in the first component at the final vertical location, $H(\neg,7)=[1]$. At $V=6$, $\neg A$ and $\neg B$ are both $\neg$-formulae in the second component. As we also find $\neg (A\wedge B)$ in the first dimension, we obtain: $H(\neg,6)=[1,2,2]$. The latter formula is justified by an application of $1\neg$. We therefore locate the active formula of this $\neg$-rule application, $A\wedge B$, which occurs in the second component at $V=5$. Thus, we record its dimension, yielding $H(\neg,5)=[2,2,2]$. We repeat this process for both $2\neg$ applications in the previous steps, while still keeping track of our previously discovered active formulae. We obtain $H(\neg,4)=[1,2,2]$ and $H(\neg,3)=[1,1,2]$. Although $A\wedge B$ is an active formula of a $\neg$-rule, as we said, $A$ and $B$ in the second component of $V=1$ and $V=2$, respectively, are \textit{not}. Thus, $H(\neg,1)=[1]$ and $H(\neg,2)=[1]$. Hence, $\mathtt{ib}^\neg_\mathfrak{p}= \langle [1], [1], [1,1,2], [1,2,2]$, $[2,2,2]$, $[1,2,2]\rangle$, and  $\mathtt{ibp}^\neg_\mathfrak{p}=\{[1], [1,1,2], [1,2,2], [2,2,2]\}$.

Finally, we obtain for disjunction: $\mathtt{ib}^\vee_\mathfrak{p}=\langle \emptyset, \emptyset, \emptyset, [2], [2,2], [2,2], [2]\rangle$ and, thus, $\mathtt{ibp}^\vee_\mathfrak{p}=\{\emptyset, [2], [2,2]\}$.
\end{example}
We can think of $\texttt{ib}^\#_\mathfrak{p}$ as the syntactic fingerprint of \# in the series of inferences formalised by the proof $\mathfrak{p}$. Its profile, $\texttt{ibp}^\#_\mathfrak{p}$, tells us what kind of patterns can be found in this fingerprint. Inference behaviour is central to I-bS as the core identity criterion for connectives. The profile $\texttt{ibp}^\#_\mathfrak{p}$ is a useful technical tool that helps us find the right \textit{semantic clause} for \#, i.e. the minimal sequent rule yielding the minimal inference behaviour of \#. We will discuss semantic clauses and what we mean by the `minimal sequent rule' in \S6.

Note that we use a simplified notion of inference behaviour relative to \cite{nagler2026measuring}. We do not record which proof rules are applied at each step. This is acceptable since we investigate only definability proofs that succeed in the minimal derivability relation; Nagler \cite{nagler2026measuring} additionally allows for extensions of the derivability relation when a definability proof fails. In part, this is a pragmatic constraint to limit the scope of the investigation. However, there is also a more substantial reason: as argued in \S 6.1, connectives definable in the minimal derivability relation alone are \textit{minimally meaningful} and therefore semantically primitive.\footnote{In the absence of additional structural rules, tracking active formulae  automatically records at least the use of \textit{operational} rules.}

Let us now see how we can find the \textit{minimal} inference behaviour of a connective.

\subsection{Data type: definability proofs}

According to Nagler \cite{nagler2026measuring}, connectives share meaning if and only if they share minimal inference behaviour. Inference behaviour is \textit{minimal} if it is sufficient to demonstrate that a connective \textit{has meaning}. Given her inferentialism, Nagler connects this notion of \textit{meaningfulness} to \textit{usability}. We thus encounter a connective's minimal inference behaviour precisely when we prove that the connective \textit{can be used} in reasoning.

Showing that a connective \textit{can be used} means proving that its operational rules \textit{can be defined} in a given proof system. This is in keeping with the Gentzenian roots of P-tS. Following \cite{nagler2026measuring}, we embrace the definition in \cite{belnap1962tonk}:
\begin{definition}[definability \cite{belnap1962tonk}]\

    \noindent The connective $\#$ with operational rules $1\#,2\#$ is definable in the calculus \textbf{Cal} iff $\textbf{Cal}\cup \{1\#,2\#\}$ 
    \begin{enumerate}
        \item \textit{conservatively} extends \textbf{Cal}, and
        \item \textit{uniquely} determines the inferential role of $\#$ up to isomorphism.
    \end{enumerate}
\end{definition}
Intuitively, conservativity ensures that $\#$ has distinct inference behaviour in \textbf{Cal}, while uniqueness guarantees that it is the only connective with that behaviour. Put formally:
\begin{definition}[conservativity \cite{belnap1962tonk}]\

    \noindent $\textbf{Cal}\cup \{1\#,2\#\}$ conservatively extends \textbf{Cal} iff any sequents which are provable in $\textbf{Cal}\cup \{1\#,2\#\}$ but not provable in \textbf{Cal} contain `\#'.
\end{definition}
In plain terms, we can prove nothing new in the extended system, except with the rules we added. 
\begin{definition}[uniqueness \cite{belnap1962tonk}]\

    \noindent $\textbf{Cal}\cup \{1\#,2\#\}$ \textit{uniquely} determines the inferential role of $\#$ up to isomorphism iff for any $i\in\{1,2\}$ and any $i\natural=i\#$ (modulo substituting $\natural$ for $\#$), there are some $\textbf{Cal}\cup\{ 1\#,2\#, 1\natural, 2\natural\}$-derivations $\mathfrak{d}, \mathfrak{d}^{\prime}$ such that
    \[
        \begin{prooftree}
        \hypo{[i:\texttt{\#}]s}
        \infer[rule code={\hbox{\tikz
    \draw (0,0) -- (\hsize,0) -- (0.5\hsize,-2em) (0.5\hsize,-0.875em) node{\small $\mathfrak{d}$}  (0.5\hsize,-2em) -- (0,0);}}]1{[i:\na]s}
        \end{prooftree}\qquad \text{and} \qquad
        \begin{prooftree}
        \hypo{[i:\na]s}
        \infer[rule code={\hbox{\tikz
    \draw (0,0) -- (\hsize,0) -- (0.5\hsize,-2em) (0.5\hsize,-0.875em) node{\small $\mathfrak{d^\prime}$}  (0.5\hsize,-2em) -- (0,0);}}]1{[i:\texttt{\#}]s}
        \end{prooftree}\, ,
    \]
    where `\texttt{\#}' denotes any $\#$-formula and `$\na$' denotes any $\natural$-formula.
\end{definition}
Informally, connectives with the same operational rules behave identically in all contexts.

We now know in what type of data we will be measuring inference behaviour: in proofs of a connective's \textit{definability}. However, we have yet to specify in \textit{which calculus} definability is to be proved.

\subsection{Testing environment: minimal derivability relation}

While Nagler \cite{nagler2026measuring} is a pluralist about the choice of derivability relation for Belnap-style definability proofs, she requires it to be \textit{minimal}. To defend this, she invokes a form of \textit{(meta-)semantic anti-exceptionalism}.

Nagler \cite{nagler2026measuring} holds that measurements of inference behaviour, and thus of meaning, are continuous with scientific measurements. Both are subject to similar desiderata of test design; in particular, one must isolate the variables under investigation and exclude confounding variables. In the context of I-bS, these `confounders' are the structural properties of the background calculus, specifically, properties exceeding what is required to formulate a definability proof.

I-bS thus involves a trade-off: to discern the meaning of a connective independently of a specific logic, we must study its minimal (i.e. definitional) use; however, we can only study this minimal use in the context of a logical calculus. Full isolation of connectives from the structural properties of a derivability relation is impossible; one can only \textit{minimise} structural confounders.

What structural resources are minimal, i.e. indispensable for a connective's definability proof? Most obviously, an axiom is needed to initiate proofs, especially for non-axiomatic connectives. Moreover, as we are operating in two-dimensional calculi, it seems reasonable to require some rule that allows for cross-component interactions. To satisfy these two constraints, \cite{nagler2026measuring} argues for a background calculus of $\{\textsc{Id}, \textsc{Cut}\}$ yielding a solely reflexive and transitive derivability relation as is typical of substructural calculi such as the Lambek calculus \cite{lambek2011mathematics} or linear logic \cite{girard1989proofs}. This is a reasonable starting point, as the results are compatible with a large range of sub- and fully structural calculi. We therefore follow Nagler \cite{nagler2026measuring} in this choice. 

While Nagler \cite{nagler2026measuring} discusses both context-differentiating and context-sharing formulations of \textsc{Cut} without clearly favouring either, we adopt the former:
\[
\begin{prooftree}
    \hypo{[2:A]s}
    \hypo{[1:A]t}
    \infer2[\textsc{Cut}$cd$]{st}
\end{prooftree}\, .
\]
While replicating this study with context-sharing \textsc{Cut} remains worthwhile, \textsc{Cut}$cd$ is the more commonly used formulation and corresponds more straightforwardly to the Tarskian property of transitivity as \cite{nagler2026measuring} argues.

However, in the case of \textsc{Id}, we find that \cite{nagler2026measuring} does not choose the most minimal version available. Restrictions of \textsc{Id} to only primitive formulae are standard in the literature \cite{troelstra2000basic, negri2001structural} and prove strictly fewer theorems than depth-unrestrained formulations. Therefore, we will limit ourselves to the restricted version, \textsc{Id}$^{pr}$:
\[
\begin{prooftree}
    \hypo{\emptyset}
    \infer1[\textsc{Id}$^{pr}$]{1:A^{pr};2:A^{pr}}
\end{prooftree}\, ,
\]
where A$^{pr}$ is primitive.

This design choice will simplify and shorten some of our later proofs. Nevertheless, it does not affect the results since full \textsc{Id} remains derivable for all connectives uniquely definable in ${\textsc{Cut}cd, \textsc{Id}^{pr}}$.\footnote{Thus, we will drop the superscript $^{pr}$ in the following.}

As we will only prove definability for one connective at a time, we can stipulate that our language $\mathcal{L}$ only consists of the sequent symbol $\vdash$, arbitrarily many atomic sentential variables $p_0, \dots, p_{n\in \mathbb{N}}$, the connective $\#$, and nothing else, unless indicated otherwise.\footnote{This is our official language; our shorthand sequent notation is defined in terms of it.}

\begin{definition}[formula depth $d$]\

\begin{itemize}
\item If formula $A$ is an atomic variable, we call it \textit{primitive}, and $d(A)=0$;
\item otherwise, $A$ contains $\#$, we call it \textit{complex}, and 
\begin{itemize}
    \item  if \# is nullary, $A=\#$ and $d(\#)=1$,
    \item if $\#$ is unary, $A=\#B$ and $d(\#B)=1+d(A)$, 
    \item if $\#$ is binary, $A=B\#C$ and $d(A\#B)=1+max\{d(B), d(C)\}$.
\end{itemize}
\end{itemize}
\end{definition}
\begin{theorem}[derivability of \textsc{Id}]\

\noindent
Consider a calculus consisting of the operational rules of an arbitrary uniquely definable connective $\#$, \textsc{Id$^{pr}$} and potentially further structural rules such as \textsc{Cut}$cd$ (or \textsc{Cut}$cs$). In this calculus, full \textsc{Id} is derivable, i.e. for any formula $A$ of arbitrary depth:
\begin{prooftree*}\hypo{\emptyset} \infer1[\textsc{Id}]{[1:A;2:A]} \end{prooftree*}
\end{theorem}
\begin{proof}\

\noindent
By induction on the depth of $A$. If $A$ is primitive, $[1:A;2:A]$ by \textsc{Id$^{pr}$}. Otherwise, $A$ is a $\#$-formula, \texttt{\#}, as all complex formulae can only be a result of 1\# or 2\#. Then, by \textit{uniqueness}, there is a proof $\mathfrak{p}$ such that \[
\begin{prooftree}
\hypo{~~~\mathfrak{A}~~~}
\infer[rule code={\hbox{\tikz
    \draw (0,0) -- (\hsize,0) -- (0.5\hsize,-2em) (0.5\hsize,-0.875em) node{\small $\mathfrak{p}$}  (0.5\hsize,-2em) -- (0,0);}}]1{\texttt{\#} \vdash \na}
\end{prooftree}\, ,
\]
where $\natural$ is defined by the same operational rules as $\#$ and $\na$ is the same formula as \texttt{\#}, modulo symbol substitution, and $\mathfrak{A}$ is a set of axioms. All $\mathfrak{s}\in \mathfrak{A}$ are either of form
\begin{enumerate}
\item $[k:\texttt{\#}]s$ (if k$\#$ is an axiomatic rule),
\item $[k:\na]s$ (if k$\natural$ is an axiomatic rule), or
\item $[1:B;2:B]$,
\end{enumerate} 
for $k\in\{1,2\}$. 

1./2. hold trivially. If $B$ is primitive, 3. is a result of \textsc{Id$^{pr}$}; otherwise $B$ is of shape $\texttt{\#}^\prime$, which holds by virtue of the inductive hypothesis. \hfill $\blacksquare$
\end{proof}
Note that Nagler \cite{nagler2026measuring}, by defining sequents over multisets, implicitly assumes full \textsc{Exchange}. We adopt this constraint, though we regard an \textsc{Exchange}-free derivability relation as the more minimal option. However, this choice significantly limits the length of the investigation, which is an important consideration here. We conjecture that dropping \textsc{Exchange} will result in a finer-grained non-commutative picture akin to the ordered logic connectives \cite{lambek1958mathematics, polakow1999relating}.

We now have all the ingredients needed to begin the investigation. In addition to our object of inquiry, the measurement device of \textit{inference behaviour} and the testing environment of $\{\textsc{Id}, \textsc{Cut}cd\}$ are at our disposal. Let us now generate our dataset: definability proofs for each formulable connective.

\section{Data generation: the proofs}

To generate the data for our measurement of inference behaviour, let us see which formulable connectives are definable in $\{\textsc{Id}, \textsc{Cut}cd\}$. In \S 4.3, we looked at the 10,816 connectives we can formulate in two-dimensional sequent calculi with $\leq 2$ premiss sequents,  $\leq 2$ active formulae (filter 1). In \S5.1, we will find that  exactly $376$ among these connectives (non-trivially) satisfy \textit{conservativity} (filter 2), and $150$ of these are also (non-trivially) unique (filter 3), as shown in \S5.2.

\begin{figure}[h]
    \centering
    \begin{tikzpicture}
        \filldraw[fill=gray!20] (4.25,10) -- (9.75,10) -- (8,9) -- (6,9) -- (4.25,10);
        \filldraw[fill=gray!20] (5,8) -- (9,8) -- (7.75,7.25) -- (6.25, 7.25) -- (5,8);
        \filldraw[fill=gray!20] (5.75,6.25) -- (8.25,6.25) -- (7.5,5.75) -- (6.5,5.75) -- (5.75,6.25);
        
        \draw[->] (7,10.25) -- (7,9.75);
        \draw[->] (7,9.25) -- (7,7.75);
        \draw[->] (7,7.45) -- (7,6.05);
        \draw[->] (7,5.85) -- (7,5.5);

        \draw (7,8.5) node[fill= white]{$10,816$};
        \draw (7,6.75) node[fill=white]{$376$};
        \draw (7,5.25) node[]{$150$};
        
        \draw (3,9.5) node{filter 1: \textit{formulability}};
        \draw (3.5,7.65) node{filter 2: \textit{conservativity}};
        \draw (4,6) node{filter 3: \textit{uniqueness}};
        
    \end{tikzpicture}
    \caption{data generation, definability filters}
\end{figure}
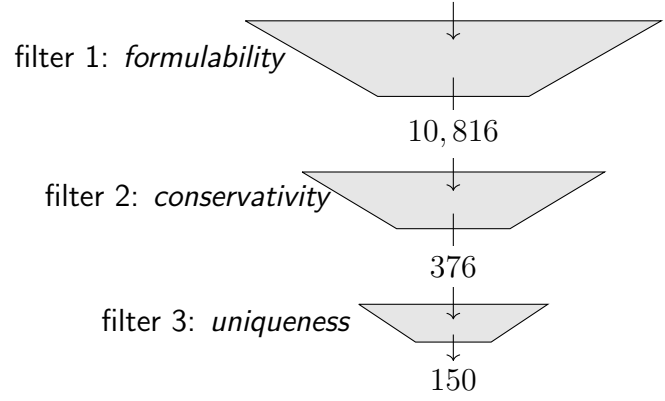

\subsection{Conservativity}

Let us first establish the positive cases, i.e. which of our connectives are \textit{conservative} in ${\textsc{Id}, \textsc{Cut}cd}$. We then show that our set of positive cases is exhaustive by proving that all other (negative) cases are non-conservative.

\subsubsection{Positive cases}

We prove \textsc{Cut}-elimination for the positive cases and then use this result to establish conservativity.

\begin{definition}[positive cases, conservativity]\

    \noindent Let $\mathtt{PC}$ be the set of the following $376$ connectives; for $k,l,m,n \in \{1,2\}$, $F,G \in \{A, B\}$ and $a \in \mathbb{N}^0$:
        \begin{itemize}
            \item[1.1)] $\langle \alpha, \gamma 0, a, \mathbb{N}^0:\mathbb{N}^0 \rangle, \langle \gamma 0, \alpha, a, \mathbb{N}^0:\mathbb{N}^0 \rangle$ (2);
            \item[1.2)] $\langle \alpha, \delta cs 0, a, \mathbb{N}^0:\mathbb{N}^0 \rangle, \langle \delta cs 0, \alpha, a, \mathbb{N}^0:\mathbb{N}^0 \rangle$ (2);
            \item[2.1)] $\langle \alpha, \delta cd^- kF, a, \mathbb{N}^0:\mathbb{N}^0 \rangle, \langle \delta cd^- kF, \alpha, a, \mathbb{N}^0:\mathbb{N}^0 \rangle$ (8);
            \item [2.2)] $\langle \gamma 0, \beta kF, a, \mathbb{N}^0:\mathbb{N}^0 \rangle, \langle \beta kF, \gamma 0, a, \mathbb{N}^0:\mathbb{N}^0 \rangle$ (8);
            \item [2.3)] $\langle \delta cs 0, \beta kF, a, \mathbb{N}^0:\mathbb{N}^0 \rangle, \langle \beta kF, \delta cs 0, a, \mathbb{N}^0:\mathbb{N}^0 \rangle$ (8);
            \item[3.1)] $\langle \alpha, \delta cd^- kFlG, a, \mathbb{N}^0:\mathbb{N}^0 \rangle, \langle \delta cd^- kFlG, \alpha, a, \mathbb{N}^0:\mathbb{N}^0 \rangle$ (20);
            \item [3.2)] $\langle \gamma 0, \beta kFlG, a, \mathbb{N}^0:\mathbb{N}^0 \rangle, \langle \beta kFlG, \gamma 0, a, \mathbb{N}^0:\mathbb{N}^0 \rangle$ (20);
            \item [3.3)] $\langle \gamma 0, \gamma^- 1F2F, a, \mathbb{N}^0:\mathbb{N}^0 \rangle, \langle \gamma^- 1F2F, \gamma 0, a, \mathbb{N}^0:\mathbb{N}^0 \rangle$ (4);
            \item [3.4)] $\langle \delta cs 0, \beta kFlG, a, \mathbb{N}^0:\mathbb{N}^0 \rangle, \langle \beta kFlG, \delta cs 0, a, \mathbb{N}^0:\mathbb{N}^0 \rangle$ (20);
            \item [3.5)] $\langle \delta cs 0, \gamma^- 1F2F, a, \mathbb{N}^0:\mathbb{N}^0 \rangle, \langle \gamma^- 1F2F, \delta cs 0, a, \mathbb{N}^0:\mathbb{N}^0 \rangle$ (4);
            \item[4.1)] $\langle \alpha, \delta cd 1F-2F, a, \mathbb{N}^0:\mathbb{N}^0 \rangle, \langle \delta cd 1F-2F, \alpha, a, \mathbb{N}^0:\mathbb{N}^0 \rangle$ (4);
            \item[5.1)] $\langle \beta kF, \delta cd^- lG, a, \mathbb{N}^0:\mathbb{N}^0 \rangle, \langle \delta cd^- kF, \beta lG, a, \mathbb{N}^0:\mathbb{N}^0 \rangle$ (32);
            \item[5.2)] $\langle \gamma kF, \gamma lF, a, \mathbb{N}^0:\mathbb{N}^0 \rangle, k\neq l$ (4);
            \item[6.1)] $\langle \beta kF, \delta cd^- lFmG, a, \mathbb{N}^0:\mathbb{N}^0 \rangle, \langle \delta cd^- lFmG, \beta kF, a, \mathbb{N}^0:\mathbb{N}^0 \rangle$ (40);
            \item[6.2)] $\langle \delta cd^- kF, \beta lFmG, a, \mathbb{N}^0:\mathbb{N}^0 \rangle, \langle \beta lFmG, \delta cd^- kF, a, \mathbb{N}^0:\mathbb{N}^0 \rangle$ (40);
            \item[7.1)] $\langle \beta kF, \delta cd 1F-2F, a, \mathbb{N}^0:\mathbb{N}^0 \rangle, \langle \delta cd 1F-2F, \beta kF, a, \mathbb{N}^0:\mathbb{N}^0 \rangle$ (8);
            \item[7.2)] $\langle \gamma kF, \delta cs lF-lF, a, \mathbb{N}^0:\mathbb{N}^0 \rangle, \langle \delta cs lF-lF, \gamma kF, a, \mathbb{N}^0:\mathbb{N}^0 \rangle, k\neq l$ (4);
            \item[8.1)] $\langle \beta kFlG, \delta cd^- mFnG, a, \mathbb{N}^0:\mathbb{N}^0 \rangle, \langle \delta cd^- mFnG, \beta kFlG, a, \mathbb{N}^0:\mathbb{N}^0 \rangle$ (68);
            \item[8.2)] $\langle \gamma^- 1F2F, \delta cd^- 1F2F, a, \mathbb{N}^0:\mathbb{N}^0 \rangle, \langle \delta cd^- 1F2F, \gamma^- 1F2F, a, \mathbb{N}^0:\mathbb{N}^0 \rangle$ (4);
            \item[9.1)] $\langle \beta kFlF, \delta cd 1F-2F, a, \mathbb{N}^0:\mathbb{N}^0 \rangle, \langle \delta cd 1F-2F,\beta kFlF, a, \mathbb{N}^0:\mathbb{N}^0 \rangle$ (12);
            \item[9.2)] $\langle \gamma kFlG, \delta cd mF-nG, a, \mathbb{N}^0:\mathbb{N}^0 \rangle, \langle \delta cd mF-nG, \gamma kFlG, a, \mathbb{N}^0:\mathbb{N}^0 \rangle, k\neq m, l\neq n$ (20);
            \item[9.3)] $\langle \delta cs kFlG, \delta cd mF-nG, a, \mathbb{N}^0:\mathbb{N}^0 \rangle, \langle \delta cd mF-nG, \delta cs kFlG, a, \mathbb{N}^0:\mathbb{N}^0 \rangle, k\neq m, l\neq n$ (20);
            \item[10.1)] $\langle \delta cs kF-kF, \delta cs mF-mF, a, \mathbb{N}^0:\mathbb{N}^0 \rangle, k\neq m$ (4);
            \item[10.2)] $\langle \delta cs kF-lG, \epsilon mF-nG, a, \mathbb{N}^0:\mathbb{N}^0 \rangle, \langle \epsilon mF-nG, \delta cs kF-lG, a, \mathbb{N}^0:\mathbb{N}^0 \rangle, k\neq m, l\neq n$ (20).
    \end{itemize}
 We write `$\delta cd^-$' for rules of shape $\delta cd$ for which the context of (at least) one premiss sequent (say, $s$) is restricted to the empty sequent (i.e. $s=[1:\emptyset;2: \emptyset]$). Analogously, we use `$\gamma ^-$'  to designate rules of shape $\gamma$ in which the context is restricted to the empty sequent.

 Our enumeration of \texttt{PC} follows the order of cases in the proof of Lemma 1, iterating through all possible \texttt{P}-constellations. Setting aside arities and frameworks, there are $376$ distinct members of \texttt{PC}. The brackets next to each bullet point indicate how many of these $376$ distinct connectives are captured in each case.
\end{definition}
Let us prove \textsc{Cut}-elimination in $\{\textsc{Id}, \textsc{Cut}cd, 1\#, 2\#\}$ for each $\# \in \mathtt{PC}$. Our proof follows \cite{mancosu2021introduction}'s variant of \cite{gentzen1935auntersuchungen, gentzen1935buntersuchungen}'s \textit{Hauptsatz}. 
\begin{definition}[derivation height $h$]\

\noindent The height $h$ of a derivation is the maximum number of successive rule applications from a leaf (initial sequent) to the root (end-sequent).
\end{definition}
\begin{definition}[\textsc{Cut}-depth $\mathcal{D}$, \textsc{Cut}-height $\mathcal{H}$]\

\noindent Consider any \textsc{Cut}-derivation $\mathfrak{d}$ of the shape:
\[ \begin{prooftree}
\hypo{}
\infer[rule code={\hbox{\tikz
    \draw (0,0) -- (\hsize,0) -- (0.5\hsize,-2em) (0.5\hsize,-0.875em) node{\small $\mathfrak{d}_1$}  (0.5\hsize,-2em) -- (0,0);}}]1{[1:A]s}
\hypo{}
\infer[rule code={\hbox{\tikz
    \draw (0,0) -- (\hsize,0) -- (0.5\hsize,-2em) (0.5\hsize,-0.875em) node{\small $\mathfrak{d}_2$}  (0.5\hsize,-2em) -- (0,0);}}]1{[2:A]t}
\infer2[\textsc{Cut}]{st}
\end{prooftree}\, .
\]
The \textsc{Cut}-depth $\mathcal{D}$ of $\mathfrak{d}$ is the depth of its active formula: $\mathcal{D}(\mathfrak{d})=d(A)$. The \textsc{Cut}-height $\mathcal{H}$ of $\mathfrak{d}$ is the sum of heights of the sub-derivations of each \textsc{Cut}-premiss: $\mathcal{H}(\mathfrak{d})=h(\mathfrak{d}_1)+h(\mathfrak{d}_2)$.
\end{definition}
\begin{definition}[Reducibility]\

\noindent We call any \textbf{Cal}-derivation using \textsc{Cut} once in its final step \textit{reducible} iff we can construct another \textbf{Cal}-derivation from the same premises to the same conclusion but without using \textsc{Cut}.
\end{definition}
Our tools at hand, we first prove a core lemma for the \textsc{Cut}-elimination proof, focussing on the principal formulae. 
\begin{lemma}\

\noindent 
Let our calculus \textbf{Cal} be $\{$\textsc{Id}$^{pr}$, \textsc{Cut}$cd$, 1$\#$, 2$\#\}$.  Let \texttt{\#} be any $\mathcal{L}$-formula containing $\#$ as its main connective (for $\#$ of any arity).

Assume $\mathfrak{d}$ to be a \textbf{Cal}-derivation which cuts over the conclusions of two operational rules in its final step with $\mathcal{D}(\mathfrak{d})>0$ and $\mathcal{H}(\mathfrak{d})=2$:
\[ \begin{prooftree}
\hypo{\mathfrak{S}_1}
\infer1[1\#]{[1:\texttt{\#}]s}
\hypo{\mathfrak{S}_2}
\infer1[2\#]{[2:\texttt{\#}]t}
\infer2[\textsc{Cut}]{st}
\end{prooftree}\, ,
\]
where $\mathfrak{S}_1, \mathfrak{S}_2$ are sets of sequents.

For any $\#\in \mathtt{PC}$, there is a \textbf{Cal}-derivation $\mathfrak{d}^\prime$ of the form
\[ \begin{prooftree}
\hypo{\mathfrak{S}_1 \cup \mathfrak{S}_2}
\infer[rule code={\hbox{\tikz
    \draw (0,0) -- (\hsize,0) -- (0.5\hsize,-2em) (0.5\hsize,-0.875em) node{\small $\mathfrak{d}^\prime$}  (0.5\hsize,-2em) -- (0,0);}}]1{st}
\end{prooftree}\, ,
\]
such that $\mathfrak{d}^\prime$ is \textsc{Cut}-free, or $\mathcal{H}(\mathfrak{d}^\prime)<\mathcal{H}(\mathfrak{d})$, or $\mathcal{D}(\mathfrak{d}^\prime)<\mathcal{D}(\mathfrak{d})$.
\end{lemma}
\begin{proof}\

\noindent 
Let $i,j,k,l\in \{1,2\}$ and $i \neq j$. Let $\F, G\in \{A, B\}$.\\ \\
Case 1) \texttt{P}$(i\#)=\texttt{P}(j\#)=0$.
\begin{enumerate}[label=1.\arabic*)]
\item $\texttt{S}(i\#)=\alpha; \texttt{S}(j\#)=\gamma $:
\[
\begin{prooftree}
\hypo{\emptyset}
\infer1[i\#]{[i:\texttt{\#}]}
\hypo{s}
\infer1[j\#]{[j:\texttt{\#}]s}
\infer2[\textsc{Cut}]{s}
\end{prooftree}
\]
\text{can be reduced to} 
\[\begin{prooftree}
\hypo{s}
\infer1{s}
\end{prooftree}\, .
\]
\item $\texttt{S}(i\#)=\alpha; \texttt{S}(j\#)=\delta cs$: cf. 1.1).
\end{enumerate}
Case 2) \texttt{P}$(i\#)=0$; \texttt{P}$(j\#)=kF$.
\begin{enumerate}[label=2.\arabic*)]
\item $\texttt{S}(i\#)=\alpha; \texttt{S}(j\#)=\delta cd^-$: 
    \[
    \begin{prooftree}
    \hypo{\emptyset}
    \infer1[i\#]{[i:\texttt{\#}]}
    \hypo{[k:F]}
    \hypo{s}
    \infer2[j\#]{[j:\texttt{\#}]s}
    \infer2[\textsc{Cut}]{s}
    \end{prooftree}
\]
    \text{can be reduced to} 
\[
    \begin{prooftree}
    \hypo{s}
    \infer1{s}
    \end{prooftree}\, .
    \] 
\item $\texttt{S}(i\#)=\gamma; \texttt{S}(j\#)=\beta$: cf. 1.1), and
\[
\begin{prooftree}
\hypo{s}
\infer1[i\#]{[i:\texttt{\#}]s}
\hypo{[k:F]}
\infer1[j\#]{[j:\texttt{\#}]}
\infer2[\textsc{Cut}]{s}
\end{prooftree}
\]
    \text{can be reduced to} 
\[
\begin{prooftree}
\hypo{s}
\infer1{s}
\end{prooftree}\, .
\] 
\item $\texttt{S}(i\#)=\delta cs; \texttt{S}(j\#)=\beta$: cf. 1.1).
\end{enumerate}
Case 3) \texttt{P}$(i\#)=0$; \texttt{P}$(j\#)=kFlG$:
\begin{enumerate}[label=3.\arabic*)]
    \item $\texttt{S}(i\#)=\alpha; \texttt{S}(j\#)=\delta cd^-$: cf. 2.1).
    \item $\texttt{S}(i\#)=\gamma; \texttt{S}(j\#)=\beta$: cf. 2.2).
    \item $\texttt{S}(i\#)=\gamma; \texttt{S}(j\#)=\gamma^-$: let $k\neq l$, $F=G$, then
    \[
\begin{prooftree}
\hypo{s}
\infer1[i\#]{[i:\texttt{\#}]s}
\hypo{[k:F;l:G]}
\infer1[j\#]{[j:\texttt{\#}]}
\infer2[\textsc{Cut}]{s}
\end{prooftree}
\]
    \text{can be reduced to} 
\[
\begin{prooftree}
\hypo{s}
\infer1{s}
\end{prooftree}\, .
\]
\item $\texttt{S}(i\#)=\delta cs; \texttt{S}(j\#)=\beta$: cf. 1.1).
\item $\texttt{S}(i\#)=\delta cs; \texttt{S}(j\#)=\gamma^-$: cf. 3.3).
\end{enumerate}
Case 4) $\texttt{P}(i\#)=0; \texttt{P}(j\#)=kF-lG$.
\begin{enumerate}[label=4.\arabic*)]
    \item $\texttt{S}(i\#)=\alpha; \texttt{S}(j\#)=\delta cd$: let $k\neq l$, $F=G$, then
\[
\begin{prooftree}
\hypo{\emptyset}
\infer1[i\#]{[i:\texttt{\#}]}
\hypo{[k:F]s}
\hypo{[l:F]t}
\infer2[j\#]{[j:\texttt{\#}]st}
\infer2[\textsc{Cut}]{st}
\end{prooftree}
\]
    \text{can be reduced to} 
\[
\begin{prooftree}
\hypo{[k:F]s}
\hypo{[l:F]t}
\infer2[\textsc{Cut}]{st}
\end{prooftree}\, .
\]
\end{enumerate}
Case 5) $\texttt{P}(i\#)=kF; \texttt{P}(j\#)=lG$.
\begin{enumerate}[label=5.\arabic*)]
\item $\texttt{S}(i\#)=\beta; \texttt{S}(j\#)=\delta cd^-$: cf. 2.1).
\item $\texttt{S}(i\#)=\gamma; \texttt{S}(j\#)=\gamma$, let $k\neq l, F=G$:
\[
\begin{prooftree}
\hypo{[k:F]s}
\infer1[i\#]{[i:\texttt{\#}]s}
\hypo{[l:F]t}
\infer1[j\#]{[j:\texttt{\#}]t}
\infer2[\textsc{Cut}]{st}
\end{prooftree}
\]
    \text{can be reduced to} 
\[
\begin{prooftree}
\hypo{[k:F]s}
\hypo{[l:F]t}
\infer2[\textsc{Cut}]{st}
\end{prooftree}\, .
\]
\end{enumerate}
Case 6) $\texttt{P}(i\#)=kF; \texttt{P}(j\#)=lFmG$.
\begin{enumerate}[label=6.\arabic*)]
\item $\texttt{S}(i\#)=\beta; \texttt{S}(j\#)=\delta cd^-$: cf. 2.1).
\item $\texttt{S}(i\#)=\delta cd^-; \texttt{S}(j\#)=\beta$: cf. 2.1).
\end{enumerate}
Case 7) $\texttt{P}(i\#)=kF; \texttt{P}(j\#)=lF-mG$.
\begin{enumerate}[label=7.\arabic*)]
    \item $\texttt{S}(i\#)=\beta; \texttt{S}(j\#)=\delta cd$: if $F=G, l\neq m$,
    \[
    \begin{prooftree}
        \hypo{[k:F]+\emptyset}
        \infer1[i\#]{[i:\texttt{\#}]}
        \hypo{[l:F]s}
        \hypo{[m:F]t}
        \infer2[j\#]{[j:\texttt{\#}]st}
        \infer2[\textsc{Cut}]{st}
    \end{prooftree}
\]
    \text{can be reduced to} 
\[
    \begin{prooftree}
        \hypo{[l:F]s}
        \hypo{[m:F]t}
        \infer2[\textsc{Cut}]{st}
    \end{prooftree}\, .
    \]
    \item $\texttt{S}(i\#)=\gamma; \texttt{S}(j\#)=\delta cs$: if $k\neq l$,
    \[
    \begin{prooftree}
        \hypo{[k:F]s}
        \infer1[i\#]{[i:\texttt{\#}]s}
        \hypo{[l:F]t}
        \hypo{[m:G]t}
        \infer2[j\#]{[j:\texttt{\#}]t}
        \infer2[\textsc{Cut}]{st}
    \end{prooftree}
\]
    \text{can be reduced to} 
\[
    \begin{prooftree}
        \hypo{[k:F]s}
        \hypo{[l:F]t}
        \infer2[\textsc{Cut}]{st}
    \end{prooftree}\, .
    \]    

\end{enumerate}
Case 8) $\texttt{P}(i\#)=kFlG; \texttt{P}(j\#)=mFnG.$
\begin{enumerate}[label=8.\arabic*)]
\item $\texttt{S}(i\#)=\beta; \texttt{S}(j\#)=\delta cd^-$: cf. 2.1).
\item $\texttt{S}(i\#)=\gamma^-; \texttt{S}(j\#)=\delta cd^-$: cf. 2.1).
\end{enumerate}
Case 9) $\texttt{P}(i\#)=kFlG; \texttt{P}(j\#)=mF-nG$.
\begin{enumerate}[label=9.\arabic*)]
\item $\texttt{S}(i\#)=\beta; \texttt{S}(j\#)=\delta cd$: if $m\neq n, F=G$, cf. 7.1).
\item $\texttt{S}(i\#)=\gamma; \texttt{S}(j\#)=\delta cd$: if $k\neq m, l\neq n$,
\[
\begin{prooftree}
\hypo{[k:F;l:G]s}
\infer1[i\#]{[i:\texttt{\#}]s}
\hypo{[m:F]t}
\hypo{[n:G]u}
\infer2[j\#]{[j:\texttt{\#}]tu}
\infer2[\textsc{Cut}]{stu}
\end{prooftree} 
\]
    \text{can be reduced to} 
\[
\begin{prooftree}
\hypo{[k:F;l:G]s}
\hypo{[m:F]t}
\infer2[\textsc{Cut}]{[l:G]st}
\hypo{[n:G]u}
\infer2[\textsc{Cut}]{stu}
\end{prooftree}\, ,
\]
or
\[
\begin{prooftree}
\hypo{[k:F;l:G]s}
\hypo{[n:G]u}
\infer2[\textsc{Cut}]{[k:F]su}
\hypo{[m:F]t}
\infer2[\textsc{Cut}]{stu}
\end{prooftree}\, .
\]
\item $\texttt{S}(i\#)=\delta cs; \texttt{S}(j\#)=\delta cd$: if $k\neq m, l\neq n$, cf. 9.2).
\end{enumerate}
Case 10) $\texttt{P}(i\#)=kF-lG; \texttt{P}(j\#)=mF-nG$.
\begin{enumerate}[label=10.\arabic*)]
\item $\texttt{S}(i\#)=\delta cs; \texttt{S}(j\#)=\delta cs$: if $k\neq m, l\neq n$,
\[
\begin{prooftree}
\hypo{[k:F]s}
\hypo{[l:G]s}
\infer2[i\#]{[i:\texttt{\#}]s}
\hypo{[m:F]t}
\hypo{[n:G]t}
\infer2[j\#]{[j:\texttt{\#}]t}
\infer2[\textsc{Cut}]{st}
\end{prooftree} 
\]
    \text{can be reduced to} 
\[
\begin{prooftree}
\hypo{[k:F]s}
\hypo{[m:F]t}
\infer2[\textsc{Cut}]{st}
\end{prooftree}\, .
\]
\item $\texttt{S}(i\#)=\delta cs; \texttt{S}(j\#)=\epsilon$: if $k\neq m, l\neq n$,
\[
\begin{prooftree}
\hypo{[k:F]s}
\hypo{[l:G]s}
\infer2[i\#]{[i:\texttt{\#}]s}
\hypo{[m:F]t}
\infer1[j\#]{[j:\texttt{\#}]t}
\infer2[\textsc{Cut}]{st}
\end{prooftree} 
\]
    \text{can be reduced to} 
\[
\begin{prooftree}
\hypo{[k:F]s}
\hypo{[m:F]t}
\infer2[\textsc{Cut}]{st} 
\end{prooftree}\, ,
\]
and
\[
\begin{prooftree}
\hypo{[k:F]s}
\hypo{[l:G]s}
\infer2[i\#]{[i:\texttt{\#}]s}
\hypo{[n:G]t}
\infer1[j\#]{[j:\texttt{\#}]t}
\infer2[\textsc{Cut}]{st}
\end{prooftree} 
\]
    \text{can be reduced to} 
\[
\begin{prooftree}
\hypo{[l:G]s}
\hypo{[n:G]t}
\infer2[\textsc{Cut}]{st} 
\end{prooftree}\, .
\]
\end{enumerate}
\hfill $\blacksquare$
\end{proof}
Using this lemma, let us now prove \textsc{Cut}-elimination:
\begin{theorem}\

\noindent
Let \textbf{Cal} be $\{\textsc{Id}, \textsc{Cut}cd,1\#,2\#\}$, for any $\# \in \mathtt{PC}$.

For any \textbf{Cal}-proof containing \textsc{Cut}, we can construct another \textbf{Cal}-proof which does not contain \textsc{Cut} but terminates in the same root.
\end{theorem}
\begin{proof}\

\noindent By induction on the \textsc{Cut}-height and sub-induction on the \textsc{Cut}-depth of any \textbf{Cal}-proof $\mathfrak{p}$ using \textsc{Cut} once in its final step.
\begin{itemize}
\item \textit{Base case: }we can reduce any minimal $\mathfrak{p}$, where $\mathcal{D}(\mathfrak{p})=0$ and $\mathcal{H}(\mathfrak{p})=2$.
\item \textit{Inductive hypothesis (IH):} we can reduce any $\mathfrak{p}^\prime$ if its \textsc{Cut}-depth or -height is lower than that of $\mathfrak{p}$, i.e. $\mathcal{D}(\mathfrak{p}^\prime)<\mathcal{D}(\mathfrak{p})$, or $\mathcal{D}(\mathfrak{p}^\prime)=\mathcal{D}(\mathfrak{p})$ and $\mathcal{H}(\mathfrak{p}^\prime)<\mathcal{H}(\mathfrak{p})$.
\item \textit{Inductive step:} we can reduce any $\mathfrak{p}$.
\end{itemize}
\textit{Base case.} $\mathcal{H}(\mathfrak{p})=1$. The \textsc{Cut}-formula is introduced immediately before the \textsc{Cut}. Then, it must be a conclusion of \textsc{Id}, since no $\#\in\mathtt{PC}$ satisfies $\mathtt{SP}(\#)=\langle \alpha, \alpha \rangle$ and \textsc{Cut} is excluded by hypothesis:
\[ \begin{prooftree}
\hypo{[1:A; 2:A]}
\hypo{[1:A; 2:A]}
\infer2[\textsc{Cut}]{[1:A; 2:A]}
\end{prooftree}\, .
\]
This reduces trivially, as the conclusion sequent can be derived from \textsc{Id} directly.\\ \\
\textit{Inductive step.} 
\begin{enumerate}
\item $\mathcal{D}(\mathfrak{p})>0$ and $\mathcal{H}(\mathfrak{p})=2$.
\begin{enumerate}
\item One of the premises could be an instance of \textsc{Id}. Since \textsc{Id} is restricted to primitives, the base case already proven is the only available one. All other rules of the calculus only derive formulae of higher depth, and therefore their principal formula cannot be used for a \textsc{Cut} together with the principal formula of \textsc{Id}.
\item Otherwise, both of the \textsc{Cut}-premises are introduced by connective rules. Then, there is a derivation $\mathfrak{p}^\prime$ with the same root and leaves as $\mathfrak{p}$ and $\mathcal{D}(\mathfrak{p}^\prime)<\mathcal{D}(\mathfrak{p})$, as we have proven in Lemma 1. By the IH, $\mathfrak{p}^\prime$ is reducible.
\end{enumerate} 
\item $\mathcal{H}(\mathfrak{p})>2$. Then, $\mathfrak{p}$ is of the shape
\[
\begin{prooftree}
\hypo{~~~~~~~~~}
\infer[rule code={\hbox{\tikz
    \draw (0,0) -- (\hsize,0) -- (0.5\hsize,-2em) (0.5\hsize,-0.875em) node{\small $\mathfrak{p}_1$}  (0.5\hsize,-2em) -- (0,0);}}]1{[i:A]s}
    \hypo{~~~~~~~~~}
\infer[rule code={\hbox{\tikz
    \draw (0,0) -- (\hsize,0) -- (0.5\hsize,-2em) (0.5\hsize,-0.875em) node{\small $\mathfrak{p}_2$}  (0.5\hsize,-2em) -- (0,0);}}]1{[j:A]t}
    \infer2[\textsc{Cut}]{st}
\end{prooftree}\, ,
\]
where $i, j \in \{1,2\}, i\neq j$, and $h(\mathfrak{p}_1)>1, h(\mathfrak{p}_2)\geq 1$ .

\begin{enumerate}
\item $A=\texttt{\#}$ (the \textsc{\textsc{Cut}}-formula is a $\#$-formula):
\[
\begin{prooftree}
\hypo{~~~~~~~~~}
\infer[rule code={\hbox{\tikz
    \draw (0,0) -- (\hsize,0) -- (0.5\hsize,-2em) (0.5\hsize,-0.875em) node{\small $\mathfrak{p}_3$}  (0.5\hsize,-2em) -- (0,0);}}]1{\mathfrak{S}_1}
    \infer1[i\#]{[i:\texttt{\#}]s^\prime}
    \infer[rule code={\hbox{\tikz
    \draw (0,0) -- (\hsize,0) -- (0.5\hsize,-2em) (0.5\hsize,-0.875em) node{\small $\mathfrak{p}_4$}  (0.5\hsize,-2em) -- (0,0);}}]1{[i:\texttt{\#}]s}
    \hypo{~~~~~~~~~}
\infer[rule code={\hbox{\tikz
    \draw (0,0) -- (\hsize,0) -- (0.5\hsize,-2em) (0.5\hsize,-0.875em) node{\small $\mathfrak{p}_2$}  (0.5\hsize,-2em) -- (0,0);}}]1{[j:\texttt{\#}]t}
    \infer2[\textsc{Cut}]{st}
\end{prooftree}\, ,
\]
where $\mathfrak{S}_1$ is the set of premiss sequents of $i\#$.

Then, $\texttt{\#}$ must have been introduced at some point in $\mathfrak{p}_2$ using the $j\#$-rule. This is because \textsc{Id} is restricted to primitives, applications of $i\#$ by definition only yield new $\#$-formulae in the $i$th dimension, and $\mathfrak{p}_2$ is \textsc{\textsc{Cut}}-free by the specification of $\mathfrak{p}$.
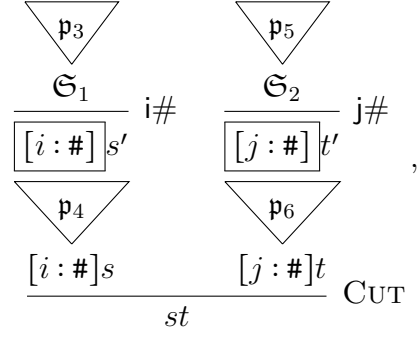
\[
\begin{prooftree}
\hypo{~~~~~~~~~}
\infer[rule code={\hbox{\tikz
    \draw (0,0) -- (\hsize,0) -- (0.5\hsize,-2em) (0.5\hsize,-0.875em) node{\small $\mathfrak{p}_3$}  (0.5\hsize,-2em) -- (0,0);}}]1{\mathfrak{S}_1}
    \infer1[i\#]{\boxed{[i:\texttt{\#}]}s^\prime}
    \infer[rule code={\hbox{\tikz
    \draw (0,0) -- (\hsize,0) -- (0.5\hsize,-2em) (0.5\hsize,-0.875em) node{\small $\mathfrak{p}_4$}  (0.5\hsize,-2em) -- (0,0);}}]1{[i:\texttt{\#}]s}
\hypo{~~~~~~~~~}
\infer[rule code={\hbox{\tikz
    \draw (0,0) -- (\hsize,0) -- (0.5\hsize,-2em) (0.5\hsize,-0.875em) node{\small $\mathfrak{p}_5$}  (0.5\hsize,-2em) -- (0,0);}}]1{\mathfrak{S}_2}
    \infer1[j\#]{\boxed{[j:\texttt{\#}]}t^{\prime}}
    \infer[rule code={\hbox{\tikz
    \draw (0,0) -- (\hsize,0) -- (0.5\hsize,-2em) (0.5\hsize,-0.875em) node{\small $\mathfrak{p}_6$}  (0.5\hsize,-2em) -- (0,0);}}]1{[j:\texttt{\#}]t}
    \infer2[\textsc{Cut}]{st}
\end{prooftree}\, ,
\]
where $\mathfrak{S}_2$ is the set of premiss sequents of $j\#$.

We can transform this proof into $\mathfrak{p}^\prime$:
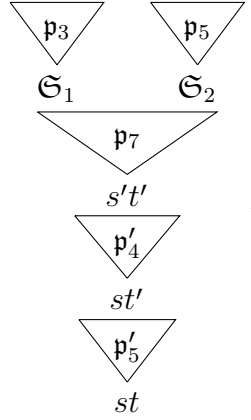
\[
\begin{prooftree}
    \hypo{~~~~~~~~~}
    \infer[rule code={\hbox{\tikz
    \draw (0,0) -- (\hsize,0) -- (0.5\hsize,-2em) (0.5\hsize,-0.875em) node{\small $\mathfrak{p}_3$}  (0.5\hsize,-2em) -- (0,0);}}]1{\mathfrak{S}_1}
    \hypo{~~~~~~~~~}
    \infer[rule code={\hbox{\tikz
    \draw (0,0) -- (\hsize,0) -- (0.5\hsize,-2em) (0.5\hsize,-0.875em) node{\small $\mathfrak{p}_5$}  (0.5\hsize,-2em) -- (0,0);}}]1{\mathfrak{S}_2}
    \infer[rule code={\hbox{\tikz
    \draw (0,0) -- (\hsize,0) -- (0.5\hsize,-2em) (0.5\hsize,-0.875em) node{\small $\mathfrak{p}_7$}  (0.5\hsize,-2em) -- (0,0);}}]2{~~~s^{\prime}t^{\prime}~~~}
    \infer[rule code={\hbox{\tikz
    \draw (0,0) -- (\hsize,0) -- (0.5\hsize,-2em) (0.5\hsize,-0.875em) node{\small $\mathfrak{p}_4^\prime$}  (0.5\hsize,-2em) -- (0,0);}}]1{~~~st^\prime~~~}
    \infer[rule code={\hbox{\tikz
    \draw (0,0) -- (\hsize,0) -- (0.5\hsize,-2em) (0.5\hsize,-0.875em) node{\small $\mathfrak{p}_5^\prime$}  (0.5\hsize,-2em) -- (0,0);}}]1{st}
\end{prooftree}\, .
\]
We construct $\mathfrak{p}_7$ as specified in Lemma 1.

We construct $\mathfrak{p}_4^\prime$ in parallel to $\mathfrak{p}_4$. We copy each step of $\mathfrak{p}_4$, except that we replace the boxed occurrence of sub-sequent `$[i:\texttt{\#}]$' with `$t^\prime$' in any sequent $\mathfrak{s}$ derived from the displayed sequent $[i:\texttt{\#}]s^\prime$. This replacement does not affect the derivation in any relevant way, regardless of the rule we applied to $\mathfrak{s}$ at step $m$ of $\mathfrak{p}_4$. Let us prove this claim.

First, note that the boxed occurrence of `$[i:\texttt{\#}]$' was introduced one step prior to the first step of $\mathfrak{p_4}$. Thus, it only occurs in the context of any sequent in $\mathfrak{p}_4$.

As we specified that a rule was applied \textit{to} $\mathfrak{s}$ at step $m$, $\mathfrak{s}$ is in the set of premiss sequents of said rule application. Hence, we did not apply an axiom. Further, $\mathfrak{p_4}$ is \textsc{Cut}-free by the specification of $\mathfrak{p}$.

Therefore, we applied a non-axiomatic operational rule at $m$. Call it $k\#$ with $k\in \{1,2\}$. Let $u, v$ be the sub-sequent(s) containing exclusively the active formulae of $k\in \{1,2\}$. Let $\mathfrak{p}_m, \mathfrak{p}_n (\mathfrak{p}^\prime_m, \mathfrak{p}^\prime_n)$ be the subproof(s) of $\mathfrak{p}_4 (\mathfrak{p}^\prime_4)$ up to step $m$.
\begin{enumerate}
\item $k\#$ is a single-premiss rule. Then, $\mathfrak{p}_m$ is of the following form in $\mathfrak{p}_4$:
\[
\begin{prooftree}
\hypo{~~~~~~~~~}
\infer[rule code={\hbox{\tikz
    \draw (0,0) -- (\hsize,0) -- (0.5\hsize,-2em) (0.5\hsize,-0.875em) node{\small $\mathfrak{p}_m$}  (0.5\hsize,-2em) -- (0,0);}}]1{[i:\texttt{\#}]s^{\prime\prime}u}
\infer1[k\#]{[i:\texttt{\#}]s^{\prime\prime}}
\end{prooftree}\, .
\]
In $\mathfrak{p^\prime_4}$, we replace this step with:
\[
\begin{prooftree}
\hypo{~~~~~~~~~}
\infer[rule code={\hbox{\tikz
    \draw (0,0) -- (\hsize,0) -- (0.5\hsize,-2em) (0.5\hsize,-0.875em) node{\small $\mathfrak{p}^\prime_m$}  (0.5\hsize,-2em) -- (0,0);}}]1{s^{\prime\prime}t^{\prime}u}
\infer1[k\#]{~~s^{\prime\prime}t^{\prime}~~}
\end{prooftree}\, .
\]
\item $k\#$ is a context-differentiating two-premiss rule. We obtain for $\mathfrak{p}_4$:
\[
\begin{prooftree}
\hypo{~~~~~~~~~}
\infer[rule code={\hbox{\tikz
    \draw (0,0) -- (\hsize,0) -- (0.5\hsize,-2em) (0.5\hsize,-0.875em) node{\small $\mathfrak{p}_m$}  (0.5\hsize,-2em) -- (0,0);}}]1{[i:\texttt{\#}]s^{\prime\prime}u}
    \hypo{~~~~~~~~~}
\infer[rule code={\hbox{\tikz
    \draw (0,0) -- (\hsize,0) -- (0.5\hsize,-2em) (0.5\hsize,-0.875em) node{\small $\mathfrak{p}_n$}  (0.5\hsize,-2em) -- (0,0);}}]1{s^{\prime\prime\prime}v}
\infer2[k\#]{[i:\texttt{\#}]s^{\prime\prime}s^{\prime\prime\prime}}
\end{prooftree}\, .
\] 
In $\mathfrak{p_4^\prime}$, we replace this with:
\[
\begin{prooftree}
\hypo{~~~~~~~~~}
\infer[rule code={\hbox{\tikz
    \draw (0,0) -- (\hsize,0) -- (0.5\hsize,-2em) (0.5\hsize,-0.875em) node{\small $\mathfrak{p}^\prime_m$}  (0.5\hsize,-2em) -- (0,0);}}]1{s^{\prime\prime}t^{\prime}u}
    \hypo{~~~~~~~~~}
\infer[rule code={\hbox{\tikz
    \draw (0,0) -- (\hsize,0) -- (0.5\hsize,-2em) (0.5\hsize,-0.875em) node{\small $\mathfrak{p}_n$}  (0.5\hsize,-2em) -- (0,0);}}]1{s^{\prime\prime\prime}v}
\infer2[k\#]{s^{\prime\prime}s^{\prime\prime\prime}t^{\prime}}
\end{prooftree}\, .
\]

\item $k\#$ is a context-sharing two-premiss rule. In $\mathfrak{p}_4$, we get:
\[
\begin{prooftree}
\hypo{~~~~~~~~~}
\infer[rule code={\hbox{\tikz
    \draw (0,0) -- (\hsize,0) -- (0.5\hsize,-2em) (0.5\hsize,-0.875em) node{\small $\mathfrak{p}_m$}  (0.5\hsize,-2em) -- (0,0);}}]1{\boxed{[i:\texttt{\#}]}s^{\prime\prime}u}
\hypo{~~~~~~~~~}
\infer[rule code={\hbox{\tikz
    \draw (0,0) -- (\hsize,0) -- (0.5\hsize,-2em) (0.5\hsize,-0.875em) node{\small $\mathfrak{p}_n$}  (0.5\hsize,-2em) -- (0,0);}}]1{[i:\texttt{\#}]s^{\prime\prime}v}
\infer2[k\#]{[i:\texttt{\#}]s^{\prime\prime}}
\end{prooftree}\, .
\]
Here, we need a more sophisticated construction for $\mathfrak{p}_4^\prime$. Assume the boxed `$[i:\texttt{\#}]$' in the left branch is our tracked occurrence. Then, the occurrence of `$[i:\texttt{\#}]$' must have been introduced from a set of sequents $\mathfrak{S_3}$ using an application of $i\#$ at some prior step in $\mathfrak{p_n}$: 
\[
\begin{prooftree}
\hypo{~~~~~~~~~}
\infer[rule code={\hbox{\tikz
    \draw (0,0) -- (\hsize,0) -- (0.5\hsize,-2em) (0.5\hsize,-0.875em) node{\small $\mathfrak{p}_m$}  (0.5\hsize,-2em) -- (0,0);}}]1{[i:\texttt{\#}]s^{\prime\prime}u}
\hypo{~~~~~~~~~}
\infer[rule code={\hbox{\tikz
    \draw (0,0) -- (\hsize,0) -- (0.5\hsize,-2em) (0.5\hsize,-0.875em) node{\small $\mathfrak{p}_o$}  (0.5\hsize,-2em) -- (0,0);}}]1{\mathfrak{S_3}}
    \infer1[i\#]{[i:\texttt{\#}]s^{\prime\prime\prime}}
    \infer[rule code={\hbox{\tikz
    \draw (0,0) -- (\hsize,0) -- (0.5\hsize,-2em) (0.5\hsize,-0.875em) node{\small $\mathfrak{p}_p$}  (0.5\hsize,-2em) -- (0,0);}}]1{[i:\texttt{\#}]s^{\prime\prime}v}
\infer2[k\#]{[i:\texttt{\#}]s^{\prime\prime}}
\end{prooftree}\, .
\]
We can construct the step in $p_4^\prime$ as follows:
\[
\begin{prooftree}
    \hypo{~~~~~~~~~}
    \infer[rule code={\hbox{\tikz
    \draw (0,0) -- (\hsize,0) -- (0.5\hsize,-2em) (0.5\hsize,-0.875em) node{\small $\mathfrak{p}^\prime_m$}  (0.5\hsize,-2em) -- (0,0);}}]1{s^{\prime\prime}t^{\prime}u}
    \hypo{~~~~~~~~~}
    \infer[rule code={\hbox{\tikz
    \draw (0,0) -- (\hsize,0) -- (0.5\hsize,-2em) (0.5\hsize,-0.875em) node{\small $\mathfrak{p}_5$}  (0.5\hsize,-2em) -- (0,0);}}]1{\mathfrak{S}_2}
    \hypo{~~~~~~~~~}
    \infer[rule code={\hbox{\tikz
    \draw (0,0) -- (\hsize,0) -- (0.5\hsize,-2em) (0.5\hsize,-0.875em) node{\small $\mathfrak{p}_o$}  (0.5\hsize,-2em) -- (0,0);}}]1{\mathfrak{S}_3}
    \infer[rule code={\hbox{\tikz
    \draw (0,0) -- (\hsize,0) -- (0.5\hsize,-2em) (0.5\hsize,-0.875em) node{\small $\mathfrak{p}_7^\prime$}  (0.5\hsize,-2em) -- (0,0);}}]2{s^{\prime\prime\prime}t^\prime}
    \infer[rule code={\hbox{\tikz
    \draw (0,0) -- (\hsize,0) -- (0.5\hsize,-2em) (0.5\hsize,-0.875em) node{\small $\mathfrak{p}_p^\prime$}  (0.5\hsize,-2em) -- (0,0);}}]1{~~s^{\prime\prime}t^{\prime}v~~}
    \infer2[k\#]{st^{\prime\prime}}
\end{prooftree}\, ,
\]
where $\mathfrak{p}_7^\prime$ is constructed as per the relevant reduction case of Lemma 1. We construct $\mathfrak{p}_5^\prime$ from $\mathfrak{p}_5$ analogously to the construction of $\mathfrak{p}_4^\prime$, replacing the boxed occurrence of sub-sequent `$[j:\texttt{\#}]$' with `$s$'. Intermediate steps follow 2ai-iii, using $\mathfrak{p}_3$ in lieu of $\mathfrak{p}_5$ in 2aiii. Proof $\mathfrak{p^\prime}$ is generally \textsc{Cut}-free, except if $\mathfrak{p}_7$ and $\mathfrak{p}_7^\prime$ (and the analogues in $\mathfrak{p}_5^\prime$) require \textsc{Cut} according to the relevant reduction case in Lemma 1. However, in each reduction case of Lemma 1, the \textsc{Cut}-depth is reduced by 1 compared to the unreduced case as we cut over the active formulae rather than the principal formula. Thus, $\mathcal{D}(\mathfrak{p^\prime})<\mathcal{D}(\mathfrak{p})$ and the theorem holds by the inductive hypothesis.

\end{enumerate}
\item $A \neq \texttt{\#}$ (the \textsc{Cut}-formula is not a \#-formula): As $h(\mathfrak{p_1})>1$, the last rule applied in $\mathfrak{p_1}$ must have been a (non-axiomatic) \#-rule since, if it was an axiom, $h(\mathfrak{p_1})=1$ and $\mathfrak{p_1}$ is \textsc{Cut}-free by the specification of $\mathfrak{p}$:
\[
\begin{prooftree}
\hypo{~~~~~~~~~}
\infer[rule code={\hbox{\tikz
    \draw (0,0) -- (\hsize,0) -- (0.5\hsize,-2em) (0.5\hsize,-0.875em) node{\small $\mathfrak{p}_3$}  (0.5\hsize,-2em) -- (0,0);}}]1{\mathfrak{S}_1}
    \infer1[n\#]{[i:A; n:\texttt{\#}]s}
    \hypo{~~~~~~~~~}
\infer[rule code={\hbox{\tikz
    \draw (0,0) -- (\hsize,0) -- (0.5\hsize,-2em) (0.5\hsize,-0.875em) node{\small $\mathfrak{p}_2$}  (0.5\hsize,-2em) -- (0,0);}}]1{[j:A]t}
    \infer2[\textsc{Cut}]{[n:\texttt{\#}]st}
\end{prooftree}\, ,
\]
where $h(\mathfrak{p}_{3})=h(\mathfrak{p}_1)-1, \mathfrak{S}_1$ is the set of premiss sequents of n\#, and $n \in \{1,2\}$. Let $v,w$ be the sub-sequents only containing the active formula(e) of n\# in their respective dimensions.
\begin{enumerate}
\item n\# is a single-premiss rule:
\[
\begin{prooftree}
\hypo{~~~~~~~~~}
\infer[rule code={\hbox{\tikz
    \draw (0,0) -- (\hsize,0) -- (0.5\hsize,-2em) (0.5\hsize,-0.875em) node{\small $\mathfrak{p}_3$}  (0.5\hsize,-2em) -- (0,0);}}]1{[i:A]sv}
    \infer1[n\#]{[i:A; n:\texttt{\#}]s}
    \hypo{~~~~~~~~~}
\infer[rule code={\hbox{\tikz
    \draw (0,0) -- (\hsize,0) -- (0.5\hsize,-2em) (0.5\hsize,-0.875em) node{\small $\mathfrak{p}_2$}  (0.5\hsize,-2em) -- (0,0);}}]1{[j:A]t}
    \infer2[\textsc{Cut}]{[n:\texttt{\#}]st}
\end{prooftree}
\]
 We can transform this proof into $\mathfrak{p}^\prime$, with $\mathcal{H}(\mathfrak{p}^\prime)<\mathcal{H}(\mathfrak{p})$ by pushing the \textsc{Cut} one step up:
\[
\begin{prooftree}
\hypo{~~~~~~~~~}
\infer[rule code={\hbox{\tikz
    \draw (0,0) -- (\hsize,0) -- (0.5\hsize,-2em) (0.5\hsize,-0.875em) node{\small $\mathfrak{p}_3$}  (0.5\hsize,-2em) -- (0,0);}}]1{[i:A]sv}
    \hypo{~~~~~~~~~}
\infer[rule code={\hbox{\tikz
    \draw (0,0) -- (\hsize,0) -- (0.5\hsize,-2em) (0.5\hsize,-0.875em) node{\small $\mathfrak{p}_2$}  (0.5\hsize,-2em) -- (0,0);}}]1{[j:A]t}
    \infer2[\textsc{Cut}]{stv}
     \infer1[n\#]{[n:\texttt{\#}]st}
\end{prooftree}\, .
\]
\item n\# is a context-differentiating two-premiss rule:
\[
\begin{prooftree}
\hypo{~~~~~~~~~}
\infer[rule code={\hbox{\tikz
    \draw (0,0) -- (\hsize,0) -- (0.5\hsize,-2em) (0.5\hsize,-0.875em) node{\small $\mathfrak{p}_3$}  (0.5\hsize,-2em) -- (0,0);}}]1{[i:A]sv}
\hypo{~~~~~~~~~}
    \infer[rule code={\hbox{\tikz
    \draw (0,0) -- (\hsize,0) -- (0.5\hsize,-2em) (0.5\hsize,-0.875em) node{\small $\mathfrak{p}_4$}  (0.5\hsize,-2em) -- (0,0);}}]1{uw}
    \infer2[n\#]{[i:A; n:\texttt{\#}]su}
    \hypo{~~~~~~~~~}
\infer[rule code={\hbox{\tikz
    \draw (0,0) -- (\hsize,0) -- (0.5\hsize,-2em) (0.5\hsize,-0.875em) node{\small $\mathfrak{p}_2$}  (0.5\hsize,-2em) -- (0,0);}}]1{[j:A]t}
    \infer2[\textsc{Cut}]{[n:\texttt{\#}]stu}
\end{prooftree}\, .
\]
We can transform this proof into $\mathfrak{p}^\prime$, with $\mathcal{H}(\mathfrak{p}^\prime)<\mathcal{H}(\mathfrak{p})$ analogously to the previous step:
\[
\begin{prooftree}
\hypo{~~~~~~~~~}
\infer[rule code={\hbox{\tikz
    \draw (0,0) -- (\hsize,0) -- (0.5\hsize,-2em) (0.5\hsize,-0.875em) node{\small $\mathfrak{p}_3$}  (0.5\hsize,-2em) -- (0,0);}}]1{[i:A]sv}
    \hypo{~~~~~~~~~}
\infer[rule code={\hbox{\tikz
    \draw (0,0) -- (\hsize,0) -- (0.5\hsize,-2em) (0.5\hsize,-0.875em) node{\small $\mathfrak{p}_2$}  (0.5\hsize,-2em) -- (0,0);}}]1{[j:A]t}
    \infer2[\textsc{Cut}]{stv}
\hypo{~~~~~~~~~}
    \infer[rule code={\hbox{\tikz
    \draw (0,0) -- (\hsize,0) -- (0.5\hsize,-2em) (0.5\hsize,-0.875em) node{\small $\mathfrak{p}_4$}  (0.5\hsize,-2em) -- (0,0);}}]1{uw}
    \infer2[n\#]{[n:\texttt{\#}]stu}
\end{prooftree}\, .
\]
\item n\# is a context-sharing two-premiss rule:
\[
\begin{prooftree}
\hypo{~~~~~~~~~}
\infer[rule code={\hbox{\tikz
    \draw (0,0) -- (\hsize,0) -- (0.5\hsize,-2em) (0.5\hsize,-0.875em) node{\small $\mathfrak{p}_3$}  (0.5\hsize,-2em) -- (0,0);}}]1{[i:A]sv}
\hypo{~~~~~~~~~}
    \infer[rule code={\hbox{\tikz
    \draw (0,0) -- (\hsize,0) -- (0.5\hsize,-2em) (0.5\hsize,-0.875em) node{\small $\mathfrak{p}_4$}  (0.5\hsize,-2em) -- (0,0);}}]1{[i:A]sw}
    \infer2[n\#]{[i:A; n:\texttt{\#}]s}
    \hypo{~~~~~~~~~}
\infer[rule code={\hbox{\tikz
    \draw (0,0) -- (\hsize,0) -- (0.5\hsize,-2em) (0.5\hsize,-0.875em) node{\small $\mathfrak{p}_2$}  (0.5\hsize,-2em) -- (0,0);}}]1{[j:A]t}
    \infer2[\textsc{Cut}]{[n:\texttt{\#}]st}
\end{prooftree}\, .
\]
We can again transform this proof into $\mathfrak{p}^\prime$, with $\mathcal{H}(\mathfrak{p}^\prime)<\mathcal{H}(\mathfrak{p})$:
\[
\begin{prooftree}
\hypo{~~~~~~~~~}
\infer[rule code={\hbox{\tikz
    \draw (0,0) -- (\hsize,0) -- (0.5\hsize,-2em) (0.5\hsize,-0.875em) node{\small $\mathfrak{p}_3$}  (0.5\hsize,-2em) -- (0,0);}}]1{[i:A]sv}
    \hypo{~~~~~~~~~}
\infer[rule code={\hbox{\tikz
    \draw (0,0) -- (\hsize,0) -- (0.5\hsize,-2em) (0.5\hsize,-0.875em) node{\small $\mathfrak{p}_2$}  (0.5\hsize,-2em) -- (0,0);}}]1{[j:A]t}
    \infer2[\textsc{Cut}]{stv}
\hypo{~~~~~~~~~}
    \infer[rule code={\hbox{\tikz
    \draw (0,0) -- (\hsize,0) -- (0.5\hsize,-2em) (0.5\hsize,-0.875em) node{\small $\mathfrak{p}_4$}  (0.5\hsize,-2em) -- (0,0);}}]1{[i:A]sw}
        \hypo{~~~~~~~~~}
\infer[rule code={\hbox{\tikz
    \draw (0,0) -- (\hsize,0) -- (0.5\hsize,-2em) (0.5\hsize,-0.875em) node{\small $\mathfrak{p}_2$}  (0.5\hsize,-2em) -- (0,0);}}]1{[j:A]t}
    \infer2[\textsc{Cut}]{stw}
    \infer2[n\#]{[n:\texttt{\#}]st}
\end{prooftree} \, .
\]
\end{enumerate}
In all three cases, the theorem holds by virtue of the IH. \hfill $\blacksquare$
\end{enumerate}
\end{enumerate}
\end{proof}
We can now prove the first tenet of our notion of definability, \textit{conservativity}.

\begin{theorem}[conservativity]\

    \noindent Let $\textbf{Cal}=\{\textsc{Id}, \textsc{Cut} cd\}$ and let $\textbf{Cal}^+=\textbf{Cal}\cup \{1\#,2\#\}$.
    
    For any $\#\in \mathtt{PC}$,
        $\textbf{Cal}^+$ conservatively extends \textbf{Cal}.
\end{theorem}

\begin{proof}\

\noindent
    By Definition 15, to show that $\textbf{Cal}^+$ conservatively extends \textbf{Cal}, we must demonstrate that for any sequent $\mathfrak{s}$ in which `$\#$' does not occur and for any proof $\mathfrak{p}$ of $\mathfrak{s}$ in $\textbf{Cal}^+$, there is a proof $\mathfrak{p^\prime}$ of $\mathfrak{s}$ in \textbf{Cal}.
    
    We prove this via an exhaustive inverted proof search on $\mathfrak{p}$. The final proof step in $\mathfrak{p}$ must have been an application of a \#-rule, \textsc{Id}, or \textsc{\textsc{Cut}}.
    \begin{enumerate}
        \item It could not have been a \#-rule for if it was, $\mathfrak{s}$ would contain `\#' by the specification of operational rules, which it cannot by the definition of $\mathfrak{s}$.
        \item Assume that the final step was an application of \textsc{Id}. Thus, $h(\mathfrak{p})=1$, i.e. this application of \textsc{Id} is the only rule used in $\mathfrak{p}$. As $\textsc{Id}\in \textbf{Cal}$ and \textsc{Id} is restricted to primitive formulae, $\mathfrak{p}^\prime=\mathfrak{p}$ and the theorem holds.
        \item Assume it was a \textsc{Cut}-application. Then, by virtue of Theorem 2, there is a \textsc{Cut}-free proof $\mathfrak{p}^{\prime\prime}$ of $\mathfrak{s}$ in $\textbf{Cal}^+$.
        The final step of $\mathfrak{p}^{\prime\prime}$ could not have been an application of a \#-rule (cf. case 1). Therefore, it was an application of \textsc{Id}. Hence, $\mathfrak{p}^{\prime}=\mathfrak{p}^{\prime\prime}$ and the theorem holds (cf. case 2). \hfill $\blacksquare$
    \end{enumerate} 
\end{proof}

\subsubsection{Negative cases}

Our set of positive cases, \texttt{PC}, is exhaustive, i.e. its members are the only connectives with operational rules formulable in two-dimensional sequent calculi containing at most two premiss sequents and two active formulae which are conservative over $\{\textsc{Id}, \textsc{Cut} cd\}$. The non-conservativity of most other cases is a consequence of the following lemma. 
\begin{lemma}\

    \noindent For any sequent $\mathfrak{s}$, if $\mathfrak{s}$ is provable in $\{\textsc{Id}, \textsc{Cut} cd\}$, $\mathfrak{s}$ is an instance of \textsc{Id}$^{pr}$.
\end{lemma}
\begin{proof}\

    \noindent By induction on the proof height $h(\mathfrak{p})$. Let $h(\mathfrak{p})=1$. The last rule used is \textsc{Id}$^{pr}$ and the lemma holds trivially. Let $h(\mathfrak{p})>1$. The last rule used is \textsc{Cut}, each of the premises of which are the conclusion of a proof of lower height than $\mathfrak{p}$. The lemma holds by virtue of the inductive hypothesis.
    \hfill $\blacksquare$
\end{proof}
We can now prove that only the $\#\in \mathtt{PC}$ satisfy conservativity, among all candidate connectives defined in \S 4  .
\begin{theorem}\

\noindent Theorem 3 does not hold for any $\#\not{\in} \mathtt{PC}$, where \# is some connective whose operational rules are formulable in two-dimensional sequent calculi using at most two premiss sequents and at most two active formulae.

\end{theorem}
\begin{proof}\

\noindent
To give a counterexample to Theorem 3, we must show that there is a sequent $\mathfrak{s}$ in which \# does not occur and a proof $\mathfrak{p}$ of $\mathfrak{s}$ in \textbf{Cal}$^+$ but no proof $\mathfrak{p^\prime}$ of $\mathfrak{s}$ in \textbf{Cal}.
        
\textit{Proof idea:} for each specification of $\# \not{\in} \mathtt{PC}$, we attempt to prove a result akin to Lemma 1. We then use failure of this attempt to construct a counterexample to Theorem 3. Most instances are immediate consequences of Lemma 2. We give a general strategy for how we can give counterexamples for these in cases 1. and 2. In case 3., we give direct counterexamples to all non-immediate instances not covered in 1. or 2.

\textit{Notation:} we write `$n(A, i, \mathfrak{s})$' for the number of occurrences of formula $A$ in the $i$th dimension of sequent $\mathfrak{s}$. Further, for any set of sequents $\mathfrak{S}$, $n(A, i, \mathfrak{S})=\sum_{\mathfrak{s}\in\mathfrak{S}}n(A, i, \mathfrak{s}).$

\textit{Proof:} Let connective \# have the operational rules 
\[
\Bigg \langle
\begin{prooftree}
    \hypo{\mathfrak{S_1}}
    \infer1[1\#]{[i:\texttt{\#}]s}
\end{prooftree},
\begin{prooftree}
    \hypo{\mathfrak{S_2}}
    \infer1[2\#]{[j:\texttt{\#}]t}
\end{prooftree}
\Bigg \rangle\, ,
\]
where $\mathfrak{S_1}, \mathfrak{S_2}$ are sets of (premiss) sequents. 

\begin{enumerate}
    \item Consider only the instances of \# such that $1\#$ or $2\#$ does \textit{not} have balanced active formula occurrence numbers across dimensions. That is, for any $\# \not{\in} \mathtt{PC}$ and any non-empty $\mathfrak{S}^{\prime}\subseteq \mathfrak{S}_1\cup \mathfrak{S_2}$, there is an active formula $A$ of $1\#$ or $2\#$ with $n(A,i,\mathfrak{S}^{\prime})\neq n(A,j,\mathfrak{S}^{\prime})$. In other words, we cannot form a subset of premiss sequents with equal occurrence numbers per dimension for each active formula.

 As our counterexample, we show that there is a \textbf{Cal}$^+$-proof $\mathfrak{p}$,
\[
\begin{prooftree}
    \hypo{~~~~~~~~~}
            \infer[rule code={\hbox{\tikz\draw (0,0) -- (\hsize,0) -- (0.5\hsize,-2em) (0.5\hsize,-0.875em) node{\small $\mathfrak{p}_1$}  (0.5\hsize,-2em) -- (0,0);}}]1{\mathfrak{S_1}}
    \infer1[1\#]{[i:\texttt{\#}]s}
    \hypo{~~~~~~~~~}
            \infer[rule code={\hbox{\tikz\draw (0,0) -- (\hsize,0) -- (0.5\hsize,-2em) (0.5\hsize,-0.875em) node{\small $\mathfrak{p}_2$}  (0.5\hsize,-2em) -- (0,0);}}]1{\mathfrak{S_2}}
    \infer1[2\#]{[j:\texttt{\#}]t}
    \infer2[\textsc{Cut}]{st}
\end{prooftree},
\]
with $\mathfrak{p_1},\mathfrak{p_2}$ not containing $\#$, with no corresponding \#-free proof in the base system; i.e. there is no \textbf{Cal}-proof $\mathfrak{p}^\prime$ of $st$.

Given our specification of $\#$, there must be some active formula $A$ of $1\#$ or $2\#$ in $\mathfrak{p}$ such that $n(A,i,\mathfrak{S_1}\cup \mathfrak{S_2})\neq n(A,j,\mathfrak{S_1}\cup \mathfrak{S_2})$. Since $\#$ does not occur in $\mathfrak{p_1}$ or $\mathfrak{p_2}$, nor do $1\#$ or $2\#$. Hence, $\mathfrak{p_1}, \mathfrak{p_2}$ are \textbf{Cal}-proofs. By Lemma 2, $\mathfrak{S_1}$ and $ \mathfrak{S_2}$ only contain instances of \textsc{Id}. Hence for any formula $C$ occurring in $\mathfrak{S_1}\cup \mathfrak{S_2}$, $n(C,i,\mathfrak{S_1}\cup \mathfrak{S_2})=n(C,j,\mathfrak{S_1}\cup \mathfrak{S_2})$. All occurrences in $\mathfrak{S_1}$ and $\mathfrak{S_2}$ are either active or contextual. Therefore, there must be at least one more \textit{contextual} occurrence of our $A$ in one dimension than in the other, i.e. $n(A,i,\{s,t\})\neq n(A,j,\{s,t\})$. Hence, there is an $A$ in $st$ such that $n(A,i,st)\neq n(A,j,st)$.

Now suppose there was a $\mathfrak{p}^\prime$. Then, $st$ is provable in \textbf{Cal}. Thus by Lemma 2, $st$ is an instance of \textsc{Id}. By the definition of \textsc{Id}, for any formula $B$ occurring in $st$, $n(B,i,st)= n(B,j,st)$. This is a contradiction.

Let us illustrate this result with an example. Consider the connectives with premiss type $\langle \gamma iF, \delta jFiG \rangle$ for $i,j\in\{1,2\}, i\neq j$, $F,G\in \{A, B\}$. Then, $[j:A]$ (for primitive $A$) is not provable in the base system of $\{\textsc{Id}^{pr}, \textsc{Cut}cd\}$ (cf. Lemma 2) but becomes provable after adding rules for $\#$ to the language and applying them to instances of \textsc{Id} as per our proof strategy:
\[
\begin{prooftree}
  \hypo{[i:A;j:A]}
  \infer1[1\#]{[1:\texttt{\#};j:A]}
  \hypo{[i:A;j:A]}
  \infer1[2\#]{[2:\texttt{\#}]}
  \infer2[\textsc{Cut}]{[j:A]}
\end{prooftree}\, .
\]
\item Consider only the instances of \# such that $1\#$ and $2\#$ \textit{do} have balanced active formula occurrence numbers across dimensions, but exclusively in subsets of the premiss sequents which contain \textit{no active formulae}. Further, $st$ is not in said subsets. That is, for any $\# \not{\in} \mathtt{PC}$, for any non-empty $\mathfrak{S}^{\prime}\subseteq \mathfrak{S}_1\cup \mathfrak{S_2}$, and for any active formula $A$ of $1\#$ or $2\#$ with $n(A,i,\mathfrak{S}^{\prime})= n(A,j,\mathfrak{S}^{\prime})=0$, $st\not{\in}\mathfrak{S}^{\prime}$.

We use the same strategy as in the previous case to construct a counterexample. Given the previous case, we can disregard all cases with unbalanced active formula occurrence numbers. As $s$ and $t$ do not contain active formulae, they must be instances of \textsc{Id} by our specification of $\mathfrak{p_1, p_2}$, i.e. of the form $[1:A;2:A]$, for any primitive $A$. Then, $st$ is of the form $[1:A,B;2:A,B]$ given the definition of our sequent multiplication $\times$ and, thus, not an instance of \textsc{Id}. However, $st$ must be an instance of \textsc{Id} by virtue of Lemma 2. This is a contradiction.

Example: consider the connectives with premiss type $\langle \gamma 0, \gamma iF \rangle$ for $i,j\in\{1,2\}, i\neq j$, $F,G\in \{A, B\}$. From 
\[
\begin{prooftree}
  \hypo{~~~~~~~~~}
            \infer[rule code={\hbox{\tikz\draw (0,0) -- (\hsize,0) -- (0.5\hsize,-2em) (0.5\hsize,-0.875em) node{\small $\mathfrak{p}_1$}  (0.5\hsize,-2em) -- (0,0);}}]1{s}
\infer1[i\#]{[1:\texttt{\#}]s}
  \hypo{~~~~~~~~~}
            \infer[rule code={\hbox{\tikz\draw (0,0) -- (\hsize,0) -- (0.5\hsize,-2em) (0.5\hsize,-0.875em) node{\small $\mathfrak{p}_1$}  (0.5\hsize,-2em) -- (0,0);}}]1{[i:F]t}
\infer1[j\#]{[2:\texttt{\#}]t}
\infer2[\textsc{Cut}]{st}
\end{prooftree}\, ,
\]
we construct
\[
\begin{prooftree}
  \hypo{\emptyset}
            \infer1{[i:A;j:A]}
\infer1[i\#]{[1:\texttt{\#};i:A;j:A]}
  \hypo{\emptyset}
            \infer1{[i:B;j:B]t}
\infer1[j\#]{[2:\texttt{\#};j:B]}
\infer2[\textsc{Cut}]{[i:A;j:A,B]}.
\end{prooftree}\, .
\]

\item We provide direct counterexamples for all remaining instances of \# not covered by 1. or 2. (i.e. no active formulae in $1\#,2\#$ or $n(A,i,\mathfrak{S}^{\prime})= n(A,j,\mathfrak{S}^{\prime})$). To do so, we provide a \textbf{Cal}$^+$-proof of a sequent that is unprovable in \textbf{Cal} (by virtue of Lemma 2) for each such instance. The enumeration of cases again follows that in Lemma 1. Let $i,j,k,l\in \{1,2\}$ and $i \neq j$. Let $\F, G\in \{A, B\}$.\\ \\
Case 1) $\texttt{P}(i\#)=\texttt{P}(j\#)=0$.
\begin{enumerate}[label=1.\arabic*)]
\item $\texttt{S}(i\#)=\alpha; \texttt{S}(j\#)=\alpha$:
\[
\begin{prooftree}
\hypo{\emptyset}
\infer1[i\#]{[i:\texttt{\#}]}
\hypo{\emptyset}
\infer1[j\#]{[j:\texttt{\#}]}
\infer2[\textsc{Cut}]{\emptyset\vdash\emptyset}
\end{prooftree}\, .
\]
\end{enumerate}
Case 2) $\texttt{P}(i\#)=0; \texttt{P}(j\#)=kF$: all cases not already covered in Lemma 1 fall under step 1. or 2. of this proof.\\ \\
Case 3) $\texttt{P}(i\#)=0; \texttt{P}(j\#)=kFlG$. We will only consider cases where $k\neq l$ and $F=G$, as all other negative cases fall under step 1. or 2. of this proof.
\begin{enumerate}[label=3.\arabic*)]
    \item $\texttt{S}(i\#)=\alpha; \texttt{S}(j\#)=\beta$: cf. 1.1).
    \item $\texttt{S}(i\#)=\alpha; \texttt{S}(j\#)=\gamma$: from
    \[
    \begin{prooftree}
        \hypo{\emptyset}
        \infer1[i\#]{[i:\texttt{\#}]}
        \hypo{~~~~~~~~~}
            \infer[rule code={\hbox{\tikz\draw (0,0) -- (\hsize,0) -- (0.5\hsize,-2em) (0.5\hsize,-0.875em) node{\small $\mathfrak{p}_1$}  (0.5\hsize,-2em) -- (0,0);}}]1{[k:F;l:F]s}
        \infer1[j\#]{[j:\texttt{\#}]s}
        \infer2[\textsc{Cut}]{s}
    \end{prooftree}\, ,
    \]
    we get
    \[
    \begin{prooftree}
        \hypo{\emptyset}
        \infer1[i\#]{[i:\texttt{\#}]}
        \hypo{\emptyset}
        \infer1[\textsc{Id}]{[k:A;l:A]}
        \infer1[j\#]{[j:\texttt{\#}]}
        \infer2[\textsc{Cut}]{\emptyset \vdash \emptyset}
    \end{prooftree}\, .
    \]
    \item $\texttt{S}(i\#)=\gamma; \texttt{S}(j\#)=\delta cd$:
    from
    \[
    \begin{prooftree}
        \hypo{~~~~~~~~~}
            \infer[rule code={\hbox{\tikz\draw (0,0) -- (\hsize,0) -- (0.5\hsize,-2em) (0.5\hsize,-0.875em) node{\small $\mathfrak{p}_1$}  (0.5\hsize,-2em) -- (0,0);}}]1{s}
        \infer1[i\#]{[i:\texttt{\#}]s}
        \hypo{~~~~~~~~~}
            \infer[rule code={\hbox{\tikz\draw (0,0) -- (\hsize,0) -- (0.5\hsize,-2em) (0.5\hsize,-0.875em) node{\small $\mathfrak{p}_2$}  (0.5\hsize,-2em) -- (0,0);}}]1{[k:F;l:F]t}
        \hypo{~~~~~~~~~}
        \infer[rule code={\hbox{\tikz\draw (0,0) -- (\hsize,0) -- (0.5\hsize,-2em) (0.5\hsize,-0.875em) node{\small $\mathfrak{p}_3$}  (0.5\hsize,-2em) -- (0,0);}}]1{u}
        \infer2[j\#]{[j:\texttt{\#}]tu}
        \infer2[\textsc{Cut}]{stu}
    \end{prooftree}\, ,
    \]
    we get
    \[
    \begin{prooftree}
        \hypo{\emptyset}
        \infer1[\textsc{Id}]{[k:A;l:A]}
        \infer1[i\#]{[i:\texttt{\#};k:A;l:A]}
        \hypo{\emptyset}
        \infer1[\textsc{Id}]{[k:B;l:B]}
        \hypo{\emptyset}
        \infer1[\textsc{Id}]{[k:C;l:C]}
        \infer2[j\#]{[j:\texttt{\#}; k:C;l:C]}
        \infer2[\textsc{Cut}]{[k:B,C;l:B,C]}
    \end{prooftree}\, .
    \]
\item $\texttt{S}(i\#)=\delta cs; \texttt{S}(j\#)=\delta cd$: cf. 3.4).
\item $\texttt{S}(i\#)=\delta cd; \texttt{S}(j\#)=\gamma$: from
\[
    \begin{prooftree}
        \hypo{~~~~~~~~~}
            \infer[rule code={\hbox{\tikz\draw (0,0) -- (\hsize,0) -- (0.5\hsize,-2em) (0.5\hsize,-0.875em) node{\small $\mathfrak{p}_1$}  (0.5\hsize,-2em) -- (0,0);}}]1{s}
                    \hypo{~~~~~~~~~}
            \infer[rule code={\hbox{\tikz\draw (0,0) -- (\hsize,0) -- (0.5\hsize,-2em) (0.5\hsize,-0.875em) node{\small $\mathfrak{p}_2$}  (0.5\hsize,-2em) -- (0,0);}}]1{t}
        \infer2[i\#]{[i:\texttt{\#}]st}
        \hypo{~~~~~~~~~}
        \infer[rule code={\hbox{\tikz\draw (0,0) -- (\hsize,0) -- (0.5\hsize,-2em) (0.5\hsize,-0.875em) node{\small $\mathfrak{p}_3$}  (0.5\hsize,-2em) -- (0,0);}}]1{[i:F;j:F]u}
        \infer1[j\#]{[j:\texttt{\#}]u}
        \infer2[\textsc{Cut}]{stu}
    \end{prooftree}\, ,
    \]
    we get
    \[
    \begin{prooftree}
        \hypo{\emptyset}
        \infer1[\textsc{Id}]{[k:A;l:A]}
        \hypo{\emptyset}
        \infer1[\textsc{Id}]{[k:B;l:B]}
        \infer2[i\#]{[i:\texttt{\#};k:A, B;l:A, B]}
        \hypo{\emptyset}
        \infer1[\textsc{Id}]{[k:C;l:C]}
        \infer1[j\#]{[j:\texttt{\#}]}
        \infer2[\textsc{Cut}]{[k:A,B;l:A,B]}
    \end{prooftree}\, .
    \]
\item $\texttt{S}(i\#)=\delta cd; \texttt{S}(j\#)=\delta cd$: cf. 3.4; 3.5).
\end{enumerate}
Case 4) $\texttt{P}(i\#)=0; \texttt{P}(j\#)=kF-lG$. We will only consider cases where $k\neq l$ and $F=G$, as all other negative cases fall under step 1. or 2. of this proof.
\begin{enumerate}[label=4.\arabic*)]
\item $\texttt{S}(i\#)=\gamma; \texttt{S}(j\#)=\delta cd$: from
\[
\begin{prooftree}
    \hypo{~~~~~~~~~}
            \infer[rule code={\hbox{\tikz\draw (0,0) -- (\hsize,0) -- (0.5\hsize,-2em) (0.5\hsize,-0.875em) node{\small $\mathfrak{p}_1$}  (0.5\hsize,-2em) -- (0,0);}}]1{s}
    \infer1[i\#]{[i:\texttt{\#}]s}
    \hypo{~~~~~~~~~}
            \infer[rule code={\hbox{\tikz\draw (0,0) -- (\hsize,0) -- (0.5\hsize,-2em) (0.5\hsize,-0.875em) node{\small $\mathfrak{p}_2$}  (0.5\hsize,-2em) -- (0,0);}}]1{[k:F]t}
    \hypo{~~~~~~~~~}
            \infer[rule code={\hbox{\tikz\draw (0,0) -- (\hsize,0) -- (0.5\hsize,-2em) (0.5\hsize,-0.875em) node{\small $\mathfrak{p}_3$}  (0.5\hsize,-2em) -- (0,0);}}]1{[l:F]u}
\infer2[j\#]{[j:\texttt{\#}]tu}
\infer2[\textsc{Cut}]{stu}
\end{prooftree}\, ,
\]
we get
\[
\begin{prooftree}
    \hypo{\emptyset}
    \infer1[\textsc{Id}]{[k:B;l:B]}
    \infer1[i\#]{[i:\texttt{\#};k:B;l:B]}
    \hypo{\emptyset}
    \infer1[\textsc{Id}]{[k:A;l:A]}
    \hypo{\emptyset}
    \infer1[\textsc{Id}]{[k:A;l:A]}
\infer2[j\#]{[j:\texttt{\#}, k:A;l:A]}
\infer2[\textsc{Cut}]{[k:A, B;l:A, B]}
\end{prooftree}\, .
\]
\item $\texttt{S}(i\#)=\delta cs; \texttt{S}(j\#)=\delta cd$: cf. 4.1).
\item $\texttt{S}(i\#)=\delta cd; \texttt{S}(j\#)=\delta cd$: cf. 4.1).
\end{enumerate}
Case 5) $\texttt{P}(i\#)=kF; \texttt{P}(j\#)=lG$. We will only consider cases where $k\neq l$ and $F=G$, as all other negative cases fall under step 1. or 2. of this proof.
\begin{enumerate}[label=5.\arabic*)]
\item $\texttt{S}(i\#)=\gamma; \texttt{S}(j\#)=\delta cd$: from
\[
\begin{prooftree}
  \hypo{~~~~~~~~~}
            \infer[rule code={\hbox{\tikz\draw (0,0) -- (\hsize,0) -- (0.5\hsize,-2em) (0.5\hsize,-0.875em) node{\small $\mathfrak{p}_1$}  (0.5\hsize,-2em) -- (0,0);}}]1{[k:F]s}
\infer1[i\#]{[i:\texttt{\#}]s}
  \hypo{~~~~~~~~~}
            \infer[rule code={\hbox{\tikz\draw (0,0) -- (\hsize,0) -- (0.5\hsize,-2em) (0.5\hsize,-0.875em) node{\small $\mathfrak{p}_2$}  (0.5\hsize,-2em) -- (0,0);}}]1{[l:F]t}
  \hypo{~~~~~~~~~}
            \infer[rule code={\hbox{\tikz\draw (0,0) -- (\hsize,0) -- (0.5\hsize,-2em) (0.5\hsize,-0.875em) node{\small $\mathfrak{p}_3$}  (0.5\hsize,-2em) -- (0,0);}}]1{u}
\infer2[j\#]{[j:\texttt{\#}]tu}
\infer2[\textsc{Cut}]{stu}
\end{prooftree}\, ,
\]
we get
\[
\begin{prooftree}
  \hypo{\emptyset}
            \infer1[\textsc{Id}]{[k:A;l:A]}
\infer1[i\#]{[i:\texttt{\#};l:A]}
  \hypo{\emptyset}
            \infer1[\textsc{Id}]{[k:A;l:A]}
  \hypo{\emptyset}
            \infer1[\textsc{Id}]{[k:B;l:B]}
\infer2[j\#]{[j:\texttt{\#}, k:A,B;l:B]}
\infer2[\textsc{Cut}]{[k:A,B;l:A,B]}
\end{prooftree}\, .
\]
\item $\texttt{S}(i\#)=\delta cd; \texttt{S}(j\#)=\delta cd$: from
\[
\begin{prooftree}
\hypo{~~~~~~~~~}
            \infer[rule code={\hbox{\tikz\draw (0,0) -- (\hsize,0) -- (0.5\hsize,-2em) (0.5\hsize,-0.875em) node{\small $\mathfrak{p}_1$}  (0.5\hsize,-2em) -- (0,0);}}]1{[k:F]s}
\hypo{~~~~~~~~~}
            \infer[rule code={\hbox{\tikz\draw (0,0) -- (\hsize,0) -- (0.5\hsize,-2em) (0.5\hsize,-0.875em) node{\small $\mathfrak{p}_2$}  (0.5\hsize,-2em) -- (0,0);}}]1{t}
\infer2[i\#]{[i:\texttt{\#}]st}
\hypo{~~~~~~~~~}
            \infer[rule code={\hbox{\tikz\draw (0,0) -- (\hsize,0) -- (0.5\hsize,-2em) (0.5\hsize,-0.875em) node{\small $\mathfrak{p}_3$}  (0.5\hsize,-2em) -- (0,0);}}]1{[l:F]u}
\hypo{~~~~~~~~~}
            \infer[rule code={\hbox{\tikz\draw (0,0) -- (\hsize,0) -- (0.5\hsize,-2em) (0.5\hsize,-0.875em) node{\small $\mathfrak{p}_4$}  (0.5\hsize,-2em) -- (0,0);}}]1{v}
\infer2[j\#]{[j:\texttt{\#}]uv}
\infer2[\textsc{Cut}]{stuv}
\end{prooftree}\, ,
\]
we get
\[
\begin{prooftree}
\hypo{\emptyset}
            \infer1[\textsc{Id}]{[k:A;l:A]}
\hypo{\emptyset}
            \infer1[\textsc{Id}]{[k:B;l:B]}
\infer2[i\#]{[i:\texttt{\#};k:B;l:A,B]}
\hypo{\emptyset}
            \infer1[\textsc{Id}]{[k:A;l:A]}
\hypo{\emptyset}
            \infer1[\textsc{Id}]{[k:C;l:C]}
\infer2[j\#]{[j:\texttt{\#};k:C,A;l:C]}
\infer2[\textsc{Cut}]{[k:A,B,C;l:A,B,C]}
\end{prooftree}\, .
\]
\end{enumerate}
Case 6) $\texttt{P}(i\#)=kF; \texttt{P}(j\#)=lF-mG$: all cases not already covered in Lemma 1 fall under step 1. or 2. of this proof.\\
Case 7) $\texttt{P}(i\#)=kF; \texttt{P}(j\#)=lF-mG$.
\begin{enumerate}[label=7.\arabic*)]
    \item $\texttt{S}(i\#)=\gamma; \texttt{S}(j\#)=\delta cd$: from
\[
    \begin{prooftree}
        \hypo{~~~~~~~~~}
            \infer[rule code={\hbox{\tikz\draw (0,0) -- (\hsize,0) -- (0.5\hsize,-2em) (0.5\hsize,-0.875em) node{\small $\mathfrak{p}_1$}  (0.5\hsize,-2em) -- (0,0);}}]1{[k:F]s}
\infer1[i\#]{[i:\texttt{\#}]s}
\hypo{~~~~~~~~~}
            \infer[rule code={\hbox{\tikz\draw (0,0) -- (\hsize,0) -- (0.5\hsize,-2em) (0.5\hsize,-0.875em) node{\small $\mathfrak{p}_2$}  (0.5\hsize,-2em) -- (0,0);}}]1{[l:F]t}
            \hypo{~~~~~~~~~}
            \infer[rule code={\hbox{\tikz\draw (0,0) -- (\hsize,0) -- (0.5\hsize,-2em) (0.5\hsize,-0.875em) node{\small $\mathfrak{p}_2$}  (0.5\hsize,-2em) -- (0,0);}}]1{[m:G]u}
\infer2[j\#]{[j:\texttt{\#}]tu}
\infer2[\textsc{Cut}]{stu}
\end{prooftree}\, ,
\]
we get for $k\neq l$ or $l\neq m, F=G$ or $k\neq m, F=G$
\[
    \begin{prooftree}
\hypo{\emptyset}
            \infer1[\textsc{Id}]{[k:A;l:A]}
\infer1[i\#]{[i:\texttt{\#};l:A]}
\hypo{\emptyset}
            \infer1[\textsc{Id}]{[k:A;l:A]}
            \hypo{\emptyset}
            \infer1[\textsc{Id}]{[m:A;n:A]}
\infer2[j\#]{[j:\texttt{\#};k:A;n:A]}
\infer2[\textsc{Cut}]{[k:A;l:A;n:A]}
\end{prooftree}\, .
\]
\item $\texttt{S}(i\#)=\delta cs; \texttt{S}(j\#)=\delta cd$: from
\[
\begin{prooftree}
\hypo{~~~~~~~~~}
            \infer[rule code={\hbox{\tikz\draw (0,0) -- (\hsize,0) -- (0.5\hsize,-2em) (0.5\hsize,-0.875em) node{\small $\mathfrak{p}_1$}  (0.5\hsize,-2em) -- (0,0);}}]1{[k:F]s}
            \hypo{~~~~~~~~~}
            \infer[rule code={\hbox{\tikz\draw (0,0) -- (\hsize,0) -- (0.5\hsize,-2em) (0.5\hsize,-0.875em) node{\small $\mathfrak{p}_2$}  (0.5\hsize,-2em) -- (0,0);}}]1{s}
\infer2[i\#]{[i:\texttt{\#}]s}
\hypo{~~~~~~~~~}
            \infer[rule code={\hbox{\tikz\draw (0,0) -- (\hsize,0) -- (0.5\hsize,-2em) (0.5\hsize,-0.875em) node{\small $\mathfrak{p}_3$}  (0.5\hsize,-2em) -- (0,0);}}]1{[l:F]t}
            \hypo{~~~~~~~~~}
            \infer[rule code={\hbox{\tikz\draw (0,0) -- (\hsize,0) -- (0.5\hsize,-2em) (0.5\hsize,-0.875em) node{\small $\mathfrak{p}_2$}  (0.5\hsize,-2em) -- (0,0);}}]1{[m:G]u}
\infer2[j\#]{[j:\texttt{\#}]tu}
\infer2[\textsc{Cut}]{stuv}
\end{prooftree}\, ,
\]
we get for $m\neq l$ and $\texttt{\#}^\prime=\texttt{\#}$ (i.e. $\#$ is nullary)
{\scriptsize
\[
\begin{prooftree}
    \hypo{\emptyset}
    \infer1[\textsc{Id}]{[k:A;l:A]}
    \hypo{\emptyset}
    \infer1[\textsc{Id}]{[k:A;l:A]}
    \infer2[j\#]{[j:\texttt{\#}^\prime;k:A;l:A]}
        \hypo{\emptyset}
    \infer1[\textsc{Id}]{[k:A;l:A]}
    \infer2[i\#]{[i:\texttt{\#};k:A;l:A]}
    \hypo{\emptyset}
    \infer1[\textsc{Id}]{[k:B;l:B]}
    \hypo{\emptyset}
    \infer1[\textsc{Id}]{[k:C;l:C]}
    \infer2[j\#]{[j:\texttt{\#};k:B;l:C]}
    \infer2[\textsc{Cut}]{[k:B;l:C]}
\end{prooftree}\, .
\]
}
    \item $\texttt{S}(i\#)=\delta cd; \texttt{S}(j\#)=\delta cs$: from
    \[
    \begin{prooftree}
\hypo{~~~~~~~~~}
            \infer[rule code={\hbox{\tikz\draw (0,0) -- (\hsize,0) -- (0.5\hsize,-2em) (0.5\hsize,-0.875em) node{\small $\mathfrak{p}_1$}  (0.5\hsize,-2em) -- (0,0);}}]1{[k:F]s}
            \hypo{~~~~~~~~~}
            \infer[rule code={\hbox{\tikz\draw (0,0) -- (\hsize,0) -- (0.5\hsize,-2em) (0.5\hsize,-0.875em) node{\small $\mathfrak{p}_2$}  (0.5\hsize,-2em) -- (0,0);}}]1{t}
\infer2[i\#]{[i:\texttt{\#}]st}
\hypo{~~~~~~~~~}
            \infer[rule code={\hbox{\tikz\draw (0,0) -- (\hsize,0) -- (0.5\hsize,-2em) (0.5\hsize,-0.875em) node{\small $\mathfrak{p}_3$}  (0.5\hsize,-2em) -- (0,0);}}]1{[l:F]u}
            \hypo{~~~~~~~~~}
            \infer[rule code={\hbox{\tikz\draw (0,0) -- (\hsize,0) -- (0.5\hsize,-2em) (0.5\hsize,-0.875em) node{\small $\mathfrak{p}_2$}  (0.5\hsize,-2em) -- (0,0);}}]1{[m:G]u}
\infer2[j\#]{[j:\texttt{\#}]u}
\infer2[\textsc{Cut}]{stuv}
\end{prooftree}\, ,
\]
we get for $m=l, k\neq l$
\[
\begin{prooftree}
    \hypo{\emptyset}
    \infer1[\textsc{Id}]{[k:A;l:A]}
    \hypo{\emptyset}
    \infer1[\textsc{Id}]{[k:B;l:B]}
    \infer2[i\#]{[i:\texttt{\#};k:B;l:A,B]}
    \hypo{\emptyset}
    \infer1[\textsc{Id}]{[k:A;l:A]}
    \hypo{\emptyset}
    \infer1[\textsc{Id}]{[k:A;l:A]}
    \infer2[j\#]{[j:\texttt{\#};k:A,A]}
    \infer2[\textsc{Cut}]{[k:A,A;l:A,B]}
\end{prooftree}\, .
\]
   \item $\texttt{S}(i\#)=\delta cd; \texttt{S}(j\#)=\delta cd$: from
\[
    \begin{prooftree}
\hypo{~~~~~~~~~}
            \infer[rule code={\hbox{\tikz\draw (0,0) -- (\hsize,0) -- (0.5\hsize,-2em) (0.5\hsize,-0.875em) node{\small $\mathfrak{p}_1$}  (0.5\hsize,-2em) -- (0,0);}}]1{[k:F]s}
            \hypo{~~~~~~~~~}
            \infer[rule code={\hbox{\tikz\draw (0,0) -- (\hsize,0) -- (0.5\hsize,-2em) (0.5\hsize,-0.875em) node{\small $\mathfrak{p}_2$}  (0.5\hsize,-2em) -- (0,0);}}]1{t}
\infer2[i\#]{[i:\texttt{\#}]st}
\hypo{~~~~~~~~~}
            \infer[rule code={\hbox{\tikz\draw (0,0) -- (\hsize,0) -- (0.5\hsize,-2em) (0.5\hsize,-0.875em) node{\small $\mathfrak{p}_3$}  (0.5\hsize,-2em) -- (0,0);}}]1{[l:F]u}
            \hypo{~~~~~~~~~}
            \infer[rule code={\hbox{\tikz\draw (0,0) -- (\hsize,0) -- (0.5\hsize,-2em) (0.5\hsize,-0.875em) node{\small $\mathfrak{p}_2$}  (0.5\hsize,-2em) -- (0,0);}}]1{[m:G]v}
\infer2[j\#]{[j:\texttt{\#}]uv}
\infer2[\textsc{Cut}]{stuv}
\end{prooftree}\, ,
\]
we get for $k\neq l$ or $l\neq m, F=G$ or $k\neq m, F=G$
\[
\begin{prooftree}
    \hypo{\emptyset}
    \infer1[\textsc{Id}]{[k:A;l:A]}
    \hypo{\emptyset}
    \infer1[\textsc{Id}]{[1:B;2:B]}
    \infer2[i\#]{[i:\texttt{\#};1:B;2:B;l:A]}
    \hypo{\emptyset}
    \infer1[\textsc{Id}]{[k:A;l:A]}
    \hypo{\emptyset}
    \infer1[\textsc{Id}]{[m:A;n:A]}
    \infer2[j\#]{[j:\texttt{\#};k:A;n:A]}
    \infer2[\textsc{Cut}]{[1:B;2:B;k:A;l:A;n:A]}
\end{prooftree}\, .
\]
\end{enumerate}
Case 8) $\texttt{P}(i\#)=kFlG; \texttt{P}(j\#)=mFnG$.
\begin{enumerate}[label=8.\arabic*)]
\item $\texttt{S}(i\#)=\beta; \texttt{S}(j\#)=\beta$: from
\[
\begin{prooftree}
\hypo{~~~~~~~~~}
            \infer[rule code={\hbox{\tikz\draw (0,0) -- (\hsize,0) -- (0.5\hsize,-2em) (0.5\hsize,-0.875em) node{\small $\mathfrak{p}_1$}  (0.5\hsize,-2em) -- (0,0);}}]1{[k:F;l:G]}
\infer1[i\#]{[i:\texttt{\#}]}
\hypo{~~~~~~~~~}
            \infer[rule code={\hbox{\tikz\draw (0,0) -- (\hsize,0) -- (0.5\hsize,-2em) (0.5\hsize,-0.875em) node{\small $\mathfrak{p}_2$}  (0.5\hsize,-2em) -- (0,0);}}]1{[m:F;n:G]}
\infer1[j\#]{[j:\texttt{\#}]}
\infer2[\textsc{Cut}]{st}
\end{prooftree}\, ,
\]
we get for $k \neq l, m\neq n, F=G$
\[
\begin{prooftree}
\hypo{\emptyset}
            \infer1[\textsc{Id}]{[k:A;l:A]}
\infer1[i\#]{[i:\texttt{\#}]}
\hypo{\emptyset}
            \infer1[\textsc{Id}]{[k:A;l:A]}
\infer1[j\#]{[j:\texttt{\#}]}
\infer2[\textsc{Cut}]{\emptyset\vdash \emptyset}
\end{prooftree}\, .
\]
\item $\texttt{S}(i\#)=\beta; \texttt{S}(j\#)=\gamma$: cf. 8.1).
\item $\texttt{S}(i\#)=\beta; \texttt{S}(j\#)=\epsilon$: from
\[
\begin{prooftree}
\hypo{\emptyset}
\infer1[i\#]{[i:\texttt{\#}]}
\hypo{~~~~~~~~~}
            \infer[rule code={\hbox{\tikz\draw (0,0) -- (\hsize,0) -- (0.5\hsize,-2em) (0.5\hsize,-0.875em) node{\small $\mathfrak{p}_1$}  (0.5\hsize,-2em) -- (0,0);}}]1{[m:F;n:G]s}
\infer1[j\#]{[j:\texttt{\#}]s}
\infer2[\textsc{Cut}]{s}
\end{prooftree}\, ,
\]
we get for $k \neq l, m\neq n$
\[
\begin{prooftree}
\hypo{\emptyset}
\infer1[i\#]{[i:\texttt{\#}]}
\hypo{\emptyset}
            \infer1[\textsc{Id}]{[k:A;l:A]}
\infer1[j\#]{[j:\texttt{\#}]}
\infer2[\textsc{Cut}]{\emptyset\vdash \emptyset}
\end{prooftree}\, .
\]
\item $\texttt{S}(i\#)=\gamma; \texttt{S}(j\#)=\gamma$: from
\[
\begin{prooftree}
\hypo{~~~~~~~~~}
            \infer[rule code={\hbox{\tikz\draw (0,0) -- (\hsize,0) -- (0.5\hsize,-2em) (0.5\hsize,-0.875em) node{\small $\mathfrak{p}_1$}  (0.5\hsize,-2em) -- (0,0);}}]1{[k:F;l:G]s}
\infer1[i\#]{[i:\texttt{\#}]s}
\hypo{~~~~~~~~~}
            \infer[rule code={\hbox{\tikz\draw (0,0) -- (\hsize,0) -- (0.5\hsize,-2em) (0.5\hsize,-0.875em) node{\small $\mathfrak{p}_2$}  (0.5\hsize,-2em) -- (0,0);}}]1{[m:F;n:G]t}
\infer1[j\#]{[j:\texttt{\#}]t}
\infer2[\textsc{Cut}]{st}
\end{prooftree}\, ,
\]
we get for $k\neq l, m\neq n$
\[
\begin{prooftree}
\hypo{\emptyset}
            \infer1[\textsc{Id}]{[k:A;l:A]}
\infer1[i\#]{[i:\texttt{\#}]}
\hypo{\emptyset}
            \infer1[\textsc{Id}]{[k:A;l:A]}
\infer1[j\#]{[j:\texttt{\#}]}
\infer2[\textsc{Cut}]{\emptyset\vdash \emptyset}
\end{prooftree}\, .
\]
\item $\texttt{S}(i\#)=\gamma; \texttt{S}(j\#)=\epsilon$: cf. 8.4).
\item $\texttt{S}(i\#)=\delta cd; \texttt{S}(j\#)=\delta cd$: from 
\[
\begin{prooftree}
\hypo{~~~~~~~~~}
            \infer[rule code={\hbox{\tikz\draw (0,0) -- (\hsize,0) -- (0.5\hsize,-2em) (0.5\hsize,-0.875em) node{\small $\mathfrak{p}_1$}  (0.5\hsize,-2em) -- (0,0);}}]1{[k:F;l:G]s}
            \hypo{~~~~~~~~~}
            \infer[rule code={\hbox{\tikz\draw (0,0) -- (\hsize,0) -- (0.5\hsize,-2em) (0.5\hsize,-0.875em) node{\small $\mathfrak{p}_2$}  (0.5\hsize,-2em) -- (0,0);}}]1{t}
\infer2[i\#]{[i:\texttt{\#}]st}
\hypo{~~~~~~~~~}
            \infer[rule code={\hbox{\tikz\draw (0,0) -- (\hsize,0) -- (0.5\hsize,-2em) (0.5\hsize,-0.875em) node{\small $\mathfrak{p}_3$}  (0.5\hsize,-2em) -- (0,0);}}]1{[m:F;n:G]u}
            \hypo{~~~~~~~~~}
            \infer[rule code={\hbox{\tikz\draw (0,0) -- (\hsize,0) -- (0.5\hsize,-2em) (0.5\hsize,-0.875em) node{\small $\mathfrak{p}_4$}  (0.5\hsize,-2em) -- (0,0);}}]1{v}
\infer2[j\#]{[j:\texttt{\#}]uv}
\infer2[\textsc{Cut}]{stuv}
\end{prooftree}\, ,
\]
we get for $k\neq l, m\neq n$
\[
\begin{prooftree}
\hypo{\emptyset}
            \infer1[\textsc{Id}]{[k:A;l:A]}
\hypo{\emptyset}
            \infer1[\textsc{Id}]{[k:B;l:B]}
\infer2[i\#]{[i:\texttt{\#};k:B;l:B]}
\hypo{\emptyset}
            \infer1[\textsc{Id}]{[k:A;l:A]}
\hypo{\emptyset}
            \infer1[\textsc{Id}]{[k:C;l:C]}
\infer2[j\#]{[j:\texttt{\#};k:C;l:C]}
\infer2[\textsc{Cut}]{[k:B,C;l:B,C]}
\end{prooftree}\, .
\]
\item $\texttt{S}(i\#)=\delta cd; \texttt{S}(j\#)=\epsilon$: from 
\[
\begin{prooftree}
\hypo{~~~~~~~~~}
            \infer[rule code={\hbox{\tikz\draw (0,0) -- (\hsize,0) -- (0.5\hsize,-2em) (0.5\hsize,-0.875em) node{\small $\mathfrak{p}_1$}  (0.5\hsize,-2em) -- (0,0);}}]1{[k:F;l:G]s}
            \hypo{~~~~~~~~~}
            \infer[rule code={\hbox{\tikz\draw (0,0) -- (\hsize,0) -- (0.5\hsize,-2em) (0.5\hsize,-0.875em) node{\small $\mathfrak{p}_2$}  (0.5\hsize,-2em) -- (0,0);}}]1{t}
\infer2[i\#]{[i:\texttt{\#}]st}
\hypo{~~~~~~~~~}
            \infer[rule code={\hbox{\tikz\draw (0,0) -- (\hsize,0) -- (0.5\hsize,-2em) (0.5\hsize,-0.875em) node{\small $\mathfrak{p}_3$}  (0.5\hsize,-2em) -- (0,0);}}]1{u}
\infer1[j\#]{[j:\texttt{\#}]u}
\infer2[\textsc{Cut}]{stu}
\end{prooftree}\, ,
\]
we get for $k\neq l, m\neq n$
\[
\begin{prooftree}
\hypo{\emptyset}
            \infer1[\textsc{Id}]{[k:A;l:A]}
\hypo{\emptyset}
            \infer1[\textsc{Id}]{[k:B;l:B]}
\infer2[i\#]{[i:\texttt{\#};k:B;l:B]}
\hypo{\emptyset}
            \infer1[\textsc{Id}]{[k:C;l:C]}
\infer1[j\#]{[j:\texttt{\#};k:C;l:C]}
\infer2[\textsc{Cut}]{[k:B,C;l:B,C]}
\end{prooftree}\, .
\]
\item $\texttt{S}(i\#)=\epsilon; \texttt{S}(j\#)=\epsilon$: cf. 8.4).
\end{enumerate}
Case 9) $\texttt{P}(i\#)=kFlG; \texttt{P}(j\#)=mF-nG$.
\begin{enumerate}[label=9.\arabic*)]
\item $\texttt{S}(i\#)=\delta cd; \texttt{S}(j\#)=\delta cd$: from
\[
\begin{prooftree}
\hypo{~~~~~~~~~}
            \infer[rule code={\hbox{\tikz\draw (0,0) -- (\hsize,0) -- (0.5\hsize,-2em) (0.5\hsize,-0.875em) node{\small $\mathfrak{p}_1$}  (0.5\hsize,-2em) -- (0,0);}}]1{[k:F;l:G]s}
\hypo{~~~~~~~~~}
            \infer[rule code={\hbox{\tikz\draw (0,0) -- (\hsize,0) -- (0.5\hsize,-2em) (0.5\hsize,-0.875em) node{\small $\mathfrak{p}_2$}  (0.5\hsize,-2em) -- (0,0);}}]1{t}
\infer2[i\#]{[i:\texttt{\#}]st}
\hypo{~~~~~~~~~}
            \infer[rule code={\hbox{\tikz\draw (0,0) -- (\hsize,0) -- (0.5\hsize,-2em) (0.5\hsize,-0.875em) node{\small $\mathfrak{p}_3$}  (0.5\hsize,-2em) -- (0,0);}}]1{[m:F]u}
\hypo{~~~~~~~~~}
            \infer[rule code={\hbox{\tikz\draw (0,0) -- (\hsize,0) -- (0.5\hsize,-2em) (0.5\hsize,-0.875em) node{\small $\mathfrak{p}_4$}  (0.5\hsize,-2em) -- (0,0);}}]1{[n:G]v}
\infer2[j\#]{[j:\texttt{\#}]uv}
\infer2[\textsc{Cut}]{stuv}
\end{prooftree}\, ,
\]
we get for $k\neq m, l\neq n$ or $m\neq n, F=G$
\[
\begin{prooftree}
\hypo{\emptyset}
            \infer1[\textsc{Id}]{[k:A;l:A]}
\hypo{\emptyset}
            \infer1[\textsc{Id}]{[k:B;l:B]}
\infer2[i\#]{[i:\texttt{\#};k:B;l:B]}
\hypo{\emptyset}
            \infer1[\textsc{Id}]{[k:A;l:A]}
\hypo{\emptyset}
            \infer1[\textsc{Id}]{[k:A;l:A]}
\infer2[j\#]{[j:\texttt{\#};k:A;l:A]}
\infer2[\textsc{Cut}]{[k:A,B;l:A,B]}
\end{prooftree}\, .
\]
\item $\texttt{S}(i\#)=\epsilon; \texttt{S}(j\#)=\delta cd$: from
\[
\begin{prooftree}
\hypo{~~~~~~~~~}
            \infer[rule code={\hbox{\tikz\draw (0,0) -- (\hsize,0) -- (0.5\hsize,-2em) (0.5\hsize,-0.875em) node{\small $\mathfrak{p}_1$}  (0.5\hsize,-2em) -- (0,0);}}]1{s}
\infer1[i\#]{[i:\texttt{\#}]s}
\hypo{~~~~~~~~~}
            \infer[rule code={\hbox{\tikz\draw (0,0) -- (\hsize,0) -- (0.5\hsize,-2em) (0.5\hsize,-0.875em) node{\small $\mathfrak{p}_2$}  (0.5\hsize,-2em) -- (0,0);}}]1{[m:F]t}
\hypo{~~~~~~~~~}
            \infer[rule code={\hbox{\tikz\draw (0,0) -- (\hsize,0) -- (0.5\hsize,-2em) (0.5\hsize,-0.875em) node{\small $\mathfrak{p}_4$}  (0.5\hsize,-2em) -- (0,0);}}]1{[n:G]u}
\infer2[j\#]{[j:\texttt{\#}]tu}
\infer2[\textsc{Cut}]{stu}
\end{prooftree}\, ,
\]
we get for $m\neq n, F=G$
\[
\begin{prooftree}
\hypo{\emptyset}
            \infer1[\textsc{Id}]{[k:B;l:B]}
\infer1[i\#]{[i:\texttt{\#};k:B;l:B]}
\hypo{\emptyset}
            \infer1[\textsc{Id}]{[k:A;l:A]}
\hypo{\emptyset}
            \infer1[\textsc{Id}]{[k:A;l:A]}
\infer2[j\#]{[j:\texttt{\#};k:A;l:A]}
\infer2[\textsc{Cut}]{[k:A,B;l:A,B]}
\end{prooftree}\, .
\]
\end{enumerate}
Case 10) $\texttt{P}(i\#)=kF-lG; \texttt{P}(j\#)=mF-nG$.
\begin{enumerate}[label=10.\arabic*)]
\item $\texttt{S}(i\#)=\delta cs; \texttt{S}(j\#)=\delta cd$: from
\[
\begin{prooftree}
\hypo{~~~~~~~~~}
            \infer[rule code={\hbox{\tikz\draw (0,0) -- (\hsize,0) -- (0.5\hsize,-2em) (0.5\hsize,-0.875em) node{\small $\mathfrak{p}_1$}  (0.5\hsize,-2em) -- (0,0);}}]1{[k:F]s}
\hypo{~~~~~~~~~}
            \infer[rule code={\hbox{\tikz\draw (0,0) -- (\hsize,0) -- (0.5\hsize,-2em) (0.5\hsize,-0.875em) node{\small $\mathfrak{p}_2$}  (0.5\hsize,-2em) -- (0,0);}}]1{[l:G]s}
\infer2[i\#]{[i:\texttt{\#}]s}
\hypo{~~~~~~~~~}
            \infer[rule code={\hbox{\tikz\draw (0,0) -- (\hsize,0) -- (0.5\hsize,-2em) (0.5\hsize,-0.875em) node{\small $\mathfrak{p}_3$}  (0.5\hsize,-2em) -- (0,0);}}]1{[m:F]t}
\hypo{~~~~~~~~~}
            \infer[rule code={\hbox{\tikz\draw (0,0) -- (\hsize,0) -- (0.5\hsize,-2em) (0.5\hsize,-0.875em) node{\small $\mathfrak{p}_4$}  (0.5\hsize,-2em) -- (0,0);}}]1{[n:G]u}
\infer2[j\#]{[j:\texttt{\#}]tu}
\infer2[\textsc{Cut}]{stu}
\end{prooftree}\ , 
\]
we get for $F=G, k=l$ (assuming $k\neq k^\prime, m\neq m^\prime, n\neq n`^\prime$)
{\footnotesize
\[
\begin{prooftree}
\hypo{\emptyset}
            \infer1[\textsc{Id}]{[k:A;k^\prime:A]}
\hypo{\emptyset}
            \infer1[\textsc{Id}]{[k:A;k^\prime:A]}
\infer2[i\#]{[i:\texttt{\#};k^\prime:A]}
\hypo{\emptyset}
            \infer1[\textsc{Id}]{[m:A;m^\prime:A]}
\hypo{\emptyset}
            \infer1[\textsc{Id}]{[n:A;n^\prime:A]}
\infer2[j\#]{[j:\texttt{\#};m^\prime:A; n^\prime: A]}
\infer2[\textsc{Cut}]{[k^\prime:A; m^\prime:A; n^\prime: A]}
\end{prooftree}
\] 
}
\item $\texttt{S}(i\#)=\delta cd; \texttt{S}(j\#)=\delta cd$: from
\[
\begin{prooftree}
\hypo{~~~~~~~~~}
            \infer[rule code={\hbox{\tikz\draw (0,0) -- (\hsize,0) -- (0.5\hsize,-2em) (0.5\hsize,-0.875em) node{\small $\mathfrak{p}_1$}  (0.5\hsize,-2em) -- (0,0);}}]1{[k:F]s}
\hypo{~~~~~~~~~}
            \infer[rule code={\hbox{\tikz\draw (0,0) -- (\hsize,0) -- (0.5\hsize,-2em) (0.5\hsize,-0.875em) node{\small $\mathfrak{p}_2$}  (0.5\hsize,-2em) -- (0,0);}}]1{[l:G]t}
\infer2[i\#]{[i:\texttt{\#}]st}
\hypo{~~~~~~~~~}
            \infer[rule code={\hbox{\tikz\draw (0,0) -- (\hsize,0) -- (0.5\hsize,-2em) (0.5\hsize,-0.875em) node{\small $\mathfrak{p}_3$}  (0.5\hsize,-2em) -- (0,0);}}]1{[m:F]u}
\hypo{~~~~~~~~~}
            \infer[rule code={\hbox{\tikz\draw (0,0) -- (\hsize,0) -- (0.5\hsize,-2em) (0.5\hsize,-0.875em) node{\small $\mathfrak{p}_4$}  (0.5\hsize,-2em) -- (0,0);}}]1{[n:G]v}
\infer2[j\#]{[j:\texttt{\#}]uv}
\infer2[\textsc{Cut}]{stuv}
\end{prooftree}\, ,
\]
we get for $k\neq m, l\neq n$ or $F=G, k\neq l$ or $F=G, m\neq n$
\[
\begin{prooftree}
\hypo{\emptyset}
            \infer1[\textsc{Id}]{[k:A;l:A]}
\hypo{\emptyset}
            \infer1[\textsc{Id}]{[k:A;l:A]}
\infer2[i\#]{[i:\texttt{\#};k:A;l:A]}
\hypo{\emptyset}
            \infer1[\textsc{Id}]{[k:A;l:A]}
\hypo{\emptyset}
            \infer1[\textsc{Id}]{[k:A;l:A]}
\infer2[j\#]{[j:\texttt{\#};k:A;l:A]}
\infer2[\textsc{Cut}]{[k:A, A;l:A, A]}
\end{prooftree}\, .
\]
\item $\texttt{S}(i\#)=\delta cd; \texttt{S}(j\#)=\epsilon$: from
\[
\begin{prooftree}
\hypo{~~~~~~~~~}
            \infer[rule code={\hbox{\tikz\draw (0,0) -- (\hsize,0) -- (0.5\hsize,-2em) (0.5\hsize,-0.875em) node{\small $\mathfrak{p}_1$}  (0.5\hsize,-2em) -- (0,0);}}]1{[k:F]s}
\hypo{~~~~~~~~~}
            \infer[rule code={\hbox{\tikz\draw (0,0) -- (\hsize,0) -- (0.5\hsize,-2em) (0.5\hsize,-0.875em) node{\small $\mathfrak{p}_2$}  (0.5\hsize,-2em) -- (0,0);}}]1{[l:G]t}
\infer2[i\#]{[i:\texttt{\#}]st}
\hypo{~~~~~~~~~}
            \infer[rule code={\hbox{\tikz\draw (0,0) -- (\hsize,0) -- (0.5\hsize,-2em) (0.5\hsize,-0.875em) node{\small $\mathfrak{p}_3$}  (0.5\hsize,-2em) -- (0,0);}}]1{[m:F]u}
\infer1[j\#]{[j:\texttt{\#}]u}
\infer2[\textsc{Cut}]{stu}
\end{prooftree}\, ,
\]
we get for $k\neq m, l\neq n$ or $F=G, k\neq l$
\[
\begin{prooftree}
\hypo{\emptyset}
            \infer1{[k:A;l:A]}
\hypo{\emptyset}
            \infer1{[k:A;l:A]}
\infer2[i\#]{[i:\texttt{\#};k:A;l:A]}
\hypo{\emptyset}
            \infer1{[k:A;l:A]}
\infer1[j\#]{[j:\texttt{\#};k:A]}
\infer2[\textsc{Cut}]{[k:A, A;l:A]}
\end{prooftree}\, .
\]
\hfill $\blacksquare$
\end{enumerate}
\end{enumerate}
\end{proof}
\subsubsection{Special and Exceptional Cases}
We have proven that all and only the $376$ members of \texttt{PC} are \textit{minimally conservative}, i.e. conservative in $\{\textsc{Id}, \textsc{Cut} cd\}$. However, there are still some connectives among our $10,816$ formulable connectives that are neither covered by Theorem 3 nor by Theorem 4. These instances of \# are \textit{trivially} conservative.

This is because \textit{no} new \#-free sequent can be proven after adding the \#-rules to $\{\textsc{Id}, \textsc{Cut} cd\}$, and the antecedent of Definition 15 is therefore vacuously satisfied for these \#-specifications. Typically, this occurs because at least one of the \#-rules cannot be applied in $\{\textsc{Id}, \textsc{Cut} cd\}$ due to contextual constraints arising from the context-sharing premiss shape of at least one \#-rule. Let us consider two examples.
\begin{example}\

    \noindent Let \# be of premiss type $\langle \alpha, \delta cs kF \rangle$, $k\in \{1,2\}$, i.e. defined by the operational rules: 
    \[ \Bigg \langle
    \begin{prooftree}
        \hypo{\emptyset}
        \infer1[1\#]{[1:\texttt{\#}]}
    \end{prooftree},
    \begin{prooftree}
        \hypo{[k:F]s}
        \hypo{s}
        \infer2[2\#]{[2:\texttt{\#}]s}
    \end{prooftree}
    \Bigg \rangle\, .
    \]
    There is no proof $\mathfrak{p}$ in $\{1\#,2\#,\textsc{Id}, \textsc{Cut} cd\}$ of a sequent $\mathfrak{s}$ such that $\mathfrak{p}$ uses 1\# or 2\# but \# does not occur in $\mathfrak{s}$.
    \begin{proof}
    First, we show that there is no proof using rule 2\# in $\{1\#,2\#,\textsc{Id}, \textsc{Cut} cd\}$. Let us attempt an inverted proof search for such a proof $\mathfrak{p}^\prime=$
    \[
    \begin{prooftree}
       \hypo{~~~~~~~~~}
            \infer[rule code={\hbox{\tikz\draw (0,0) -- (\hsize,0) -- (0.5\hsize,-2em) (0.5\hsize,-0.875em) node{\small $\mathfrak{p}_1$}  (0.5\hsize,-2em) -- (0,0);}}]1{[k:F]s}
    \hypo{~~~~~~~~~}
            \infer[rule code={\hbox{\tikz\draw (0,0) -- (\hsize,0) -- (0.5\hsize,-2em) (0.5\hsize,-0.875em) node{\small $\mathfrak{p}_2$}  (0.5\hsize,-2em) -- (0,0);}}]1{s}
    \infer2[2\#]{[2:\#]s}
    \end{prooftree}\, .
    \]
    All theorems of $\{\textsc{Id}, \textsc{Cut} cd\}$ are instances of \textsc{Id}$^{pr}$ by Lemma 2, and 1\# only produces sequents of the form $[1:\texttt{\#}]$. Moreover, 1\# cannot serve as a premiss sequent of \textsc{Cut} in $\{1\#,\textsc{Id}, \textsc{Cut} cd\}$. If it were, $\texttt{\#}$---the only formula in the conclusion of 1\#---would need to be the \textsc{Cut}-formula, whereas no available rule produces sequents of form $[2:\texttt{\#}]t$. Therefore, all theorems of $\{1\#,\textsc{Id}, \textsc{Cut} cd\}$ are instances of \textsc{Id}$^{pr}$ or $[1:\texttt{\#}]$.

    Applying this insight to our proof search for $\mathfrak{p}^\prime$, it follows that there are no matching $\mathfrak{p}_1$ and $\mathfrak{p}_2$ in $\{1\#,\textsc{Id}, \textsc{Cut} cd\}$ such that we can apply 2\#. If $\mathfrak{p}_2$ ended in $s=[1:\texttt{\#}]$, we would not find a matching $\mathfrak{p}_1$ since $[k:F;1:\texttt{\#}]$ is neither an instance of \textsc{Id}$^{pr}$ nor of 1\# and, thus, unprovable in $\{1\#,\textsc{Id}, \textsc{Cut} cd\}$. If $s$ were of form $[1:A^{pr};2: A^{pr}]$, $[k:F]s=[k:F;1:A^{pr};2: A^{pr}]$ would also not be provable for the same reason. Finally, if $[k:F]s=[1:A^{pr};2: A^{pr}]$, $s=[1:A^{pr}]$ or $s=[2:A^{pr}]$, neither of which is provable. Therefore, 2\# cannot be applied in $\{1\#,\textsc{Id}, \textsc{Cut} cd\}$ and, \textit{a fortiori}, $\{1\#, 2\#,\textsc{Id}, \textsc{Cut} cd\}$ since no initial application of $2\#$ can be constructed.
    
    We can now show that there is no $\mathfrak{p}$ of $\mathfrak{s}$ using $\{1\#,2\#,\textsc{Id}, \textsc{Cut} cd\}$. Since 2\# cannot be applied, $\{1\#,2\#,\textsc{Id}, \textsc{Cut} cd\}$ has the same theorems as $\{1\#,\textsc{Id}, \textsc{Cut} cd\}$. Hence, $\mathfrak{s}$ is an instance of \textsc{Id}$^{pr}$, or $\mathfrak{s}=[1:\texttt{\#}]$. In the former case, it was derived using \textsc{Id}$^{pr}$ and \textsc{Cut} only, which violates the definition of $\mathfrak{p}$ as using 1\# or 2\# is required. In the latter case, $\mathfrak{s}$ would not be \#-free, contradicting the specification of $\mathfrak{s}$.
    \hfill $\blacksquare$ 
    \end{proof}
\end{example}
\begin{example}\

    \noindent     \noindent Let \# be of premiss type $\langle \delta cs kFlG, \delta cd mFnG \rangle$, $k,l,m,n\in \{1,2\}$, i.e. defined by the operational rules: 
    \[ \Bigg \langle
    \begin{prooftree}
        \hypo{[k:F;l:G]s}
        \hypo{s}
        \infer2[1\#]{[1:\texttt{\#}]s}
    \end{prooftree},
    \begin{prooftree}
        \hypo{[m:F]s}
        \hypo{[n:G]t}
        \infer2[2\#]{[2:\texttt{\#}]st}
    \end{prooftree}
    \Bigg \rangle\, .
    \]
    The constellations where $k\neq l, m\neq n$ are particularly interesting as they would allow Lemma 1 and, \textit{a fortiori}, conservativity to \textit{non-vacuously} succeed if the right \#-free $\mathfrak{s}$ were provable. We focus on these instances here; the others behave equivalently.
    
    There is no proof $\mathfrak{p}$ in $\{1\#,2\#,\textsc{Id}, \textsc{Cut} cd\}$ of a sequent $\mathfrak{s}$ such that $\mathfrak{p}$ uses 1\# or 2\# but \# does not occur in $\mathfrak{s}$.
    \begin{proof}
    We attempt to find such a proof top-down. We can immediately see that the premises of 1\# cannot be instances of \textsc{Id} alone. Hence, we must apply 2\# prior to 1\#.
    \begin{enumerate}
        \item Let $k\neq l$, then $m\neq n$. Then, initial 2\#-applications are proofs of the form $\mathfrak{p}_1=$
    \[
    \begin{prooftree}
        \hypo{\emptyset}
        \infer1[\textsc{Id}]{[1:A;2:A]}
        \hypo{\emptyset}
        \infer1[\textsc{Id}]{[1:A;2:A]}
        \infer2[2\#]{[1:A;2:A, \texttt{\#}]}
    \end{prooftree}\, .
    \]
    2\# takes one formula from the first and one from the second dimension of its premises. The same is true for chained 2\#-applications:
    \[
    \begin{prooftree}
        \hypo{~~~~~~~~~}
            \infer[rule code={\hbox{\tikz\draw (0,0) -- (\hsize,0) -- (0.5\hsize,-2em) (0.5\hsize,-0.875em) node{\small $\mathfrak{p}_1$}  (0.5\hsize,-2em) -- (0,0);}}]1{[1:A;2:A, \texttt{\#}]}
        \hypo{\emptyset}
        \infer1[\textsc{Id}]{[1:A;2:A]}
        \infer2[2\#]{[1:A,2:A,\texttt{\#}, \texttt{\#}]}
        \hypo{~~~~~~~~~}
            \infer[rule code={\hbox{\tikz\draw (0,0) -- (\hsize,0) -- (0.5\hsize,-2em) (0.5\hsize,-0.875em) node{\small $\mathfrak{p}_1$}  (0.5\hsize,-2em) -- (0,0);}}]1{[1:A;2:A, \texttt{\#}]}
        \infer2[2\#]{[1:A;2:A, A, \texttt{\#}, \texttt{\#}^\prime]}
    \end{prooftree}\, .
    \]
    Therefore, any theorem of $\{2\#,\textsc{Id}, \textsc{Cut} cd\}$ has exactly one formula in the first dimension and is of form $[1:F;2:F,\Gamma]$, for any multiset of formulae $\Gamma$.

    This means that 1\# cannot be used in $\{1\#, 2\#,\textsc{Id}, \textsc{Cut} cd\}$, equivalently to the previous example: looking for premises of 1\#, if $s=[1:F;2:F,\Gamma]$, $[1:F;2:G]s$ must be of the form $[1:F,F;2:F,G,\Gamma]$, and if $[1:F;2:G]s=[1:F;2:F,\Gamma]$, $s$ is of the form $[1:\emptyset;2:\Delta]$. The consequent of neither conditional is provable in $\{2\#,\textsc{Id}, \textsc{Cut} cd\}$ as already established.
    \item Let $k= l$, then $m= n$.
    \begin{enumerate}
        \item $m=1, k=2$ only results in new $\{2\#,\textsc{Id}, \textsc{Cut} cd\}$-theorems of the form $[2:A,A,\texttt{\#}]$, rendering 1\# inapplicable again. 
        \item The case where $m=2, k=1$ is more interesting. Let $\mathfrak{p}_1$ again refer to an initial 2\#-application: 
    \[
    \begin{prooftree}
        \hypo{\emptyset}
        \infer1[\textsc{Id}]{[1:A;2:A]}
        \hypo{\emptyset}
        \infer1[\textsc{Id}]{[1:A;2:A]}
        \infer2[2\#]{[1:A,A;2:\texttt{\#}]}
        \end{prooftree}\, .
    \]   
    Although there are again no fitting premisses for an application of 1\# if the active formulae of our \#-rules are subformulae of their principal formula, 1\# is \textit{applicable} to some conclusions of 2\# given $\#=\#^\prime$ (we will come back to this restriction momentarily):
    \[
    \begin{prooftree}
        \hypo{~~~~~~~~~}
        \infer[rule code={\hbox{\tikz\draw (0,0) -- (\hsize,0) -- (0.5\hsize,-2em) (0.5\hsize,-0.875em) node{\small $\mathfrak{p}_1$}  (0.5\hsize,-2em) -- (0,0);}}]1{[1:A,A;2:\texttt{\#}]}
        \hypo{~~~~~~~~~}
        \infer[rule code={\hbox{\tikz\draw (0,0) -- (\hsize,0) -- (0.5\hsize,-2em) (0.5\hsize,-0.875em) node{\small $\mathfrak{p}_1$}  (0.5\hsize,-2em) -- (0,0);}}]1{[1:A,A;2:\texttt{\#}]}
        \infer2[2\#]{[1:A,A,A,A;2:\texttt{\#}^{\prime}]}
        \hypo{~~~~~~~~~}
        \infer[rule code={\hbox{\tikz\draw (0,0) -- (\hsize,0) -- (0.5\hsize,-2em) (0.5\hsize,-0.875em) node{\small $\mathfrak{p}_1$}  (0.5\hsize,-2em) -- (0,0);}}]1{[1:A,A;2:\texttt{\#}]}
        \infer2[1\#]{[1:A,A, \texttt{\#};2: \texttt{\#}]}
    \end{prooftree}\, .
    \]
    Nevertheless, we cannot prove any \#-free $\mathfrak{s}$. As we have seen in the previous example, to reach such an $\mathfrak{s}$, we must cut out all \#-formulae previously introduced by 1\# or 2\#. Let us attempt an inverted proof search on one such $\mathfrak{s}=st$:
    \[
    \begin{prooftree}
        \hypo{~~~~~~~~~}
        \infer[rule code={\hbox{\tikz\draw (0,0) -- (\hsize,0) -- (0.5\hsize,-2em) (0.5\hsize,-0.875em) node{\small $\mathfrak{p}_2$}  (0.5\hsize,-2em) -- (0,0);}}]1{[1:\texttt{\#}]s}
        \hypo{~~~~~~~~~}
        \infer[rule code={\hbox{\tikz\draw (0,0) -- (\hsize,0) -- (0.5\hsize,-2em) (0.5\hsize,-0.875em) node{\small $\mathfrak{p}_3$}  (0.5\hsize,-2em) -- (0,0);}}]1{[2:\texttt{\#}]t}
        \infer2[\textsc{Cut}]{st}
    \end{prooftree}\, .
    \]
    Thus, $s,t$ must be \#-free. However, there is no $\mathfrak{p}_2$ of  $[1:\texttt{\#}]s$: as 1\# only takes sequents with \#-formulae in the second dimension as its premises and all its active formulae are in the first dimension,  \#-formulae will always remain in the second component of its conclusion. Eliminating them via \textsc{Cut} requires another $\#$-formula in the first dimension, itself introduced by $1\#$, leading to a regress.\hfill $\blacksquare$
    \end{enumerate}
    \end{enumerate} 
    \end{proof}
    \end{example}

These examples are paradigmatic for our trivially conservative \textit{special cases}, and we will leave the remaining special cases as an exercise to the reader. An exhaustive classification of all cases can be found in the overview tables in \S5.1.4.

Finally, we discuss two \textit{exceptional cases}. First, consider \# with rules: 
\[
    \begin{prooftree}
        \hypo{[k:F]s}
        \hypo{s}
        \infer2[i\#]{[i:\texttt{\#}]s}
    \end{prooftree}\, , \qquad
    \begin{prooftree}
        \hypo{[l:G]s+s}
        \infer1[j\#]{[j:\texttt{\#}]s}
    \end{prooftree}\, .
\]
Although this case closely resembles our vacuously conservative special cases, there is a counterexample for $k\neq j, i\neq l$:
    \[
    \begin{prooftree}
        \hypo{\emptyset}
        \infer1[\textsc{Id}]{[1:A,2:A]}
        \infer1[j\#]{[j:\texttt{\#};1:A,2:A]}
        \hypo{\emptyset}
        \infer1[\textsc{Id}]{[1:A,2:A]}
        \infer2[i\#]{[i:\texttt{\#}^{\prime};1:A,2:A]}
        \infer1[j\#]{[j:\texttt{\#}^{\prime};1:A,2:A]}
        \hypo{\emptyset}
        \infer1[\textsc{Id}]{[1:A,2:A]}
        \infer1[j\#]{[j:\texttt{\#};1:A,2:A]}
        \hypo{\emptyset}
        \infer1[\textsc{Id}]{[1:A,2:A]}
        \infer2[i\#]{[i:\texttt{\#}^{\prime};1:A,2:A]}
        \infer2[\textsc{Cut}]{[1:A,A,2:A,A]}
    \end{prooftree}\, .
    \]
In this proof, we assume \textit{subformula-connectedness}:
\begin{definition}[subformula-connectedness]\

    \noindent An operational rule for connective \# is \textit{subformula-connected} iff the subformulae of its principal \#-formula are the active formulae of its premiss sequents.
\end{definition}
We subsequently distinguish $\texttt{\#}$ (subformula $A$) and $\texttt{\#}^\prime$ (subformula $\texttt{\#}$). Our investigation up to this point was \textit{independent} of this assumption. If we were to drop this constraint, $\texttt{\#}$ would be a constant and $\texttt{\#}=\texttt{\#}^\prime$. We could then provide an even simpler counterexample to this case, also eliminating the constraint that $i=l$:
    \[
    \begin{prooftree}
        \hypo{\emptyset}
        \infer1[\textsc{Id}]{[1:A,2:A]}
        \infer1[j\#]{[j:\texttt{\#};1:A,2:A]}
        \hypo{\emptyset}
        \infer1[\textsc{Id}]{[1:A,2:A]}
        \infer2[i\#]{[i:\texttt{\#};1:A,2:A]}
        \hypo{\emptyset}
        \infer1[\textsc{Id}]{[1:A,2:A]}
        \infer1[j\#]{[j:\texttt{\#};1:A,2:A]}
        \infer2[\textsc{Cut}]{[1:A,A,2:A,A]}
    \end{prooftree}\, .
    \]
Our second exceptional case shows similar subformula-dependent behaviour. \# now has the rules: 
\[
    \begin{prooftree}
        \hypo{[k:F]s}
        \hypo{s}
        \infer2[i\#]{[i:\texttt{\#}]s}
    \end{prooftree}\, , \qquad
    \begin{prooftree}
        \hypo{[l:F]s}
        \hypo{[m:G]t}
        \infer2[j\#]{[j:\texttt{\#}]st}
    \end{prooftree}\, .
    \]
    For $k=j,l\neq m$ we again get a counterexample (see Figure 6).
\begin{sidewaysfigure}[htbp]
\fbox{
\parbox{\linewidth}{
    {\scriptsize
    \[
    \begin{prooftree}
        \hypo{\emptyset}
        \infer1[\textsc{Id}]{[l:A;m:A]}
        \hypo{\emptyset}
        \infer1[\textsc{Id}]{[l:A;m:A]}
        \infer2[j\#]{[j:\texttt{\#};l:A;m:A]}
        \hypo{\emptyset}
        \infer1[\textsc{Id}]{[l:A;m:A]}
        \infer2[i\#]{[i:\texttt{\#}^\prime;l:A;m:A]}
        \hypo{\emptyset}
        \infer1[\textsc{Id}]{[l:A;m:A]}
        \hypo{\emptyset}
        \infer1[\textsc{Id}]{[l:A;m:A]}
        \infer2[j\#]{[j:\texttt{\#};l:A;m:A]}
        \hypo{\emptyset}
        \infer1[\textsc{Id}]{[l:A;m:A]}
        \infer2[i\#]{[i:\texttt{\#}^\prime;l:A;m:A]}
        \hypo{\emptyset}
        \infer1[\textsc{Id}]{[l:A;m:A]}
        \infer2[j\#]{[j:\texttt{\#}^\prime; i:A,A]}
        \infer2[\textsc{Cut}]{[j:A,A;i:A,A,A]}
    \end{prooftree}
    \] 
    }}}
    \caption{counterexample---second exceptional case, $k=j, l\neq m$}
    \end{sidewaysfigure}

Here, i\# and j\# have different active formula numbers. It is unclear whether \# should be considered a unary or binary connective when subformula-connectedness is assumed in this scenario. If we drop this assumption, we obtain a simplified proof as in the other exceptional case:
{\footnotesize
    \[
    \begin{prooftree}
        \hypo{\emptyset}
        \infer1[\textsc{Id}]{[l:A;m:A]}
        \hypo{\emptyset}
        \infer1[\textsc{Id}]{[l:A;m:A]}
        \infer2[j\#]{[j:\texttt{\#};l:A;m:A]}
        \hypo{\emptyset}
        \infer1[\textsc{Id}]{[l:A;m:A]}
        \infer2[i\#]{[i:\texttt{\#};l:A;m:A]}
        \hypo{\emptyset}
        \infer1[\textsc{Id}]{[l:A;m:A]}
        \hypo{\emptyset}
        \infer1[\textsc{Id}]{[l:A;m:A]}
        \infer2[j\#]{[j:\texttt{\#};l:A;m:A]}
        \infer2[\textsc{Cut}]{[l:A,A;m:A,A]}
    \end{prooftree}\, .
    \]   
}

Although these exceptional cases are rare but informative illustrations of the influence of the principal formula on definability, they are largely irrelevant for the remainder of the investigation, as they are \textit{negative} cases.

\subsubsection{Overview of cases}
We now give an overview of all formulable connectives considered for the conservativity proof. Each table represents a possible active formula constellation in the operational rules. The rows indicate the premiss shape of the $i$-rule, and the columns that of the $j$-rule for $i,j\in \{1,2\}$.

Instances of \# marked with `$\checkmark$' are conservative in $\{i\#,j\#, \textsc{Id}, \textsc{Cut} cd\}$ as established in Theorem 3. `$\checkmark ^-$' indicates that only a context-restricted version is conservative. Trivially conservative special cases are shown using `s'. `!' marks the exceptional cases discussed in \S5.1.3. Non-conservative instances are marked using Arabic numerals, in accordance with the case in the proof of Theorem 4 under which they fall (i.e. 1., 2. or 3.).

We do not separately list instances that syntactically collapse into other instances. For example, $\langle \epsilon kF-kF, \texttt{SP}(2\#)\rangle$ collapses into $\langle \epsilon kF, \texttt{SP}(2\#)\rangle$, irrespective of the premiss type of $2\#$.

\begin{table}[h]
    \centering
    \begin{tabular}{|c|c|c|c|c|}
    \hline
       $\mathtt{S}(i\#)\backslash \mathtt{S}(j\#) $  & $\alpha$ & $\gamma$ & $\delta cs$  & $\delta cd$\\ \hline
       $\alpha$  &  3  & $\checkmark$  & $\checkmark$  &  2  \\ \hline
       $\gamma$  & \cellcolor{lightgray}  &  2  &  2  &  2  \\ \hline
       $\delta cs$   &  \cellcolor{lightgray}  &  \cellcolor{lightgray}  &  2  &  2  \\ \hline
       $\delta cd$   &  \cellcolor{lightgray}  &  \cellcolor{lightgray}  &  \cellcolor{lightgray}  &  2  \\ \hline
    \end{tabular}
    \caption{$\mathtt{P}(i\#)=0$; $\mathtt{P}(j\#)=0$}
\end{table}

\begin{table}[h]
    \centering
    \begin{tabular}{|c|c|c|c|c|c|}
    \hline
       $\mathtt{S}(i\#)\backslash \mathtt{S}(j\#) $  & $\beta$ & $\gamma$ & $\delta cs$  & $\delta cd$  & $\epsilon$\\ \hline
       $\alpha$  &  1  & 1  & s  &  $\checkmark^-$, otherwise 2  &  1  \\ \hline
       $\gamma$  & $\checkmark$  &  2  &  s  &  2  &  2  \\ \hline
       $\delta cs$   & $\checkmark$  &  2  &  s  &  2  &  2 \\ \hline
       $\delta cd$   & 2  &  2  &  s  &  2  &  2  \\ \hline
    \end{tabular}
    \caption{$\mathtt{P}(i\#)=0$; $\mathtt{P}(j\#)=kF$}
\end{table}

\begin{table}[h]
    \centering
    \begin{tabular}{|c|c|c|c|c|c|}
    \hline
       $\mathtt{S}(i\#)\backslash \mathtt{S}(j\#) $  & $\beta$ & $\gamma$ & $\delta cs$  & $\delta cd$  & $\epsilon$\\ \hline
       $\alpha$  &  3 if $k\neq l, F=G$,  & 3 if $k\neq l, F=G$,  & s  &  $\checkmark^-$,  &  1 \\
       & otherwise 1 & otherwise 1 & & otherwise 2 & \\ \hline
       $\gamma$  & $\checkmark$  &  $\checkmark^-$ if $k\neq l, F=G$,  &  s  &  3 if $k\neq l, F=G$, &  2 \\
       & & otherwise 2 & & otherwise 2 & \\ \hline
       $\delta cs$   & $\checkmark$  &  $\checkmark^-$ if $k\neq l, F=G$,   &  s  &  3 if $k\neq l, F=G$,   &  2 \\
       & & otherwise 2 & & otherwise 2 &\\ \hline
       $\delta cd$   & 2  &  3 if $k\neq l, F=G$,  &  s  &  3 if $k\neq l, F=G$, &  2 \\
       & & otherwise 2 & & otherwise 2 & \\ \hline
    \end{tabular}
    \caption{$\mathtt{P}(i\#)=0$; $\mathtt{P}(j\#)=kFlG$}
\end{table}

\begin{table}[h]
    \centering
    \begin{tabular}{|c|c|c|c|}
    \hline
       $\mathtt{S}(i\#)\backslash \mathtt{S}(j\#) $  &  $\delta cs$  & $\delta cd$ & $\epsilon$\\ \hline
       $\alpha$  &  s if $k\neq l, F=G$, otherwise 1  & $\checkmark$ if $k\neq l; F=G$, otherwise 1  & 1   \\ \hline
       $\gamma$  & s if $k\neq l, F=G$, otherwise 2  &  3 if $k\neq l; F=G$, otherwise 2  &  2   \\ \hline
       $\delta cs$   &  s if $k\neq l, F=G$, otherwise 2  &  3 if $k\neq l; F=G$, otherwise 2  &  2  \\ \hline
       $\delta cd$   &  s if $k\neq l, F=G$, otherwise 2  &  3 if $k\neq l; F=G$, otherwise 2  &  2  \\ \hline
    \end{tabular}
    \caption{$\mathtt{P}(i\#)=0$;$ \mathtt{P}(j\#)=kF-lG$}
\end{table}

\begin{table}[h]
    \centering
    \begin{tabular}{|c|c|c|c|c|c|}
    \hline
       $\mathtt{S}(i\#)\backslash \mathtt{S}(j\#) $  & $\beta$ & $\gamma$ & $\delta cs$  & $\delta cd$  & $\epsilon$\\ \hline
       $\beta$  &  1  & 1  &  s  &  $\checkmark^-$, otherwise 2  &  1  \\ \hline
       $\gamma$  & \cellcolor{lightgray}  &  $\checkmark$ if $k\neq l, F=G$,   &  s  &  3 if $k\neq l, F=G$,   &  2 if $k\neq l, F=G$,  \\
       & \cellcolor{lightgray} & otherwise 1 & & otherwise 2 & otherwise 1\\ \hline
       $\delta cs$   & \cellcolor{lightgray}  &  \cellcolor{lightgray}  &  s  &  s  &  ! if $k= j$, \\ 
       & \cellcolor{lightgray} & \cellcolor{lightgray} & & & otherwise s \\\hline 
       $\delta cd$   & \cellcolor{lightgray}  &  \cellcolor{lightgray}  &  \cellcolor{lightgray}  &  3 if $k\neq l, F=G$,  &  2  \\
       & \cellcolor{lightgray} & \cellcolor{lightgray} & \cellcolor{lightgray} & otherwise 2 & \\ \hline
       $\epsilon$   & \cellcolor{lightgray}  &  \cellcolor{lightgray}  &  \cellcolor{lightgray}  &  \cellcolor{lightgray}  &  2 if $k\neq l, F=G$,\\
       & \cellcolor{lightgray}  &  \cellcolor{lightgray}  &  \cellcolor{lightgray}  &  \cellcolor{lightgray}  &  otherwise 1 \\ \hline
    \end{tabular}
    \caption{$\mathtt{P}(i\#)=kF$; $\mathtt{P}(j\#)=lG$}
\end{table}

\begin{table}[h]
    \centering
    \begin{tabular}{|c|c|c|c|c|c|}
    \hline
       $\mathtt{S}(i\#)\backslash \mathtt{S}(j\#) $  & $\beta$ & $\gamma$ & $\delta cs$  & $\delta cd$  & $\epsilon$\\ \hline
       $\beta$  &  1  & 1  &  s  &  $\checkmark^-$, otherwise 2  &  1  \\ \hline
       $\gamma$  & 1  &  1  &  s  &  2  &  1  \\ \hline
       $\delta cs$   &  s  &  s  &  s  &  s  &  s \\ \hline
       $\delta cd$   & $\checkmark^-$, otherwise 2  &  2  &  s  &  2  &  1  \\ \hline
       $\epsilon$   & 1  &  1  &  s  &  2  &  1  \\ \hline
    \end{tabular}
    \caption{$\mathtt{P}(i\#)=kF$; $\mathtt{P}(j\#)=lFmG$}
\end{table}

\begin{table}[h]
    \centering
    \begin{tabular}{|c|c|c|c|}
    \hline
       $\mathtt{S}(i\#)\backslash \mathtt{S}(j\#) $  & $\delta cs$  & $\delta cd$  & $\epsilon$\\ \hline
       $\beta$  &  s if $l\neq m, F=G$,   & $\checkmark$ if $l\neq m, F=G$,   &  1  \\
       & otherwise 1 & otherwise 1 & \\ \hline
       $\gamma$  &  s if $l\neq m, F=G$,  &  3 if $k\neq l$ or    &  1\\
       & otherwise $\checkmark$ if $k\neq l$, & $l\neq m, F=G$, &    \\
       & otherwise 1 & otherwise 1 & \\ \hline
       $\delta cs$   &  s  &  ! if $k=j, l\neq m$,  &  s   \\
       & & otherwise s &\\ \hline
       $\delta cd$   & s if $m\neq l$,   &  3 if $k\neq l$ or  &  2  \\  & otherwise 3 if $k\neq l$, & $l\neq m, F=G$ or  & 
       \\  & otherwise 2 & $k\neq m, F=G$, &\\
        & & otherwise 2\\\hline
       $\epsilon$   & 2  &  2  &  2  \\ \hline
    \end{tabular}
    \caption{$\mathtt{P}(i\#)=kF$; $\mathtt{P}(j\#)=lF-mG$}
\end{table}

\begin{sidewaystable}[htbp]
    \centering
    \begin{tabular}{|c|c|c|c|c|c|}
    \hline
       $\mathtt{S}(i\#)\backslash \mathtt{S}(j\#) $  & $\beta$ & $\gamma$ & $\delta cs$  & $\delta cd$  & $\epsilon$\\ \hline
       $\beta$  &  3 if $k\neq l,$  & 3 if $k\neq l,$  &  s  &  $\checkmark^-$,   &  3 if $k\neq l,$  \\
       & $m\neq n, F=G$, & $m\neq n, F=G$, & & otherwise 2 & $m\neq n$,  \\
       & otherwise 1 & otherwise 1 & &  & otherwise 1  \\\hline
       $\gamma$  &  \cellcolor{lightgray}  &  s if $k=l$ or  &  s  &  s if $k=l, m=n$ or  &  3 if $k\neq l, m\neq n$,  \\ & \cellcolor{lightgray} & $m=n$, &  & $F\neq G$,   & otherwise 1 \\
       & \cellcolor{lightgray} &otherwise 3 &  & otherwise $\checkmark^-$ &\\\hline
       $\delta cs$   &  \cellcolor{lightgray}  &  \cellcolor{lightgray}  &  s  &  s  &  s \\ \hline
       $\delta cd$   & \cellcolor{lightgray}  &  \cellcolor{lightgray}  &  \cellcolor{lightgray}  &  s if $k=l, m=n$, & 3 if $k\neq l, m\neq n$,\\
       & \cellcolor{lightgray}  &  \cellcolor{lightgray}  &  \cellcolor{lightgray} & otherwise 3 if  &  otherwise 2  \\ 
       & \cellcolor{lightgray}  &  \cellcolor{lightgray}  &  \cellcolor{lightgray} &  $k\neq l, m\neq n$,  &    \\ 
       & \cellcolor{lightgray}  &  \cellcolor{lightgray}  &  \cellcolor{lightgray} &  otherwise 2  &    \\\hline
       $\epsilon$   & \cellcolor{lightgray}  &  \cellcolor{lightgray}  &  \cellcolor{lightgray}  &  \cellcolor{lightgray}  &  s if $k=l, m=n$, \\
       & \cellcolor{lightgray}  &  \cellcolor{lightgray}  &  \cellcolor{lightgray}  &  \cellcolor{lightgray}  &  3 if $k\neq l, m\neq n$,\\
       & \cellcolor{lightgray}  &  \cellcolor{lightgray}  &  \cellcolor{lightgray}  &  \cellcolor{lightgray}  &  otherwise 1    \\\hline
    \end{tabular}
    \caption{$\mathtt{P}(i\#)=kFlG$; $\mathtt{P}(j\#)=mFnG$}
\end{sidewaystable}

\begin{table}[h]
    \centering
    \begin{tabular}{|c|c|c|c|}
    \hline
       $\mathtt{S}(i\#)\backslash \mathtt{S}(j\#) $  & $\delta cs$  & $\delta cd$  & $\epsilon$\\ \hline
       $\beta$  &  s if $m\neq n$,  &  $\checkmark$ if $m\neq n, F=G$,  &  1  \\ 
       & otherwise 1 & otherwise 1 &  \\\hline
       $\gamma$  & s if $m\neq n$ or  &  $\checkmark$ if $k\neq m, l\neq n$,  &  2  \\
       & $k\neq m, l\neq n$, & otherwise 1 &  \\
       & otherwise 1 &  &  \\\hline
       $\delta cs$   &  s  &  s & s   \\ \hline
       $\delta cd$   &  s if $m\neq n$, &  3 if $k\neq m, l\neq n$ or  &  2 \\
       &  2 otherwise &  $m\neq n, F=G$,  &   \\
       &  &  2 otherwise  &   \\\hline
       $\epsilon$   & s if $m\neq n, F=G$,  &  3 if $m\neq n, F=G$,  &  1 \\
       & otherwise 1  &  otherwise 1  &   \\ \hline
    \end{tabular}
    \caption{$\mathtt{P}(i\#)=kFlG$; $\mathtt{P}(j\#)=mF-nG$}
\end{table}

\begin{table}[h]
    \centering
    \begin{tabular}{|c|c|c|c|}
    \hline
       $\mathtt{S}(i\#)\backslash \mathtt{S}(j\#) $  & $\delta cs$  & $\delta cd$  & $\epsilon$\\ \hline
       $\delta cs$   &  s if $m\neq n$ or $k\neq l$ or ,  &  3 if $F=G, k=l,$  &  $\checkmark$ if $k\neq m, l\neq n$,  \\
       & $F\neq G$, otherwise $\checkmark$ if & otherwise s & otherwise 1 \\
       & $k\neq m, l\neq n$, &  & \\
       & otherwise 1 &  & \\\hline
       $\delta cd$   & \cellcolor{lightgray} & 3 if $k\neq m$ or $l\neq n$ or  &  3 if $k\neq m, l\neq n$ or  \\
       & \cellcolor{lightgray} & $F=G, k\neq l$ or  &  $F=G, k\neq l$, \\ 
       & \cellcolor{lightgray} & $F=G, m\neq n$ & otherwise 1 \\
       & \cellcolor{lightgray} &  otherwise 1 & \\\hline
       $\epsilon$   & \cellcolor{lightgray}  &  \cellcolor{lightgray}  &  1  \\ \hline
    \end{tabular}
    \caption{$\mathtt{P}(i\#)=kF-lG$; $\mathtt{P}(j\#)=mF-nG$}
\end{table}
\clearpage

\subsection{Uniqueness}

Our conservative connectives at hand, let us now check which of them are also unique. As with conservativity, we will first consider all positive cases before discussing the negative and special ones.
\subsubsection{Positive cases}
We first show horizontal interderivability for the positive cases, then establish vertical interderivability, i.e. uniqueness.
\begin{definition}[positive cases, uniqueness]\

    \noindent Let \texttt{PU} be a subset of \texttt{PC} with the following $150$ members such that $k,l,m,n \in \{1,2\}$, $F,G \in \{A, B\}$ and $a \in \mathbb{N}^0$:
     \begin{itemize}
            \item[1.1)] $\langle \alpha, \gamma 0, a, \mathbb{N}^0:\mathbb{N}^0 \rangle, \langle \gamma 0, \alpha, a, \mathbb{N}^0:\mathbb{N}^0 \rangle$ (2);
            \item[1.2)] $\langle \alpha, \delta cs 0, a, \mathbb{N}^0:\mathbb{N}^0 \rangle, \langle \delta cs 0, \alpha, a, \mathbb{N}^0:\mathbb{N}^0 \rangle$ (2);
            \item [2.1)] $\langle \gamma 0, \beta kF, a, \mathbb{N}^0:\mathbb{N}^0 \rangle, \langle \beta kF, \gamma 0, a, \mathbb{N}^0:\mathbb{N}^0 \rangle$ (8);
            \item [2.2)] $\langle \delta cs 0, \beta kF, a, \mathbb{N}^0:\mathbb{N}^0 \rangle, \langle \beta kF, \delta cs 0, a, \mathbb{N}^0:\mathbb{N}^0 \rangle$ (8);
            \item[3.1)] $\langle \alpha, \delta cd^- 1F2F, a, \mathbb{N}^0:\mathbb{N}^0 \rangle, \langle \delta cd^- 1F2F, \alpha, a, \mathbb{N}^0:\mathbb{N}^0 \rangle$ (4);
            \item [3.2)] $\langle \gamma 0, \beta kFlG, a, \mathbb{N}^0:\mathbb{N}^0 \rangle, \langle \beta kFlG, \gamma 0, a, \mathbb{N}^0:\mathbb{N}^0 \rangle$ (20);
            \item [3.3)] $\langle \gamma 0, \gamma^- 1F2F, a, \mathbb{N}^0:\mathbb{N}^0 \rangle, \langle \gamma^- 1F2F, \gamma 0, a, \mathbb{N}^0:\mathbb{N}^0 \rangle$ (4);
            \item [3.4)] $\langle \delta cs 0, \beta kFlG, a, \mathbb{N}^0:\mathbb{N}^0 \rangle, \langle \beta kFlG, \delta cs 0, a, \mathbb{N}^0:\mathbb{N}^0 \rangle$ (20);
            \item [3.5)] $\langle \delta cs 0, \gamma^- 1F2F, a, \mathbb{N}^0:\mathbb{N}^0 \rangle, \langle \gamma^- 1F2F, \delta cs 0, a, \mathbb{N}^0:\mathbb{N}^0 \rangle$ (4);
            \item[5.1)] $\langle \gamma kF, \gamma lF, a, \mathbb{N}^0:\mathbb{N}^0 \rangle, k\neq l$ (4);
            \item[6.1)] $\langle \beta kF, \delta cd^- 1F2F, a, \mathbb{N}^0:\mathbb{N}^0 \rangle, \langle \delta cd^- 1F2F, \beta kF, a, \mathbb{N}^0:\mathbb{N}^0 \rangle$ (8);
            \item[7.1)] $\langle \gamma kF, \delta cs lF-lF, a, \mathbb{N}^0:\mathbb{N}^0 \rangle, \langle \delta cs lF-lF, \gamma kF, a, \mathbb{N}^0:\mathbb{N}^0 \rangle, k\neq l$ (8);
            \item[8.1)] $\langle \beta kFlF, \delta cd^- 1F2F, a, \mathbb{N}^0:\mathbb{N}^0 \rangle, \langle \delta cd^- 1F2F, \beta kFlF, a, \mathbb{N}^0:\mathbb{N}^0 \rangle$ (12);
            \item[8.2)] $\langle \gamma^- 1F2F, \delta cd^- 1F2F, a, \mathbb{N}^0:\mathbb{N}^0 \rangle, \langle \delta cd^- 1F2F, \gamma^- 1F2F, a, \mathbb{N}^0:\mathbb{N}^0 \rangle$ (4);
            \item[9.1)] $\langle \gamma kFlG, \delta cd mF-nG, a, \mathbb{N}^0:\mathbb{N}^0 \rangle, \langle \delta cd mF-nG, \gamma kFlG, a, \mathbb{N}^0:\mathbb{N}^0 \rangle, k\neq m, l\neq n$ (20);
            \item[10.1)] $\langle \delta cs kF-kF, \delta cs mF-mF, a, \mathbb{N}^0:\mathbb{N}^0 \rangle, k\neq m$ (4);
            \item[10.2)] $\langle \delta cs kF-lG, \epsilon mF-nG, a, \mathbb{N}^0:\mathbb{N}^0 \rangle, \langle \epsilon mF-nG, \delta cs kF-lG, a, \mathbb{N}^0:\mathbb{N}^0 \rangle, k\neq m, l\neq n$ (20).
    \end{itemize}
    \end{definition}
\begin{theorem}[horizontal interderivability]\

\noindent
Let our calculus be $\{$\textsc{Id}, \textsc{Cut}$cd$, 1$\#$, 2$\#$, 1$\natural, 2\natural\}$ with $1\#=1\natural,2\#=2\natural$ modulo symbol substitution. For all $\#\in \texttt{PU}$,
$$\texttt{\#} \dashv \vdash \na,$$
where  $\texttt{\#}, \na$ are well-formed formulae whose main connective is $\#$ or $\natural$, respectively.
\end{theorem}
\begin{proof}\

\noindent 
Let $i,j,k,l\in \{1,2\}$ and $i \neq j$. Let $\F, G\in \{A, B\}$.\\ \\
Case 1) \texttt{P}$(i\#)=\texttt{P}(j\#)=0$.
\begin{enumerate}[label=1.\arabic*)]
\item $\texttt{S}(i\#)=\alpha; \texttt{S}(j\#)=\gamma$:
\[
\begin{prooftree}
\hypo{\emptyset}
\infer1[i\#]{[i:\texttt{\#}]}
\infer1[j$\natural$]{[i:\texttt{\#}; j:\na]}
\end{prooftree}\, .
\]
\item $\texttt{S}(i\#)=\alpha; \texttt{S}(j\#)=\delta cs$:
\[
\begin{prooftree}
\hypo{\emptyset}
\infer1[i\#]{[i:\texttt{\#}]}
\hypo{\emptyset}
\infer1[i\#]{[i:\texttt{\#}]}
\infer2[j$\natural$]{[i:\texttt{\#}; j:\na]}
\end{prooftree}\, .
\]
\end{enumerate}
Case 2) \texttt{P}$(i\#)=0; \texttt{P}(j\#)=kF$.
\begin{enumerate}[label=2.\arabic*)]
\item $\texttt{S}(i\#)=\gamma ; \texttt{S}(j\#)=\beta$: cf. 1.1).
\item $\texttt{S}(i\#)=\delta cs; \texttt{S}(j\#)=\beta$: cf. 1.2).
\end{enumerate}
Case 3) \texttt{P}$(i\#)=0; \texttt{P}(j\#)=kFlF$.
\begin{enumerate}[label=3.\arabic*)]
\item $\texttt{S}(i\#)=\alpha ; \texttt{S}(j\#)=\delta cd^-$, where $k\neq l, F=G$:
\[
\begin{prooftree}
    \hypo{\emptyset}
    \infer1[\textsc{Id}]{[k:F;l:F]}
    \hypo{\emptyset}
    \infer1[i\#]{[i:\texttt{\#}]}
    \infer2[j$\natural$]{[i:\texttt{\#}; j:\na]}
\end{prooftree}\, .
\]
\item $\texttt{S}(i\#)=\gamma ; \texttt{S}(j\#)=\beta$: cf. 1.1).
\item $\texttt{S}(i\#)=\gamma; \texttt{S}(j\#)=\gamma^-$, where $k\neq l, F=G$: 
\[
\begin{prooftree}
    \hypo{\emptyset}
    \infer1[\textsc{Id}]{[k:F;l:F]}
    \infer1[j$\natural$]{[j:\na]}
    \infer1[i\#]{[i:\texttt{\#}; j:\na]}
\end{prooftree}\, .
\]
\item $\texttt{S}(i\#)=\delta cs; \texttt{S}(j\#)=\beta$: cf. 1.2).
\item $\texttt{S}(i\#)=\delta cs; \texttt{S}(j\#)=\gamma^-$, where $k\neq l, F=G$: 
\[
\begin{prooftree}
    \hypo{\emptyset}
    \infer1[\textsc{Id}]{[k:F;l:F]}
    \infer1[j$\natural$]{[j:\na]}
    \hypo{\emptyset}
    \infer1[\textsc{Id}]{[k:F;l:F]}
    \infer1[j$\natural$]{[j:\na]}
    \infer2[i\#]{[i:\texttt{\#}; j:\na]}
\end{prooftree}\, .
\]
\end{enumerate}
Case 4) \texttt{P}$(i\#)=0; \texttt{P}(j\#)=kF-lG$: no unique options.\\
Case 5) \texttt{P}$(i\#)=kF; \texttt{P}(j\#)=lG$.
\begin{enumerate}[label=5.\arabic*)]
\item $\texttt{S}(i\#)=\gamma; \texttt{S}(j\#)=\gamma$, where $k\neq l, F=G$:
\[
\begin{prooftree}
    \hypo{\emptyset}
    \infer1[\textsc{Id}]{[k:F;l:F]}
    \infer1[j$\natural$]{[k:F;j:\na]}
    \infer1[i\#]{[i:\texttt{\#}; j:\na]}
\end{prooftree}\, .
\]
\end{enumerate}
Case 6) \texttt{P}$(i\#)=kF; \texttt{P}(j\#)=lFmG$.
\begin{enumerate}[label=6.\arabic*)]
\item $\texttt{S}(i\#)=\beta; \texttt{S}(j\#)=\delta cd^-$, where $l\neq m, F=G$:
\[
\begin{prooftree}
    \hypo{\emptyset}
    \infer1[\textsc{Id}]{[l:F;m:F]}
    \hypo{\emptyset}
    \infer1[i\#]{[i:\texttt{\#}]}
    \infer2[j$\natural$]{[i:\#;j:\na]}
\end{prooftree}\, .
\]
\end{enumerate}
Case 7) \texttt{P}$(i\#)=kF; \texttt{P}(j\#)=lF-mG$.
\begin{enumerate}[label=7.\arabic*)]
\item $\texttt{S}(i\#)=\gamma; \texttt{S}(j\#)=\delta cs$ with $l=m, F=G, k\neq l$:
\[
\begin{prooftree}
    \hypo{\emptyset}
    \infer1[\textsc{Id}]{[k:F;l:F]}
    \hypo{\emptyset}
    \infer1[\textsc{Id}]{[k:F;l:F]}
    \infer2[2$\natural$]{[k:F;j:\na]}
    \infer1[i\#]{[i:\texttt{\#}; j:\na]}\, .
\end{prooftree}\, .
\]
\end{enumerate}
Case 8) \texttt{P}$(i\#)=kFlG; \texttt{P}(j\#)=mFnG$.
\begin{enumerate}[label=8.\arabic*)]
\item $\texttt{S}(i\#)=\beta; \texttt{S}(j\#)=\delta cd^-$, where $m\neq n, F=G$: cf. 6.1).
\item $\texttt{S}(i\#)=\gamma^-; \texttt{S}(j\#)=\delta cd^-$, where $k\neq l, m\neq n, F=G$:
\[
\begin{prooftree}
    \hypo{\emptyset}
    \infer1[\textsc{Id}]{[k:F;l:F]}
    \hypo{\emptyset}
    \infer1[\textsc{Id}]{[k:F;l:F]}
    \infer1[i\#]{[i:\texttt{\#}]}
    \infer2[j$\natural$]{[i:\#;j:\na]}\, .
\end{prooftree}\, .
\]
\end{enumerate}
Case 9) \texttt{P}$(i\#)=kFlG; \texttt{P}(j\#)=mF-nG$.
\begin{enumerate}[label=9.\arabic*)]
\item $\texttt{S}(i\#)=\gamma; \texttt{S}(j\#)=\delta cd$ where $k\neq m, l\neq n$:
\[
\begin{prooftree}
    \hypo{\emptyset}
    \infer1[\textsc{Id}]{[k:F; m:F]}
    \hypo{\emptyset}
    \infer1[\textsc{Id}]{[l:G; n:G]}
    \infer2[j$\natural$]{[k:F; l:G; j:\na]}
    \infer1[i\#]{[i:\texttt{\#}; j:\na]}
\end{prooftree}\, .
\]
\end{enumerate}
Case 10) \texttt{P}$(i\#)=kF-lG; \texttt{P}(j\#)=mF-nG$.
\begin{enumerate}[label=10.\arabic*)]
\item $\texttt{S}(i\#)=\delta cs; \texttt{S}(j\#)=\delta cs$, where $m=n, k=l, F=G, k\neq m$:
\[
\begin{prooftree}
    \hypo{\emptyset}
    \infer1[\textsc{Id}]{[k:F; m:F]}
    \hypo{\emptyset}
    \infer1[\textsc{Id}]{[k:F; m:F]}
    \infer2[i\#]{[m:F; i:\texttt{\#}]}
    \hypo{\emptyset}
    \infer1[\textsc{Id}]{[k:F; m:F]}
    \hypo{\emptyset}
    \infer1[\textsc{Id}]{[k:F; m:F]}
    \infer2[i\#]{[m:F; i:\texttt{\#}]}
    \infer2[j$\natural$]{[i:\texttt{\#}; j:\na]}
\end{prooftree}\, .
\]
\item $\texttt{S}(i\#)=\delta cs; \texttt{S}(j\#)=\epsilon$, where $k\neq m, l\neq n$:
\[
\begin{prooftree}
    \hypo{\emptyset}
    \infer1[\textsc{Id}]{[k:F;m:F]}
    \infer1[j$\natural$]{[k:F;j:\na]}
    \hypo{\emptyset}
    \infer1[\textsc{Id}]{[l:G;n:G]}
    \infer1[j$\natural$]{[l:G;j:\na]}
    \infer2[i\#]{[i:\texttt{\#};j:\na]}\, .
\end{prooftree}
\]
\end{enumerate}
\end{proof}
We can now establish vertical interderivability.
\begin{corollary}[uniqueness]\

    \noindent For all $\#\in \texttt{PU}$, $\{$\textsc{Id}, \textsc{Cut}$cd$, 1$\#$, 2$\#\}$ uniquely determines the inferential role of $\#$ up to isomorphism.
\end{corollary}

\begin{proof}[vertical interderivability]\

\noindent
By virtue of Definition 16, we must prove for any $i\in\{1,2\}$ and any $i\natural=i\#$, modulo symbol substitution:
    \[
        \begin{prooftree}
        \hypo{[i:\texttt{\#}]s}
        \infer[]1{[i:\na]s}
        \end{prooftree}\qquad \text{and} \qquad
        \begin{prooftree}
        \hypo{[i:\na]s}
        \infer[]1{[i:\texttt{\#}]s}
        \end{prooftree}
    \]
    in $\{$\textsc{Id}, \textsc{Cut}$cd$, 1$\#$, 2$\#$, 1$\natural$, 2$\natural\}$. Using Theorem 5, we obtain:
\begin{itemize}
    \item from top to bottom:\[
\begin{prooftree}
\hypo{[i:\texttt{\#}]s}
\hypo{[i:\na; j:\texttt{\#}]}
\infer2[\textsc{Cut}]{[i:\na]s}
\end{prooftree}\, ,
\]
\item from bottom to top:
\[
\begin{prooftree}
\hypo{[i:\na]s}
\hypo{[i:\texttt{\#}; j:\na]}
\infer2[\textsc{Cut}]{[i:\texttt{\#}]s}
\end{prooftree}\, .
\]\hfill $\blacksquare$
\end{itemize}
\end{proof}

\subsubsection{Negative cases}
Let us now show that \texttt{PU} comprises all (non-vacuously) unique and conservative formulable \#. Specifically, we show that for all $\#$, if $\#\in \texttt{PC}$ but $\#\not{\in} \texttt{PU}$, then $\{$\textsc{Id}, \textsc{Cut}$cd$, 1$\#$, 2$\#\}$ does not uniquely determine the inferential role of $\#$ up to isomorphism.
\begin{theorem}[horizontal non-interderivability]\

    \noindent For all $\#$, if $\#\in \texttt{PC}$ but $\#\not{\in} \texttt{PU}$, then $[i:\texttt{\#};j:\na]$ is not provable in $\{$\textsc{Id}, \textsc{Cut}$cd$, 1$\#$, 2$\#\}$, where $i,j\in\{1,2\}, i\neq j$.
\end{theorem}
\begin{proof}\

\noindent
    We can establish the failure of horizontal interderivability via a bottom-up proof search on $[i:\texttt{\#};j:\na]$. Since $\#\in \texttt{PC}$, \textsc{Cut}-elimination is at our disposal. This significantly contains the proof searche, since any proof can only use \textsc{Id}, 1$\#$, or 2$\#$. Consider, for example, the following rules of types $\alpha$ and $\delta cd kF-lF$, where $k\neq l$:
    \[
    \begin{prooftree}
        \hypo{\emptyset}
        \infer1[i\#]{[i:\texttt{\#}]}
    \end{prooftree}\, , \qquad 
    \begin{prooftree}
        \hypo{[k:F]s}
        \hypo{[l:F]t}
        \infer2[j\#]{[j:\texttt{\#}]st}
    \end{prooftree}\, .
    \]
    The final proof step could not have been \textsc{Id} as $\texttt{\#},\na$ are non-primitive. It also could not have been i\# as $\texttt{\#}$ cannot occur context-free in the $i$th component of the conclusion sequent. Thus, it must have been j$\natural$:
    \[
    \begin{prooftree}
        \hypo{~~~~~~~~~}
            \infer[rule code={\hbox{\tikz\draw (0,0) -- (\hsize,0) -- (0.5\hsize,-2em) (0.5\hsize,-0.875em) node{\small $\mathfrak{p}_1$}  (0.5\hsize,-2em) -- (0,0);}}]1{[k:F]}
        \hypo{~~~~~~~~~}
        \infer[rule code={\hbox{\tikz\draw (0,0) -- (\hsize,0) -- (0.5\hsize,-2em) (0.5\hsize,-0.875em) node{\small $\mathfrak{p}_2$}  (0.5\hsize,-2em) -- (0,0);}}]1{[l:F;i:\texttt{\#}]}
        \infer2[j$\natural$]{[i:\texttt{\#};j:\na]}
    \end{prooftree}\, .
    \]
    If $F\neq \texttt{\#}, F\neq \na$, we could not have applied any of our rules in either branch. Hence, $F=\na$ (or, equivalently, $F=\texttt{\#}$).
    \begin{enumerate}
        \item Let $j=k$. Then $l=i$ and $[l:F; i:\texttt{\#}]=[i:\texttt{\#}, \na]$. However, the latter sequent is not provable as $\#$-/$\natural$-formulae can only be introduced in the $i$th dimension using $i\#$ and $i\natural$, respectively. However, either rule requires empty contexts in the conclusion sequent, which conflicts with the given.
        \item Let $i=k$. Although $\mathfrak{p}_1$ could have been an application of i$\natural$, $[l:F;i:\texttt{\#}]=[i:\texttt{\#};j:\na]$. Hence, the proof search for $\mathfrak{p}_2$ would repeat our proof search thus far, proceeding \textit{ad infinitum}:
    \[
    \begin{prooftree}
        \hypo{\emptyset}
        \infer1[i$\natural$]{[i:\na]}
        \hypo{\emptyset}
        \infer1[i$\natural$]{[i:\na]}
        \hypo{\emptyset}
        \infer1[i$\natural$]{[i:\na]}
        \hypo{\emptyset}
        \infer1[i$\natural$]{[i:\na]}
        \hypo{\vdots}
        \infer1[j$\natural$]{[i:\texttt{\#};j:\na]}
        \infer2[j$\natural$]{[i:\texttt{\#};j:\na]}
        \infer2[j$\natural$]{[i:\texttt{\#};j:\na]}
        \infer2[j$\natural$]{[i:\texttt{\#};j:\na]}
        \infer2[j$\natural$]{[i:\texttt{\#};j:\na]}
    \end{prooftree}\, .
    \]
    We preclude infinite proof searches, as infinite proof trees---while technically possible---are typically considered malformed. 
    \end{enumerate}
    The remaining proof searches are left to the reader. 
    \hfill $\blacksquare$
\end{proof}  
This proof alone does not establish uniqueness as set out in Definition 16. To meet the definition's standard, we must show that \textit{vertical} interderivability fails, i.e. for some $s$ and some $i\in\{1,2\}$, $[i:\texttt{\#}]s$ cannot be derived from $[i:\na]s$, or vice versa. That is, we must prove that there is no proof tree in $\textbf{Cal}=\{1\#, 2\#, \textsc{Id}, \textsc{Cut}\}$ with $[j:\texttt{\#}]s$ in its root and only $[j:\na]s$ or \textbf{Cal}-axioms in its leaves.

As before, finding such a counterexample depends on the \#-rules under consideration. Unlike in the case of horizontal interderivability, finding the right counterexample---specifically a fitting $s$---is a non-trivial task. After fixing some specific $s$, one must perform an exhaustive bottom-up search for a derivation of $[i:\texttt{\#}]s$ from $[i:\na]s$. We cannot rely on \textsc{Cut}-elimination here, since we have an additional premiss, $[i:\na]s$, and all its \textbf{Cal}-consequences at our disposal. Using these may require \textsc{Cut} in an irreducible way. Furthermore, results will depend on whether any given connective satisfies subformula-connectedness and, if so, what the subformulae of $[i:\texttt{\#}]$ are.

Based on partial proof search attempts for $s$ for some of our $\#\in \texttt{PC}$ but $\#\not{\in} \texttt{PU}$, we conjecture that the right $s$ could be found for each case. However, the complexity of the search prevents us from establishing this result here; we leave it to future research.
\begin{conjecture}[vertical interderivability]\

    \noindent For all $\#$, if $\#\in \texttt{PC}$ but $\#\not{\in} \texttt{PU}$, then $\#$ is not vertically interderivable. Hence, it is non-unique according to Definition 16. 
\end{conjecture}
How significant is this limitation for our overall argument? As we have shown in \S 5.2.1, horizontal interderivability implies vertical interderivability in the presence of $\textsc{Cut}cd$. Thus, our ability to establish positive cases is not impaired. However, negative results cannot be established by relying on the structural rules of the underlying calculus alone.

 Moreover, it is not clear that vertical interderivability is the optimal concept for capturing uniqueness in our setting. Uniqueness is a philosophical notion. It relies on the philosophical concept of the \textit{inferential role} of a connective (cf. Definition 14). In using vertical interderivability as our uniqueness criterion (cf. Definition 16), we follow Nagler \cite{nagler2026measuring}, who argues that Dicher's \cite{dicher2016proof} standard of horizontal interderivability is too weak for \cite{belnap1962tonk}-style uniqueness in a substructural setting.
 
 In fact, it is not \textit{immediately} obvious how to translate Belnapian uniqueness into the context of I-bS since Belnap \cite{belnap1962tonk} works with a single-conclusion natural deduction system (\textbf{ND}), whereas I-bS is built on potentially multi-conclusion sequent calculus machinery. In the setting of \cite{belnap1962tonk}, a connective \# is unique in \textbf{ND} iff
\begin{enumerate}
    \item $\Gamma, \texttt{\#} \vdash_{\textbf{ND}} D$ iff $\Gamma, \na \vdash_{\textbf{ND}} D$, and
    \item $\Gamma \vdash_{\textbf{ND}} \texttt{\#}$ iff $\Gamma \vdash_{\textbf{ND}} \na$.
\end{enumerate}
As Belnap's system is reflexive, transitive, monotonic and defined on sets of formulae, uniqueness can be established by proving $\texttt{\#}\vdash_{ND} \na$ and $\na \vdash_{ND} \texttt{\#}$. 

There are four natural ways to import Belnap's definition into our sequent calculus context: a connective \# is unique in the sequent calculus \textbf{Cal} iff for any $i,j\in\{1,2\}, i\neq j$,
\begin{enumerate}
    \item $[i:\texttt{\#};j:\na]$ is provable in \textbf{Cal},
    \item $[i:\texttt{\#};j:\na]s$ is provable in \textbf{Cal},
    \item $[i:\na]s$ is derivable from $[i:\texttt{\#}]s$ in \textbf{Cal}, and vice versa, i.e. 
    \[\begin{prooftree}
        \hypo{[i:\texttt{\#}]s}
        \infer1{[i:\na]s}
    \end{prooftree} \quad \text{and} \quad \begin{prooftree}
        \hypo{[i:\na]s}
        \infer1{[i:\texttt{\#}]s}
    \end{prooftree}\, ,
    \]
    \item $[i:\texttt{\#}]s$ is provable in \textbf{Cal} iff $[i:\na]s$ is provable in \textbf{Cal}.
\end{enumerate}
Contextualised horizontal interderivability, 2., is a generalised version of 1., i.e. Dicher's notion \cite{dicher2016proof}. The failure of both 1. and 2. follows from Theorem 6. As argued in \cite{nagler2026measuring}, 1. (or 2.) is only equivalent to Belnap's \textit{official} definition in the presence of fully Tarskian structural resources, i.e. reflexivity, transitivity, monotonicity and a set-based sequent framework. In our substructural setting, however, this equivalence breaks down. Nevertheless, Belnap \cite{belnap1962tonk} does not consider such a setting. It remains unclear how his notion was intended to carry over to our scenario. Indeed, Belnap \cite{belnap1962tonk} exclusively works with horizontal interderivability, which still leaves 1. and 2. as contenders for the right notion of uniqueness.

In contrast, vertical interderivability, i.e. 3., is provably equivalent to the Belnap's \cite{belnap1962tonk} official notion, thus justifying our choice in Definition 16. Although 4. seems the most intuitive translation of Belnapian uniqueness, 3. and 4. are \textit{not} equivalent. While the two definitions may be equivalent in \textit{some} substructural settings, they are not equivalent in the present one, since we have restricted \textsc{Id} to primitive formulae. One example of this is our special cases: since at least one \#-rule cannot be used, both sides of the biconditional in 4. are false, rendering 4. vacuously true. The same is not true for 3.: as it requires the construction of a proof from the \textit{assumption} of $[i:\texttt{\#}]s$ to $[i:\na]s$, no such vacuous result can be obtained.\footnote{In a sense, 3. relies on constructive equivalence; 4. on classical equivalence.} 

While we favour vertical interderivability as the preferred notion of (Belnap-style) uniqueness, we cannot make a definitive philosophical choice here. However, we hope to have illuminated the philosophical landscape sufficiently for the reader to form their own assessment.\footnote{See also the related discussion on synonymy in \cite[579f.]{humberstone2011connectives}.}

\subsubsection{Special cases}
In \S5.1.3, we discussed some trivially conservative cases, for which adding \#-rules does not increase the set of provable \#-free sequents. Despite being trivially conservative, they are non-unique.
\begin{corollary}\

    \noindent Special cases are not horizontally unique. 
\end{corollary}
\begin{proof}\

\noindent
    As we have seen in Examples 8 and 9, our special cases arise because at least one \#-rule cannot be applied in ${1\#, 2\#, \textsc{Id}, \textsc{Cut}}$. To prove $[i:\texttt{\#};j:\na]$, one must be able to introduce $\texttt{\#}$ at least once in the $i$th component. In our calculus, only $i\#$ can accomplish initial \#-introduction into the $i$th dimension. However, However, $i\#$ is inapplicable for at least one value of $i\in{1,2}$.\footnote{This proof presumes subformula-connectedness. Only if one were drop this requirement---such as in proof step 2b of Example 9---would a manual bottom-up proof search be necessary. Since \textsc{Cut}-eliminability is not available, proofs are somewhat more cumbersome, albeit doable.}
    \hfill $\blacksquare$
\end{proof}
\begin{conjecture}\

    \noindent Special cases are not vertically unique. 
\end{conjecture}
To prove this conjecture, we must show for some $s$ and some $i\in\{1,2\}$, that $[i:\texttt{\#}]s$ cannot be derived from $[i:\na]s$, or vice versa. That is, we must prove that there is no proof tree in $\textbf{Cal}=\{1\#, 2\#, \textsc{Id}\textsc{Cut}\}$ with $[i:\texttt{\#}]s$ in its root and only $[i:\na]s$ or \textbf{Cal}-axioms in its leaves. We know that $i\#$ is \textbf{Cal}-inapplicable for at least one value of $i$. Assume $i$ takes that value. Then, as there is no \textbf{Cal}-proof of $[i:\na]s$, we must have derived it from $[i:\texttt{\#}]s$. Now, choose some $s$ such that $[i:\texttt{\#}]s$ cannot be a premiss of $i\natural$. Hence, $[i:\na]s$ cannot be the root of such a proof tree as no $\natural$-formula can be introduced. We again conjecture that such an $s$ exists for each special case, but leave this to future research.

We have now generated the data for our measurement of inference behaviour, i.e. Belnap-style definability proofs in the minimal derivability relation ${\textsc{Id}^{pr}, \textsc{Cut}cd}$. We have found that 150 formulable connective rule pairs are definable in this setting; these are collected in \texttt{PU}. We have argued that these results are exhaustive, i.e. that the $\#\in \texttt{PU}$ are the only operational rule pairs that are formulable (with $\leq 2$ premiss sequents and $\leq 2$ active formulae) and definable in our minimal setting. Let us now analyse these data by measuring the meaning of our connectives in terms of their minimal inference behaviour.

\section{Data analysis}

\subsection{All$^\ast$ semantic clauses}

What is the meaning of our minimally definable formulable connectives $\# \in \texttt{PU}$? To answer this question, we must measure their inference behaviour in their minimal definability proofs (cf. \S 4). Let us first demonstrate by example how to measure minimal inference behaviour in our dataset and how to give semantic clauses. We then present the results for all remaining cases.
\begin{example}\

    \noindent Consider the rules with premiss type $\texttt{SP}(i\#)=\alpha; \texttt{SP}(j\#)=\gamma 0$. In Theorem 5, we have shown that connectives of this type are horizontally interderivable using proof $\mathfrak{p}=$
    \[
    \begin{prooftree}
       \hypo{\emptyset^0}
       \infer1[i\#]{[i:\texttt{\#}]^1}
       \infer1[j$\natural$]{[i:\texttt{\#};j:\na]^2}
    \end{prooftree}\, ,
    \]
    for $i,j\in \{1,2\}, i\neq j$.
    
    As in Example 7, we label tree nodes using superscript numerals. By fixing values for $i,j$, we can assess the inference behaviour in $\mathfrak{p}$: for $i=1$: $\mathtt{ib}^\#_\mathfrak{p}= \langle [1], [1,2] \rangle$, and for $i=2$: $\mathtt{ib}^\#_\mathfrak{p}= \langle [2], [1,2] \rangle$. Recall that \# is unique and, allowing up to $\texttt{\#}:=\na$. Here, the inference behaviour profile $\texttt{ibp}_\mathfrak{p}^\#$ is $\{ [1], [1,2] \}$ and $\{ [2], [1,2] \}$, respectively.
\end{example}
    As indicated in \S 4.1, a (minimally meaningful) connective is individuated by its (proof-theoretic) \textit{semantic clause}, a notion adopted from Nagler \cite{nagler2026measuring}.
    \begin{definition}[semantic clause]\
        
        \noindent The semantic clause of a connective \# is the \textit{minimal} operational rule pair that yields the minimal inference behaviour of \#.
    \end{definition}
    
    Given a connective's minimal inference behaviour, how can we find its \textit{minimal} rules? In Definition 10, we defined connectives in terms of their \textit{premiss type}, \textit{framework}, and \textit{arity}. We will discuss the minimality of premiss types in \S 6.3. Provisionally, we assume that all minimally formulable premiss types are of equal size. After all, only $\# \in \texttt{PU}$ \textit{have} minimal inference behaviour, in the sense that only they are definable in our minimal derivability relation. Note further, that the proofs in \S5 do not have any framework and arity constraints; we surveyed only \texttt{SP}-configurations. We now discuss how to determine \textit{minimality} for framework and arity. 
    
    We argue that the \textit{minimal framework} for a connective is the \textit{smallest} framework within which we can prove its definability in the minimal derivability relation.
    
    \begin{definition}[framework minimality]\

        \noindent Let $\mathfrak{C}$ be a set of connectives with the same premiss type and arity.
        
        $\#\in \mathfrak{C}$ is \textit{minimal} iff it has the smallest framework.
    \end{definition}
    A connective's smallest framework can be inferred directly from its minimal inference behaviour profile. For instance, in Example 10, if $\texttt{SP}(1\#)=\alpha; \texttt{SP}(2\#)=\gamma 0$, the smallest framework is $\{1\}:\{0,1\}$ since we must accommodate one formula in the first component with none in the second ($[1]$), as well as one in each component ($[1,2]$). Equivalently, if $\texttt{SP}(1\#)=\gamma 0; \texttt{SP}(2\#)=\alpha$, the minimal framework is $\{0,1\}:\{1\}$. If we were to choose any smaller framework, the proof would fail.

    Therefore, the proof in Example 10 yields two minimal connectives with the respective semantic clauses:
    \begin{enumerate}
        \item $\langle \alpha, \gamma 0, a, \{1\}:\{0,1\} \rangle$,
        \item $\langle \gamma 0, \alpha, a, \{0,1\}:\{1\} \rangle$.
    \end{enumerate}
    We can accommodate the corresponding conservativity proof using the respective minimal framework. Focussing exemplarily on the key proof step in Lemma 1, we obtain for the former connective (with context $s=[1:C]$):
    \[
    \begin{prooftree}
        \hypo{\emptyset}
        \infer1[1\#]{[1:\texttt{\#}]}
        \hypo{[1:C]}
        \infer1[2\#]{[1:C; 2:\texttt{\#}]}
        \infer2[\textsc{Cut}]{[1:C]}
        \end{prooftree}
\]
can be reduced to
\[
    \begin{prooftree}
        \hypo{[1:C]}
        \infer1{[1:C]}
    \end{prooftree}\, ,
    \]
    and for the latter connective (with $s=[2:C]$):
    \[
    \begin{prooftree}
        \hypo{[2:C]}
        \infer1[1\#]{[1:\texttt{\#}; 2: C]}
        \hypo{\emptyset}
        \infer1[2\#]{[2:\texttt{\#}]}
        \infer2[\textsc{Cut}]{[2: C]}
    \end{prooftree}
\]
can be reduced to
\[
    \begin{prooftree}
        \hypo{[2: C]}
        \infer1{[2: C]}
    \end{prooftree}\, .
    \]
    Let us now consider three options for determining the \textit{minimal arity} ($a$) of a connective in its semantic clause. The first, arguably most intuitive, option is to set $a$ to the smallest available number, i.e. $a=0$. Although minimal definability proofs fail if we choose too small a framework and succeed if we expand the framework to a suitable size, reducing or increasing $a$ does not affect the success of our minimal definability proofs, for any $\#\in \texttt{PU}$. Hence, arity is \textit{independent} of minimal inference behaviour, rendering merely numeric minimalism unsuitable for our purposes.

    A second option would be to base the arity of a connective on the number of active formulae in its operational rules. This is a common approach in Logic, especially if one assumes subformula-connectedness. Under this standard of \textit{arity minimalism}, the two connectives in Example 10 correspond to $\top$ and $\bot$:
    \begin{enumerate}
        \item $\bot_{min}=\langle \alpha, \gamma 0, a, \{1\}:\{0,1\} \rangle$,
        \item $\top_{min}=\langle \gamma 0, \alpha, a, \{0,1\}:\{1\} \rangle$.
    \end{enumerate}
    However, subformula-connectedness was not required for our definability proofs to succeed. In fact, we have not made any assumptions about the structure of the principal formula in (at least) the positive part of our results. There is thus no good inner-theoretical reason to exclude connectives that violate subformula-connectedness. Indeed, one might consider it a feature that I-bS can provide semantics for operational rules with built-in \textsc{Contractions} or \textsc{Cut}s at the subformula level. 
    
    We therefore adopt a third option: leaving the arity and general structure of the principal formula unconstrained ($a\in \mathbb{N}^0$). That said, we will give an additional subformula-connectedness-adjusted version of each of our resulting semantic clauses as these \textit{special versions} can aid understanding of our findings.\footnote{Unlike \cite{nagler2026measuring}, we disregard the fourth parameter of operational rules in this paper: adding structural rules to get previously failing definability proofs to work (cf. \S 4.1). If one were to implement this feature, however, one might want to absorb the minimally necessary structural rules into their minimal operational rules. The resulting semantic clauses would be akin to highly substructural \textbf{G3}-style rules; such an endeavour would require a suitable notion of minimality for added structural rules, which we must leave to future work.}

    Before we apply our newfound notions of framework and arity minimality to give semantic clauses for all $\#\in\texttt{PU}$, let us consider a special case.
\begin{example}\

    \noindent Let $\texttt{SP}(i\#)=\gamma kF; \texttt{S}(j\#)=\gamma lF$, where $i,j,k,l\in\{1,2\}, k\neq l, i\neq j$. In Theorem 5, we have established horizontal interderivability via proof $\mathfrak{p}$:
    \[
\begin{prooftree}
    \hypo{\emptyset^0}
    \infer1[\textsc{Id}]{[k:F;l:F]^1}
    \infer1[j$\natural$]{[k:F;j:\na]^2}
    \infer1[i\#]{[i:\texttt{\#}; j:\na]^3}
\end{prooftree}\, .
\]
Naturally, the inference behaviour in $\mathfrak{p}$ depends on the values of our variables $i,j,k,l$:
\begin{itemize}
    \item for $i=k$, $\mathtt{ib}^\#_\mathfrak{p}= \langle  [1,2], [1,2], [1,2] \rangle$ and $\mathtt{ibp}^\#_\mathfrak{p}= \{  [1,2] \}$,
    \item for $i=1; k=2$, $\mathtt{ib}^\#_\mathfrak{p}= \langle  [1,2], [2,2], [1,2] \rangle$ and $\mathtt{ibp}^\#_\mathfrak{p}= \{  [1,2], [2,2] \}$,
    \item for $i=2; k=1$, $\mathtt{ib}^\#_\mathfrak{p}= \langle  [1,2], [1,1], [1,2] \rangle$ and $\mathtt{ibp}^\#_\mathfrak{p}= \{  [1,2], [1,1] \}$.
\end{itemize}
For now, we focus on the latter two cases ($i\neq k$). Recall that we must prove horizontal interderivability in both directions, i.e. both $[1:\texttt{\#}; 2:\na]$ and $[1:\na; 2:\texttt{\#}]$. Hence, $\mathfrak{p}$ is one of 4 proof pairs, $\mathfrak{p}_1, \mathfrak{p}_2, \mathfrak{p}_3, \mathfrak{p}_4$, such that:
\[
\mathfrak{p}_1=
\begin{prooftree}
    \hypo{\emptyset}
    \infer1[\textsc{Id}]{[1:F;2:F]}
    \infer1[2$\natural$]{[2:F,\na]}
    \infer1[1\#]{[1:\texttt{\#}; 2:\na]}
\end{prooftree}\, ,
\begin{prooftree}
    \hypo{\emptyset}
    \infer1[\textsc{Id}]{[1:F;2:F]}
    \infer1[2$\#$]{[2:F,\texttt{\#}]}
    \infer1[1$\natural$]{[1:\na; 2:\texttt{\#}]}
\end{prooftree}\, ;
 \qquad \mathfrak{p}_2=
 \begin{prooftree}
    \hypo{\emptyset}
    \infer1[\textsc{Id}]{[1:F;2:F]}
    \infer1[2$\natural$]{[2:F,\na]}
    \infer1[1\#]{[1:\texttt{\#}; 2:\na]}
\end{prooftree}\, ,
\begin{prooftree}
    \hypo{\emptyset}
    \infer1[\textsc{Id}]{[1:F;2:F]}
    \infer1[1$\na$]{[1:F,\na]}
    \infer1[2\#]{[1:\na; 2:\texttt{\#}]}
\end{prooftree}\, ;
\]
\[
\mathfrak{p}_3=
\begin{prooftree}
    \hypo{\emptyset}
    \infer1[\textsc{Id}]{[1:F;2:F]}
    \infer1[1\#]{[1:F,\texttt{\#}]}
    \infer1[2$\natural$]{[1:\texttt{\#}; 2:\na]}
\end{prooftree}\, ,
\begin{prooftree}
    \hypo{\emptyset}
    \infer1[\textsc{Id}]{[1:F;2:F]}
    \infer1[2\#]{[2:F,\texttt{\#}]}
    \infer1[1$\natural$]{[1:\na; 2:\texttt{\#}]}
\end{prooftree}\, ;
 \qquad \mathfrak{p}_4=
 \begin{prooftree}
    \hypo{\emptyset}
    \infer1[\textsc{Id}]{[1:F;2:F]}
    \infer1[1\#]{[1:F,\texttt{\#}]}
    \infer1[2$\natural$]{[1:\texttt{\#}; 2:\na]}
\end{prooftree}\, ,
\begin{prooftree}
    \hypo{\emptyset}
    \infer1[\textsc{Id}]{[1:F;2:F]}
    \infer1[1$\natural$]{[1:F,\na]}
    \infer1[2\#]{[1:\na; 2:\texttt{\#}]}
\end{prooftree}\, .
\]

Each option has distinct inference behaviour. However, $\mathfrak{p}_2$ and $\mathfrak{p_3}$ do not yield a semantic clause, as they require a more-than-minimal framework of $\{1,2\}:\{1,2\}$ ($\texttt{ibp}_{\mathfrak{p_2}}^\#=\texttt{ibp}_{\mathfrak{p_3}}^\#=\{[1,2], [1,1], [2,2]\}$). In contrast, $\mathfrak{p}_1$ and $\mathfrak{p}_4$ use strictly smaller frameworks of $\{1\}:\{1,2\}$ ($\texttt{ibp}_{\mathfrak{p}_1}^\#=\{[1,2], [2,2]\}$), and $\{1,2\}:\{1\}$ ($\texttt{ibp}_{\mathfrak{p}_4}=\{[1,2], [1,1]\}$), respectively.

However, $\mathfrak{p}_1$ and $\mathfrak{p}_4$ have distinct frameworks of equally minimal size. This results in an anomaly for our dataset: \textit{one} premiss type corresponds to \textit{two} distinct semantic clauses and therefore to two distinct meaningful connectives:
\begin{enumerate}
    \item $\langle \gamma 2F, \gamma 1F, a, \{0,1\}:\{1,2\} \rangle$,
    \item $\langle \gamma 2F, \gamma 1F, a, \{1,2\}:\{0,1\} \rangle$.
\end{enumerate}
We included the option of $0$ formulae in the first and second component, respectively, as it is required in the corresponding conservativity proof. For example, consider the corresponding step in the proof of Lemma 1 for the latter connective:
\[
\begin{prooftree}
    \hypo{[1:C; 2:F]^1}
    \infer1[1\#]{[1:\texttt{\#}, C]^3}
    \hypo{[1:F,D]^2}
    \infer1[2\#]{[1:D; 2:\texttt{\#}]^4}
    \infer2[\textsc{Cut}]{[1:C, D]^5}
\end{prooftree}
\]
can be reduced to
\[
\begin{prooftree}
    \hypo{[1:C; 2:F]}
    \hypo{[1:F, D]}
    \infer2[\textsc{Cut}]{[1: C, D]}
\end{prooftree}\, .
\]
Although we could try to eliminate the empty second component in $V=2$ by fixing $t=[1:D; 2:E]$ (rather than $t=[1:D]$), this would yield the sequent $[1:D;2:\texttt{\#}, E]$ in the second component. The resulting framework of $\{1,2\}:\{1,2\}$ would be larger than $\{1,2\}:\{0,1\}$. Therefore, the latter framework is indeed \textit{minimal}.

Stipulating subformula-connectedness and $a=1$, we can see that each semantic clause corresponds to a minimal version of \textit{negation}, matching the result for classical negation in \cite{nagler2026measuring}. Moreover, Nagler \cite{nagler2026measuring} finds that the individual versions of negation correspond to (a minimal version of) intuitionistic negation and of dual-intuitionistic negation, respectively. By replicating this result in a fully calculus-independent fashion, we \textit{validate} and \textit{verify} these findings: while Nagler \cite{nagler2026measuring} shows that intuitionistic, dual-intuitionistic and classical negation \textit{have} semantic clauses, we show that they are the \textit{only} (minimally) meaningful forms of negation according to I-bS.

We obtain the semantic clause for the only other unary connective by setting $i=k$:
\begin{enumerate}
    \item[3.] $\langle \gamma 1F, \gamma 2F, a, \{1\}:\{1\} \rangle$.
\end{enumerate}
We can think of 3.---the converse of negation---as the \textit{trivial} connective: we can add it to any formula at any point. A potential parallel is the use of `!' in informal written online communication as in the move from `I am right!' to `I am right!!!'. Although adding `!'s is certainly semantically \textit{meaningful}, it is also highly \textit{uninformative}. This may explain why one rarely encounters this connective in logical treatises.
\end{example}
\subsection{21 minimally meaningful connectives}
In addition to the five semantic clauses discussed in Examples 10 and 11, we find a further 16. As promised in \S 1, each of these corresponds to a familiar binary connective, or its inverse or converse (cf. Table 1).

Due to space constraints, we will only state the connective, as individuated by its semantic clause. We then give the minimal horizontal interderivability proof and the core part of the conservativity proof (Lemma 1), following the format of the previous examples. To ease readability, we use boxes to highlight the active and principal formulae at each \#-rule application. Lastly, we present the subformula-connectedness-adjusted \textit{special version} of the connective.

Let $F,G \in \{A, B\}$, and $a \in \mathbb{N}^0$. We obtain the following semantic clauses for $\#_{min}$:
\begin{enumerate}
\item Let \texttt{SP}($i\#)=\alpha$, \texttt{SP}($j\#)=\gamma 0$.
    \begin{enumerate}
    \item for $i=1; j=2, \#_{min}=\langle \alpha, \gamma 0, a, \{1\}:\{0,1\} \rangle$: 

        Lemma 1: let $s=[1:C]$, then
            \[
            \begin{prooftree}
                \hypo{\emptyset}
                \infer1[1\#]{[1:\boxed{\texttt{\#}}]}
                \hypo{[1:C]}
                \infer1[2\#]{[1:C; 2:\boxed{\texttt{\#}}]}
                \infer2[\textsc{Cut}]{[1:C]}
            \end{prooftree}
\]
can be reduced to
\[
            \begin{prooftree}
                \hypo{[1:C]}
                \infer1{[1:C]}
            \end{prooftree}\, .
            \]
        Lemma 3:
            \[
            \begin{prooftree}
                \hypo{\emptyset}
                \infer1[1\#]{[1:\boxed{\texttt{\#}}]}
                \infer1[2$\natural$]{[1:\texttt{\#}; 2: \boxed{\na}]}
            \end{prooftree}\, .
            \]
        Special version: $\bot_{min} = \langle \alpha, \gamma 0, 0, \{1\}:\{0,1\} \rangle$.
    \item for $i=2; j=1, \#_{min}=\langle \gamma 0, \alpha, a, \{1, 0\}:\{1\} \rangle$:

        Lemma 1: $s=[2:C]$
            \[
            \begin{prooftree}
                \hypo{\emptyset}
                \infer1[2\#]{[2:\boxed{\texttt{\#}}]}
                \hypo{[2:C]}
                \infer1[1\#]{[1:\boxed{\texttt{\#}}; 2: C]}
                \infer2[\textsc{Cut}]{[1: C]}
            \end{prooftree}
\]
can be reduced to
\[
            \begin{prooftree}
                \hypo{[2: C]}
                \infer1{[2: C]}
                \end{prooftree}\, .
            \]
        Lemma 3:
            \[
            \begin{prooftree}
                \hypo{\emptyset}
                \infer1[2$\natural$]{[2: \boxed{\na}]}
                \infer1[1\#]{[1:\boxed{\texttt{\#}}; 2: \boxed{\na}]}
            \end{prooftree}\, .
            \]
        Special version: $\top_{min} = \langle \alpha, \gamma 0, 0, \{1, 0\}:\{1\}  \rangle$.
    \end{enumerate}
\item  Let \texttt{SP}($i\#)=\gamma kF$, \texttt{SP}($j\#)=\gamma lF$.
    \begin{enumerate}
        \item for $i=k=1, j=l=2, \#_{min}= \langle \gamma 1F, \gamma 2F, a, \{1\}:\{1\} \rangle:$ 

            Lemma 1: $s=[2:C]; t=[2:D]$
                \[
                \begin{prooftree}
                    \hypo{[1:\boxed{F}; 2:C]}
                    \infer1[1\#]{[1:\boxed{\texttt{\#}}; 2: C]}
                    \hypo{[1:D; 2:\boxed{F}]}
                    \infer1[2\#]{[1:D; 2:\boxed{\texttt{\#}}]}
                    \infer2[\textsc{Cut}]{[1:D; 2: C]}
                \end{prooftree}
\]
can be reduced to
\[
                \begin{prooftree}
                    \hypo{[1:F; 2:C]}
                    \hypo{[1:D; 2:F]}
                    \infer2[\textsc{Cut}]{[1:D; 2: C]}
                \end{prooftree}\, .
                \]
            Lemma 3:
                \[
                \begin{prooftree}
                    \hypo{\emptyset}
                    \infer1[\textsc{Id}]{[1:F; 2:\boxed{F}]}
                    \infer1[2$\natural$]{[1:\boxed{F}; 2: \boxed{\na}]}
                    \infer1[1\#]{[1: \boxed{\texttt{\#}}; 2: \boxed{\na}]}
                \end{prooftree}\, .
                \]

            Special version: $\neg^{cp/cv}_{min} = \langle \gamma 1A, \gamma 2A, 1, \{1\}:\{1\} \rangle$.
        \item for $k=j=2, l=i=1,$
        \begin{enumerate}
            \item either $\#_{min}= \langle \gamma 2F, \gamma 1F, a, \{1,2\}:\{0,1\} \rangle$: 

            Lemma 1: $s=[1:C], t=[1:D]$
                \[
                \begin{prooftree}
                    \hypo{[1:C; 2:\boxed{F}]}
                    \infer1[1\#]{[1:\boxed{\texttt{\#}}, C]}
                    \hypo{[1:\boxed{F},D]}
                    \infer1[2\#]{[1:D; 2:\boxed{\texttt{\#}}]}
                    \infer2[\textsc{Cut}]{[1:C, D]}
                \end{prooftree}
\]
can be reduced to
\[
                \begin{prooftree}
                    \hypo{[1:C; 2:F]}
                    \hypo{[1:F, D]}
                    \infer2[\textsc{Cut}]{[1: C, D]}
                \end{prooftree}\, .
                \]
            Lemma 3:
                \[
                \begin{prooftree}
                    \hypo{\emptyset}
                    \infer1[\textsc{Id}]{[1:F; 2:\boxed{F}]}
                    \infer1[1\#]{[1: \boxed{\texttt{\#}}, \boxed{F}]}
                    \infer1[2$\natural$]{[1: \texttt{\#}; 2: \boxed{\na}]}
                \end{prooftree}\, .
                \]
            Special version: $\neg_{min} = \langle \gamma 2A, \gamma 1A, 1, \{1,2\}:\{0,1\} \rangle$,
            \item or $\#_{min}=\langle \gamma 2F, \gamma 1F, a, \{0,1\}:\{1,2\} \rangle$:

            Lemma 1: $s=[2:C]; t=[2:D]$
                \[
                \begin{prooftree}
                    \hypo{[2:\boxed{F}, C]}
                    \infer1[1\#]{[1:\boxed{\texttt{\#}}; 2: C]}
                    \hypo{[1:\boxed{F}; 2:D]}
                    \infer1[2\#]{[2:\boxed{\texttt{\#}}, D]}
                    \infer2[\textsc{Cut}]{[2:C, D]}
                \end{prooftree}
\]
can be reduced to
\[
                \begin{prooftree}
                    \hypo{[2:F, C]}
                    \hypo{[1:F; 2:D]}
                    \infer2[\textsc{Cut}]{[2: C, D]}
                \end{prooftree}\, .
                \]
            Lemma 3:
                \[
                \begin{prooftree}
                    \hypo{\emptyset}
                    \infer1[\textsc{Id}]{[1:\boxed{F}; 2:F]}
                    \infer1[2$\natural$]{[2: \boxed{F}, \boxed{\na}]}
                    \infer1[1\#]{[1: \boxed{\texttt{\#}}; 2: \na]}
                \end{prooftree}\, .
                \]
            Special version: $\sim_{min}=\langle \gamma 2A, \gamma 1A, 1, \{0,1\}:\{1,2\} \rangle$.
        \end{enumerate} 
    \end{enumerate}
\item Let \texttt{SP}($i\#)=\gamma kFlG$, \texttt{SP}($j\#)=\delta cd nFmG$.
    \begin{enumerate}
    \item for $i=1, j=2$ and
        \begin{enumerate}
            \item $k=l=1, m=n=2, \#_{min}=\langle \gamma 1F1G, \delta cd 2F2G, a, \{1,2\}:\{1\} \rangle$:

            Lemma 1: $s=[2:C], t=[1:D], u=[1:E]$:
                \[
                \begin{prooftree}
                    \hypo{[1:\boxed{F}, \boxed{G}; 2:C]}
                    \infer1[1\#]{[1:\boxed{\texttt{\#}}; 2:C]}
                    \hypo{[1:D; 2:\boxed{F}]}
                    \hypo{[1:E;2:\boxed{G}]}
                    \infer2[2\#]{[1: D,E; 2:\boxed{\texttt{\#}}]}
                    \infer2[\textsc{Cut}]{[1:D,E;2:C]}
                \end{prooftree}
\]
can be reduced to
\[
                \begin{prooftree}
                    \hypo{[1:F, G; 2:C]}
                    \hypo{[1:D; 2:F]}
                    \infer2[\textsc{Cut}]{[1:G,D; 2:C]}
                    \hypo{[1:E;2:G]}
                    \infer2[\textsc{Cut}]{[1:D,E;2:C]}
                \end{prooftree}\, .
                \]

            Lemma 3:
                \[
                \begin{prooftree}
                    \hypo{\emptyset}
                    \infer1[\textsc{Id}]{[1:F;2:\boxed{F}]}
                    \hypo{\emptyset}
                    \infer1[\textsc{Id}]{[1:G; 2:\boxed{G}]}
                    \infer2[2$\natural$]{[1:\boxed{F},\boxed{G};2:\boxed{\na}]}
                    \infer1{[1:\boxed{\texttt{\#}};2:\na]}
                \end{prooftree}\, .
                \]
            Special version: $\otimes_{min}=\langle \gamma 1F1G, \delta cd 2F2G, 2, \{1,2\}:\{1\} \rangle$.
        \item $k=l=2, m=n=1,  \#_{min}=\langle \gamma 2F2G, \delta cd 1F1G, a, \{0,1\}:\{1,3\} \rangle$: 

            Lemma 1: $s=[2:C], t=[2:D], u=[2:E]$:
                \[
                \begin{prooftree}
                    \hypo{[2: \boxed{F}, \boxed{G}, C]}
                    \infer1[1\#]{[1:\boxed{\texttt{\#}}; 2:C]}
                    \hypo{[1:\boxed{F}; 2:D]}
                    \hypo{[1:\boxed{G}; 2:E]}
                    \infer2[2\#]{[2:\boxed{\texttt{\#}}, D, E]}
                    \infer2[\textsc{Cut}]{[2:C,D,E]}
                \end{prooftree}
\]
can be reduced to
\[
                \begin{prooftree}
                    \hypo{[2:F, G, C]}
                    \hypo{[1:F; 2: D]}
                    \infer2[\textsc{Cut}]{[2:G, C, D]}
                    \hypo{[1:G; 2:E]}
                    \infer2[\textsc{Cut}]{[2:C, D, E]}
                \end{prooftree}\, .
                \]
            Lemma 3:
                \[
                \begin{prooftree}
                    \hypo{\emptyset}
                    \infer1[\textsc{Id}]{[1:\boxed{F};2:F]}
                    \hypo{\emptyset}
                    \infer1[\textsc{Id}]{[1:\boxed{G};2:G]}
                    \infer2[2$\natural$]{[2: \boxed{F},\boxed{G}, \boxed{\na}]}
                    \infer1[1\#]{[1: \boxed{\texttt{\#}}; 2:\na]}
                \end{prooftree}\, .
                \]
            Special version: $\oplus^{cp}_{min}=\langle \gamma 2A2B, \delta cd 1A1B, 2, \{0,1\}:\{1,3\} \rangle$.
            
    \item $k=n=1, l=m=2, \#_{min}=\langle \gamma 1F2G, \delta cd 2F1G, a, \{1\}:\{1, 2\} \rangle$:

            Lemma 1: $s=[2:C], t=[1:D], u=[2:E]$:
                \[
                \begin{prooftree}
                    \hypo{[1:\boxed{F};2:\boxed{G}, C]}
                    \infer1[1\#]{[1:\boxed{\texttt{\#}};2:C]}
                    \hypo{[1:D; 2:\boxed{F}]}
                    \hypo{[1:\boxed{G}; 2:E]}
                    \infer2[2\#]{[1:D; 2:\boxed{\texttt{\#}}, E]}
                    \infer2[\textsc{Cut}]{[1:D; 2:C,E]}
                \end{prooftree} 
\]
can be reduced to
\[
                \begin{prooftree}
                    \hypo{[1:F;2:G, C]}
                    \hypo{[1:D; 2:F]}
                    \infer2[\textsc{Cut}]{[1:D; 2:G, C]}
                    \hypo{[1:G; 2: E]}
                    \infer2[\textsc{Cut}]{[1:D; 2: C, E]}
                \end{prooftree}\, .
                \]
            Lemma 3:
                \[
                \begin{prooftree}
                    \hypo{\emptyset}
                    \infer1[\textsc{Id}]{[1:F;2:\boxed{F}]}
                    \hypo{\emptyset}
                    \infer1[\textsc{Id}]{[1:\boxed{G};2:G]}
                    \infer2[2$\natural$]{[1:\boxed{F}; 2: \boxed{G}, \boxed{\na}]}
                    \infer1[1\#]{[1: \boxed{\texttt{\#}};2:\na]}
                \end{prooftree}\, .
                \]
            Special version: $\rightarrow^{cp}_{min}=\langle \gamma 1A2B, \delta cd 2A1B, 2, \{1\}:\{1, 2\} \rangle$.
            
    \item $k=n=2, l=m=1, \#_{min}=\langle \gamma 2F1G, \delta cd 1F2G, a, \{1\}:\{1, 2\} \rangle$:

            Lemma 1: $s=[2:C], t=[2:D], u=[1:E]$:
                \[
                \begin{prooftree}
                    \hypo{[1:\boxed{G};2:\boxed{F}, C]}
                    \infer1[1\#]{[1:\boxed{\texttt{\#}}; 2: C]}
                    \hypo{[1:\boxed{F}; 2:D]}
                    \hypo{[1:E; 2:\boxed{G}]}
                    \infer2[2\#]{[1:E; 2:\boxed{\texttt{\#}}, D]}
                    \infer2[\textsc{Cut}]{[1:E; 2:C, D]}
                \end{prooftree}
\]
can be reduced to
\[
                \begin{prooftree}
                    \hypo{[1:G; 2:F, C]}
                    \hypo{[1:F; 2:D]}
                    \infer2[\textsc{Cut}]{[1:G; 2: C,D]}
                    \hypo{[1:E; 2:G]}
                    \infer2[\textsc{Cut}]{[1:E; 2: C,D]}
                \end{prooftree}\, .
                \]
            Lemma 3:
                \[
                \begin{prooftree}
                    \hypo{\emptyset}
                    \infer1[\textsc{Id}]{[1:\boxed{F};2:F]}
                    \hypo{\emptyset}
                    \infer1[\textsc{Id}]{[1:G;2:\boxed{G}]}
                    \infer2[2$\natural$]{[1:\boxed{G};2:\boxed{F},\boxed{\na}]}
                    \infer1[1\#]{[1:\boxed{\texttt{\#}};2:\na]}
                \end{prooftree}\, .
                \]
            Special version: $\leftarrow^{cp}_{min}=\langle \gamma 2A1B, \delta cd 1A2B, 2, \{1\}:\{1, 2\} \rangle$ .
        \end{enumerate}
    \item for $i=2, j=1$ and
    \begin{enumerate}
    \item $k=l=1, m=n=2, \#_{min}= \langle \delta cd 2F2G, \gamma 1F1G, a, \{1,3\}:\{0, 1\} \rangle$: 

        Lemma 1: $s=[1:C], t=[1:D], u=[1:E]$
        \[
            \begin{prooftree}
                \hypo{[1:\boxed{F},\boxed{G},C]}
                \infer1[2\#]{[1: C, 2:\boxed{\texttt{\#}}]}
                \hypo{[1:D; 2:\boxed{F}]}
                \hypo{[1:E; 2:\boxed{G}]}
                \infer2[1\#]{[1:\boxed{\texttt{\#}}, D, E]}
                \infer2[\textsc{Cut}]{[1:C,D,E]}
            \end{prooftree} 
\]
can be reduced to
\[
            \begin{prooftree}
                \hypo{[1:F, G, C]}
                \hypo{[1:D; 2:F]}
                \infer2[\textsc{Cut}]{[1:G, C, D]}
                \hypo{[1:E; 2:G]}
                \infer2[\textsc{Cut}]{[1:C,D,E]}
            \end{prooftree}\, .
            \]
        Lemma 3:
            \[
            \begin{prooftree}
                \hypo{\emptyset}
                \infer1[\textsc{Id}]{[1:F;2:\boxed{F}]}
                \hypo{\emptyset}
                \infer1[\textsc{Id}]{[1:G;2:\boxed{G}]}
                \infer2[1\#]{[1:\boxed{F},\boxed{G},\boxed{\texttt{\#}}]}
                \infer1[2$\natural$]{[1:\texttt{\#};2:\boxed{\na}]}
            \end{prooftree}\, .
            \]
        Special version: $\otimes^{cp}_{min}=\langle \delta cd 2A2B, \gamma 1A1B, 2, \{1,3\}:\{0, 1\} \rangle$ .
    \item $k=l=2, m=n=1, \#_{min}= \langle \delta cd 1F1G, \gamma 2F2G, a, \{1\}:\{1,2\} \rangle$:

            Lemma 1: $s=[1:C], t=[2:D], u=[2:E]$:
            \[
            \begin{prooftree}
                \hypo{[1:C; 2:\boxed{F}, \boxed{G}]}
                \infer1[2\#]{[1:C; 2:\boxed{\texttt{\#}}]}
                \hypo{[1:\boxed{F}; 2:D]}
                \hypo{[1:\boxed{G}; 2:E]}
                \infer2[1\#]{[1:\boxed{\texttt{\#}}; 2:D, E]}
                \infer2[\textsc{Cut}]{[1:C; 2:D,E]}
            \end{prooftree} 
\]
can be reduced to
\[
            \begin{prooftree}
                \hypo{[1:C; 2:F,G]}
                \hypo{[1:F; 2:D]}
                \infer2[\textsc{Cut}]{[1:C; 2:G, D]}
                \hypo{[1:G; 2:E]}
                \infer2[\textsc{Cut}]{[1:C; 2:D,E]}
            \end{prooftree}\, .
            \]
        Lemma 3:
            \[
            \begin{prooftree}
                \hypo{\emptyset}
                \infer1[\textsc{Id}]{[1:\boxed{F};2:F]}
                \hypo{\emptyset}
                \infer1[\textsc{Id}]{[1:\boxed{G};2:G]}
                \infer2[1\#]{[1:\boxed{\texttt{\#}}; 2: \boxed{F}, \boxed{G}]}
                \infer1[2$\natural$]{[1:\texttt{\#};2:\boxed{\na}]}
            \end{prooftree}\, .
            \]
        Special version: $\oplus_{min}=\langle \delta cd 1A1B, \gamma 2A2B, 2, \{1\}:\{1,2\} \rangle$.
    \item $k=n=1, l=m=2, \#_{min}= \langle \delta cd 2F1G, \gamma 1F2G, a, \{1,2\}:\{1\} \rangle$:

         Lemma 1: $s=[1:C], t=[1:D], u=[2:E]$
            \[
            \begin{prooftree}
                \hypo{[1:\boxed{F}, C;2:\boxed{G}]}
                \infer1[2\#]{[1:C; 2:\boxed{\texttt{\#}}]}
                \hypo{[1:D; 2:\boxed{F}]}
                \hypo{[1:\boxed{G}; 2: E]}
                \infer2[1\#]{[1:\boxed{\texttt{\#}}, D; 2: E]}
                \infer2[\textsc{Cut}]{[1: C,D; 2:E]}
            \end{prooftree} 
            \]
can be reduced to
\[
            \begin{prooftree}
                \hypo{[1:F,C;2:G]}
                \hypo{[1: D; 2:F]}
                \infer2[\textsc{Cut}]{[1:C, D; 2:G]}
                \hypo{[1:G; 2:E]}
                \infer2[\textsc{Cut}]{[1:C,D; 2:E]}
            \end{prooftree}\, .
            \]
        Lemma 3:
            \[
            \begin{prooftree}
                \hypo{\emptyset}
                \infer1[\textsc{Id}]{[1:F;2:\boxed{F}]}
                \hypo{\emptyset}
                \infer1[\textsc{Id}]{[1:\boxed{G};2:G]}
                \infer2[1\#]{[\boxed{1:\texttt{\#}},\boxed{F};2:\boxed{G}]}
                \infer1[2$\natural$]{[1:\texttt{\#};2:\boxed{\na}]}
            \end{prooftree}\, .
            \]
        Special version: $\rightarrow_{min}=\langle \delta cd 2A1B, \gamma 1A2B, a, \{1,2\}:\{1\} \rangle $.
    \item $k=n=2, l=m=1, \#_{min}= \langle \delta cd 1F2G, \gamma 2F1G, 2, \{1,2\}:\{1\} \rangle$:

        Lemma 1: $s=[1:C], t=[2:D], u=[1:E]$
            \[
            \begin{prooftree}
                \hypo{[1:\boxed{G}, C; 2:\boxed{F}]}
                \infer1[2\#]{[1:C; 2:\boxed{\texttt{\#}}]}
                \hypo{[1:\boxed{F};2:D]}
                \hypo{[1:E;2:\boxed{G}]}
                \infer2[1\#]{[1:\boxed{\texttt{\#}}, E; 2:D]}
                \infer2[\textsc{Cut}]{[1: C,E; 2:D]}
            \end{prooftree} \]
can be reduced to
\[
            \begin{prooftree}
                \hypo{[1:G,C;2:F]}
                \hypo{[1:F;2:D]}
                \infer2[\textsc{Cut}]{[1:G, C; 2:D]}
                \hypo{[1:E;2:G]}
                \infer2[\textsc{Cut}]{[1:C,E;2:D]}
            \end{prooftree}\, .
            \]
        Lemma 3:
            \[
            \begin{prooftree}
                \hypo{\emptyset}
                \infer1[\textsc{Id}]{[1:\boxed{F};2:F]}
                \hypo{\emptyset}
                \infer1[\textsc{Id}]{[1:G;2:\boxed{G}]}
                \infer2[1\#]{[1:\boxed{\texttt{\#}},\boxed{G}; 2:\boxed{F}]}
                \infer1[2$\natural$]{[1: \texttt{\#};2:\boxed{\na}]}\, .
            \end{prooftree}
            \]
        Special version: $\leftarrow_{min}=\langle \delta cd 1A2B, \gamma 2A1B, 2, \{1,2\}:\{1\} \rangle$.
        \end{enumerate}
    \end{enumerate}
\item Let \texttt{SP}($i\#)=\delta cs kFlG$, \texttt{SP}($j\#)=\epsilon nFmG$.
\begin{enumerate}
    \item for $i=1, j=2$ and
    \begin{enumerate}
        \item $k=l=1, m=n=2, \#_{min}=\langle \delta cs 1F1G, \epsilon 2F2G, a, \{1\}:\{1\} \rangle:$ 

            Lemma 1: $s=[2:C], t=[1:D]$
                \[
                \begin{prooftree}
                    \hypo{[1:\boxed{F}; 2:C]}
                    \hypo{[1:\boxed{G}; 2:C]}
                    \infer2[1\#]{[1:\boxed{\texttt{\#}}; 2:C]}
                    \hypo{[1:D; 2:\boxed{F}]}
                    \infer1[2\#]{[1:D; 2:\boxed{\texttt{\#}}]}
                    \infer2[\textsc{Cut}]{[1:D; 2:C]}
                \end{prooftree} \]
can be reduced to
\[
                \begin{prooftree}
                    \hypo{[1:F;2:C]}
                    \hypo{[1:D;2:F]}
                    \infer2[\textsc{Cut}]{[1:D; 2:C]} 
                \end{prooftree}\, .
                \]
            Lemma 3:
                \[
                \begin{prooftree}
                    \hypo{\emptyset}
                    \infer1[\textsc{Id}]{[1:F;2:\boxed{F}]}
                    \infer1[2$\natural$]{[1:\boxed{F}; 2:\boxed{\na}]}
                    \hypo{\emptyset}
                    \infer1[\textsc{Id}]{[1:G;2:\boxed{G}]}
                    \infer1[2$\natural$]{[1:\boxed{G}; 2:\boxed{\na}]}
                    \infer2[1\#]{[1:\boxed{\texttt{\#}};2:\na]}
                \end{prooftree}\, .
                \]
            Special version: $\sqcup_{min}=\langle \delta cs 1A1B, \epsilon 2A2B, 2, \{1\}:\{1\} \rangle$. 
        \item $k=l=2, m=n=1, \#_{min}=\langle \delta cs 2F2G, \epsilon 1F1G, a, \{0,1\}:\{1,2\} \rangle$:

            Lemma 1: $s=[1:C], t=[2:D]$
                \[
                \begin{prooftree}
                    \hypo{[2:\boxed{F}, C]}
                    \hypo{[2:\boxed{G}, C]}
                    \infer2[1\#]{[1:\boxed{\texttt{\#}}; 2:C]}
                    \hypo{[1:\boxed{F}; 2:D]}
                    \infer1[2\#]{[2:\boxed{\texttt{\#}}, D]}
                    \infer2[\textsc{Cut}]{[2:C,D]}
                \end{prooftree} \]
can be reduced to
\[
                \begin{prooftree}
                    \hypo{[2:F, C]}
                    \hypo{[1:F; 2:D]}
                    \infer2[\textsc{Cut}]{[2:C,D]} 
                \end{prooftree}\, .
                \]
            Lemma 3:
                \[
                \begin{prooftree}
                    \hypo{\emptyset}
                    \infer1[\textsc{Id}]{[1:\boxed{F};2:F]}
                    \infer1[2$\natural$]{[2:\boxed{F}, \boxed{\na}]}
                    \hypo{\emptyset}
                    \infer1[\textsc{Id}]{[1:\boxed{G};2:G]}
                    \infer1{[2:\boxed{G}, \boxed{\na}]}
                    \infer2[1\#]{[1:\boxed{\texttt{\#}};2:\na]}
                \end{prooftree}\, .
                \]
            Special version: $\sqcap^{cp}_{min}=\langle \delta cs 2A2B, \epsilon 1A1B, 2, \{0,1\}:\{1,2\} \rangle$.
        \item $k=n=1, l=m=2, \#_{min}= \langle \delta cs 1F2G, \epsilon 2F1G, a, \{0,1\}:\{1,2\} \rangle$:

            Lemma 1: $s=[2:C], t=[1:D; 2:E]$
                \[
                \begin{prooftree}
                    \hypo{[1:\boxed{F}; 2:C]}
                    \hypo{[2:\boxed{G}, C]}
                    \infer2[1\#]{[1:\boxed{\texttt{\#}}; 2:C]}
                    \hypo{[1:D,2:\boxed{F}, E]}
                    \infer1[2\#]{[1:D; 2:\boxed{\texttt{\#}}, E]}
                    \infer2{[1:D; 2:C, E]}
                \end{prooftree} \]
can be reduced to
\[
                \begin{prooftree}
                    \hypo{[1:F; 2: C]}
                    \hypo{[1:D; 2:F, E]}
                    \infer2[\textsc{Cut}]{[1:D; 2:C, E]} 
                \end{prooftree}
                \]\, .
            Lemma 3:
                \[
                \begin{prooftree}
                    \hypo{\emptyset}
                    \infer1[\textsc{Id}]{[1:F;2:\boxed{F}]}
                    \infer1[2$\natural$]{[1:\boxed{F}; 2:\boxed{\na}]}
                    \hypo{\emptyset}
                    \infer1[\textsc{Id}]{[1:\boxed{G};2:G]}
                    \infer1{[2:\boxed{G}, \boxed{\na}]}
                    \infer2{[1:\boxed{\texttt{\#}};2:\na]}
                \end{prooftree}\, .
                \]
            Special version: $\leftsquigarrow_{min}=\langle \delta cs 1A2B, \epsilon 2A1B, 2, \{0,1\}:\{1,2\} \rangle$.
        \item $k=n=2, l=m=1, \#_{min}= \langle \delta cs 2F1G, \epsilon 1F2G, a, \{0.1\}:\{1,2\} \rangle$:

            Lemma 1: $s=[2:C], t=[2:D]$
                \[
                \begin{prooftree}
                    \hypo{[2:\boxed{F}, C]}
                    \hypo{[1:\boxed{G}; 2:C]}
                    \infer2[1\#]{[1:\boxed{\texttt{\#}}; 2:C]}
                    \hypo{[1:\boxed{F}; 2:D]}
                    \infer1[2\#]{[2:\boxed{\texttt{\#}}, D]}
                    \infer2[\textsc{Cut}]{[2:C, D]}
                \end{prooftree} \]
can be reduced to
\[
                \begin{prooftree}
                    \hypo{[2:F, C]}
                    \hypo{[1:F; 2:D]}
                    \infer2[\textsc{Cut}]{[2:C, D]} 
                \end{prooftree}\, .
                \]
            Lemma 3:
                \[
                \begin{prooftree}
                    \hypo{\emptyset}
                    \infer1[\textsc{Id}]{[1:\boxed{F};2:F]}
                    \infer1[2$\natural$]{[2:\boxed{F},\boxed{\na}]}
                    \hypo{\emptyset}
                    \infer1[\textsc{Id}]{[1:G;2:\boxed{G}]}
                    \infer1[2$\natural$]{[1:\boxed{G}; 2: \boxed{\na}]}
                    \infer2[1\#]{[1:\boxed{\texttt{\#}};2:\na]}
                \end{prooftree}\, .
                \]
            Special version: $\rightsquigarrow^{cp}_{min}= \langle \delta cs 2A1B, \epsilon 1A2B, 2, \{0,1\}:\{1,2\} \rangle$.
        \end{enumerate}
    \item for $i=2, j=1$ and
    \begin{enumerate}
        \item $k=l=1, m=n=2, \#_{min}= \langle \epsilon 1F1G, \delta cs 2F2G, a, \{1,2\}:\{0,1\} \rangle$:

            Lemmata 1 and 3: cf. 5(1)ii, modulo swapping dimensions.\\
            Special version: $\sqcup^{cp}_{min}=\langle \epsilon 1A1B, \delta cs 2A2B,  2, \{1,2\}:\{0,1\} \rangle$ .
        \item $k=l=2, m=n=1, \#_{min}= \langle \epsilon 2F2G, \delta cs 1F1G, a,  \{1\}:\{1\} \rangle$:

            Lemmata 1 and 3: cf. 5(1)i, modulo swapping dimensions.\\
            Special version: $\sqcap_{min}=\langle \epsilon 2A2B,\delta cs 1A1B, 2, \{1\}:\{1\} \rangle$.
        \item $k=n=1, l=m=2, \#_{min}= \langle \epsilon 2F1G, \delta cs 1F2G, a, \{1,2\}:\{0,1\} \rangle$:

            Lemmata 1 and 3: cf. 5(1)iv, modulo swapping dimensions.\\
            Special version: $\leftsquigarrow^{cp}_{min}=\langle \epsilon 2A1B, \delta cs 1A2B,  2, \{1,2\}:\{0,1\} \rangle$.
        \item $k=n=2, l=m=1, \#_{min}= \langle \epsilon 1F2G, \delta cs 2F1G, a, \{1,2\}:\{0,1\} \rangle$:

            Lemmata 1 and 3: cf. 5(1)iii, modulo swapping dimensions.\\
            Special version: $\rightsquigarrow_{min}= \langle \epsilon 1A2B, \delta cs 2A1B, 2, \{1,2\}:\{0,1\} \rangle$.
    \end{enumerate}
\end{enumerate}
\end{enumerate}

\subsection{Equivalent connectives}

In \S5, we found that $150$ formulable operational rule pairs ($\leq 2$ premiss sequents, $\leq 2$ active formulae) are also minimally definable. However, we have so far considered only a subset of these. How should we handle the remaining $\#\in\texttt{PU}$? We find that these rule pairs---let us call them \textit{bloatnectives}---are \textit{equivalent} to those already discussed (even relative to our weak minimal derivability relation) while yielding \textit{non-minimal} inference behaviour. Therefore, bloatnectives do not have distinct semantic clauses.
\begin{corollary}\

    \noindent
    Let $\textbf{Cal}=\{\textsc{Id}, \textsc{Cut} cd, i\#, j\#, i\natural, j\natural\}$.

     \# is \textbf{Cal}-equivalent to $\natural$, for any connectives $\#, \natural$ with equal frameworks and arities such that:
    \begin{enumerate}
        \item \texttt{SP}($i\#)=\gamma kF$, \texttt{SP}($j\#)=\delta cs lF$ and \texttt{SP}($i\natural)=\gamma kF$, \texttt{SP}($j\natural)=\gamma lF$, or
        \item
        \begin{enumerate}
            \item \texttt{SP}($i\#$), \texttt{SP}($i\natural)\in \{ \alpha, \beta kF, \beta kFlG, \gamma 1F2F)\}$ and
            \item \texttt{SP}($j\#$), \texttt{SP}($j\natural)\in \{ \gamma 0, \delta cs 0, \delta cs 1F2F \}$.
        \end{enumerate}
    \end{enumerate}
\end{corollary}
\begin{proof} The cases covered by disjunct 1. are immediate from Corollary 1. Those covered by 2. also follow from Corollary 1 and the facts that (a) $i\#$ and  $i\natural$ are equivalent, and (b) $j\#$ and  $j\natural$ are equivalent.

(a) is a consequence of $i\#/i\natural$ always acting axiomatically, regardless of \texttt{SP}: $\emptyset$ is always provable, as is $\emptyset+\mathfrak{s}$, regardless of $\mathfrak{s}$, by virtue of the definition of sequent disjunction $+$. The same holds for $[1:F;2:F]$ for if $F$ is primitive, $[1:F;2:F]$ follows from \textsc{Id}$^{pr}$, and otherwise from the derivability of full \textsc{Id} (Theorem 1).

Regarding (b), $s$ is provable iff $s\ \_ \ s$ is. This follows from the definition of sequent conjunction $\_$ and the fact that our vertical consequence relation (on sequents) is \textit{monotonic} and \textit{contractive} as per Definition 11.\footnote{In fact, our vertical consequence relation is fully Tarskian \cite[cf. also][]{blok2006equivalence, dicher2020variations}.} Moreover, since $[1:F;2:F]$ is \textit{always} provable, $[1:F;2:F]\ \_ \ s$ is provable iff $s$ is. \hfill $\blacksquare$
\end{proof}
Corollary 3 covers exactly the remaining $\#\in\texttt{PU}$. We can think of it as follows: our \textit{bloatnectives} can be grouped into two equivalence classes, defined by disjuncts 1. and 2., respectively. The former maps to our minimal \textit{unary} connectives, i.e. $\neg_{min}/\sim_{min}/\neg^{cv}_{min}$ ($\texttt{SP}(i\#)=\gamma kF, \texttt{SP}(j\#)=\gamma lF$). The latter corresponds to our minimal \textit{nullary} connectives, i.e. $\bot_{min}/\top_{min}$ ($\texttt{SP}(i\#)=\alpha, \texttt{SP}(j\#)=\gamma 0$). Bloatnectives have the same minimal framework as the corresponding minimal connective. The main difference between our minimally meaningful connectives and our bloatnectives is syntactic redundancies of no logical consequence such as redundant or inert branches.\footnote{If we had also allowed for operational rules with more than two premises when setting up this study, we would have also found similar bloatnectives for the remaining minimal binary connectives.}

How should we deal with bloatnectives? There are mainly two options. On one hand, we could consider them minimally meaningful connectives. After all, each bloatnective is minimally definable, and its minimal inference behaviour is different from that of the corresponding minimal connective. Hence, one could argue that they should have their own meaning. For example, comparing the horizontal interderivability proofs for $\langle \alpha, \gamma 0\rangle$ and $\langle \gamma1F2F, \delta cs 0 \rangle$, we can see the difference in $\texttt{ib}$ in the following proofs, as shown by the boxed formulae:
\[
\begin{prooftree}
    \hypo{\emptyset^0}
    \infer1[1\#]{[1:\boxed{\texttt{\#}}]^1}
    \infer1[2$\natural$]{[1:\boxed{\texttt{\#}}; 2:\boxed{\na}]^2}
\end{prooftree} \] 
versus
\[
\begin{prooftree}
    \hypo{\emptyset^0}
    \infer1[\textsc{Id}]{[1:\boxed{F};2:\boxed{F}]^1}
    \infer1[1\#]{[1:\boxed{\texttt{\#}}]^3}
    \hypo{\emptyset^0}
    \infer1[\textsc{Id}]{[1:\boxed{F};2:\boxed{F}]^2}
    \infer1[1\#]{[1:\boxed{\texttt{\#}}]^4}
    \infer2[2$\natural$]{[1:\boxed{\texttt{\#}}; 2:\boxed{\na}]^5}
\end{prooftree}
\]
Moreover, Corollary 3 would no longer hold if we were to adopt a suitably different minimal derivability relation (e.g. one where full \textsc{Id} is unavailable) or, alternatively, a suitably different vertical consequence relation (by altering the proof tree structure). One could contend that, since I-bS makes such fine-grained distinctions, meaning should be individuated at the level of the operational rule differences that bloatnectives make, even if these remain logically irrelevant.

On the other hand, we could exclude bloatnectives from the results of this investigation. Bloatnectives and their corresponding minimal connective have the same minimal inference behaviour \textit{profile}. For example, in both of the above proofs $\texttt{ibp}^\#=\{[1], [1,2]\}$. Hence, the minimal operational rules for both bloatnectives have the same minimal framework (here: $\{1\}:\{0,1\}$). The only difference between them and their corresponding minimal connective is their premiss type. Distinguishing the two, thus,  exclusively rests on our working assumption that all premiss types are of equal size (cf. \S 6.1).

However, this assumption appears to be misplaced in the present context of equivalence: bloatnectives and their corresponding minimal connective play the same inferential role in our minimal derivability relation. Their inferential \textit{use} is interchangeable. If one premiss type, qua equivalence, does not add any new use over another but adds syntactic clutter, we find it plausible to classify the former as `less minimal' than the latter. For instance in our present example, the bloatnective adds an additional `idle' branch to the uniqueness proof of our minimally meaningful connective. The premiss type of the former thus yields more-than-minimal inference behaviour.

We can capture this idea as follows:
\begin{definition}[premiss type minimality, via $\texttt{ib}$-cardinality]\

    \noindent Let $\mathfrak{C}$ be a set of \textbf{Cal}-equivalent operational rule pairs of equal arity and framework size.
    
    We call the premiss type of some $\#\in \mathfrak{C}$ \textit{minimal} iff its minimal inference behaviour has the lowest cardinality, among all $\natural\in \mathfrak{C}$, i.e. $|\texttt{ib}^\#_\mathfrak{p}|=min(\{|\texttt{ib}^{\natural}_\mathfrak{p}|\mid\natural\in \mathfrak{C}\})$.
\end{definition}
Hence, we link premiss type-minimality to the length of the inference behaviour list, i.e. the minimality of the minimal inference behaviour itself. Although this `official' definition thus matches our theoretical setting, it is somewhat tedious to apply. 

Fortunately, as we have seen in the proof of Corollary 3, bloatnective rules with \texttt{SP}($i\#$) are \textbf{Cal}-equivalent to rules with \texttt{SP}($i\natural$) for $i \in \{1,2\}$. Hence, an equivalent definition at the level of a single operational rule suffices for our present purposes:\footnote{One could extend this notion of minimality further to account for differences in active formula numbers; e.g. if two rules have the same premiss cardinality, the one with fewer active formulae is more minimal. However, our present data does not require such additions.}
\begin{definition}[premiss type minimality, via premiss cardinality]\

    \noindent Let $\mathfrak{R}$ be a set of \textbf{Cal}-equivalent operational rules of equal arity and framework size.
    
    We call the premiss type of rule \begin{prooftree} \hypo{\mathfrak{S}} \infer1{\mathfrak{s}} \end{prooftree} \textit{minimal} iff its multiset of premiss sequents $\mathfrak{S}$ has the \textit{lowest cardinality} of all rules in $\mathfrak{R}$, i.e. $|\mathfrak{S}|=min(\{|\mathfrak{S}^\prime|\mid \begin{prooftree} \hypo{\mathfrak{S}^\prime} \infer1{\mathfrak{s}^\prime}\end{prooftree} \in \mathfrak{R}\})$).
\end{definition}
This definition is straightforward to apply. Instead of examining all definability proofs manually to determine minimal inference behaviour \textit{a posteriori}, the simplicity of our minimal derivability relation allows premiss type minimality to be read off the rules \textit{a priori}. If both operational rules have minimal premiss types (which is true for both equivalence classes in our dataset), the premiss type of their corresponding connective is also minimal.

Yet another---and perhaps more intuitive---way to understand premiss type minimality follows if we conceive of the premiss part of a rule not as a multiset of sequents but as a single sequent. This yields:
\begin{definition}[premiss type minimality, via premiss sequent complexity]\

    \noindent Let $\mathfrak{R}$ be a set of \textbf{Cal}-equivalent operational rules of equal arity and framework size.
    
    We call the premiss type of rule \begin{prooftree} \hypo{\mathfrak{s_1}} \infer1{\mathfrak{s_2}} \end{prooftree} \textit{minimal} iff $(\mathfrak{s}_1)$ has the lowest sequent complexity of all rules in $\mathfrak{R}$, i.e. $c(\mathfrak{s})=min(\{ c(\mathfrak{s}_1^\prime)\mid \, \begin{prooftree} \hypo{\mathfrak{s}_1^\prime} \infer1{\mathfrak{s}_2^\prime}\end{prooftree} \in \mathfrak{R}\})$.
\end{definition}
The notion of complexity $c$ we apply here corresponds to the standard complexity measure for formulae albeit applied at the level sequents, i.e. to objects of the vertical rather than the horizontal derivability relation. In addition to `primitive' sequents of form $\mathfrak{s}$ and our binary sequent operators, $+$ and $\_$, we must, therefore, consider the possibility of the \textit{absence} of sequents.
\begin{definition}[sequent complexity]\

    \noindent The complexity of a sequent is a function $c: \texttt{S}\mapsto n\in \mathbb{N}$, such that
    \begin{itemize}
        \item $c(\alpha)=0$ (cf. `$\emptyset$'),
        \item $c(\gamma)=1$ (cf. `$\mathfrak{s}$'),
        \item $c(\beta)=1+c(\alpha)+c(\gamma)=2$ (cf. `$\emptyset + \mathfrak{s}$'),
        \item $c(\delta cs)=c(\delta cd)=c(\epsilon)=1+c(\gamma)+c(\gamma)=2$ (cf. `$\mathfrak{s}\ \_ \ \mathfrak{s}$'/`$\mathfrak{s} + \mathfrak{s}$').
    \end{itemize}
\end{definition}
Whichever perspective on minimal premiss types one prefers, each definition picks out our minimal connective rather than any bloatnective for both equivalence classes in Corollary 3. For equivalence class 1., $c([lF]s\_ s)>c([lF]s)$. Hence, $\texttt{SP}(j\natural)$ is minimal in $\{\texttt{SP}(j\#), \texttt{SP}(i\natural)\}$. \textit{A fortiori}, as $\texttt{SP}(i\#)=\texttt{SP}(i\natural)$, \texttt{SP}($\natural$) is minimal overall. Similarly for equivalence class 2., $\texttt{SP}(\alpha)$ is minimal in 
$\{ \alpha, \beta kF, \beta kFlG, \gamma 1F2F)\}$ and $\texttt{SP}(\gamma 0)$ is minimal in $\{ \gamma 0, \delta cs 0, \delta cs 1F2F \}$. Hence, the connectives $\#$ with premiss type $\texttt{SP}(i\#)=\alpha, \texttt{SP}(j\#)= \gamma 0$ are minimal and have a semantic clause.

We sympathise with both bloatnective policy options. Option 1, \textit{bloatnective inclusion}, is theoretically simpler as it does not require us to add a notion of premiss type minimality while nicely fitting our modular substructural picture. However, option 2, \textit{bloatnective exclusion}, appears to yield deductively more plausible results while the theoretical additions it requires fit nicely with the overall theory. For definiteness, we adopt the latter.

Having analysed our minimal definability proofs by measuring inference behaviour, we find exactly $21$ connectives with a semantic clause. Let us briefly summarise our results before we discuss some main takeaways.

\section{Results: the minimally meaningful connectives in two-dimensional sequent calculi}

In this section, we summarise our results. The following tables give the semantic clauses for our $21$ minimally meaningful connectives. Only these connectives receive distinct semantic clauses in Nagler's \cite{nagler2026measuring} I-bS, among all $10,816$ operational rule pairs one can formulate in two-dimensional sequent calculi ($\leq 2$ premiss sequents, $\leq 2$ active formulae). These results are relative to our underlying notion of definability, i.e. Belnap-style conservativity and uniqueness, as well as the minimal derivability relation of $\{ \textsc{Id}, \textsc{Cut}^{mult}\}$, and the criteria for rule minimality discussed in \S 6.

\subsection{Explanation of notation}
The component-wise table notation is loosely inspired by Smullyan \cite{smullyan1995first}. `$\#$' denotes the connective under study. \textit{Rows} indicate the position of active formulae in the premiss sequent(s) of the \textit{left} $\#$-rule $L\#$; \textit{columns} give the position of active formulae in the premiss sequent(s) of the \textit{right} $\#$-rule $R\#$. We give a legend for the former in the leftmost column and for the latter in the topmost row. For readability, we indicate active formula positions as uncontextualised sequents using standard notation.

The remaining cells show the \textit{minimal frameworks} for each $\langle L\#,R\# \rangle$-pair, i.e. each semantic clause. For connectives with two premiss sequents, we indicate whether these are arranged in a context-sharing ($cs$) or -differentiating ($cd$) manner. We note in brackets if the subformula-connectedness-adjusted special version of the connective corresponds to commonly used connectives, using the following notation:\footnote{The connective names and symbols for our binary connectives are borrowed from \cite{paoli2007implicational}.} `$\bot$' denotes bottom, `$\top$' denotes top, `$\neg$' denotes one kind of negation, `$\sim$' denotes another kind of negation with the same premiss type as $\neg$ but a different minimal framework,\footnote{We say more about the relation between the two negations in \S 8.} `$\sqcap$' (read: `meet') denotes lattice/additive conjunction, `$\otimes$' (`fusion') denotes group/multiplica\-tive conjunction, `$\sqcup$' (`join') denotes lattice/additive disjunction, `$\oplus$' (`fission') denotes group/multiplicative disjunction, `$\rightsquigarrow$' (`right squiggle') denotes lattice/additive (right-)implication, `$\rightarrow$' (`right arrow') denotes group/multiplicative (right-)implication, `$\leftsquigarrow$' (`left squiggle') denotes lattice converse/left-implication, `$\leftarrow$' (`left arrow') denotes group converse/left-implication. For any connective $\natural$, `$\natural^{cv}$' denotes its \textit{converse} (the active formulae are on the other side of $\vdash$), `$\natural^{cp}$' denotes its \textit{contrapositive} (the premiss sequents of the $L\natural / R\natural$ rules are swapped), and `$\natural^{iv}$' denotes its \textit{inverse} (the converse of its contrapositive). If there are two ways to name the same connective using this scheme, we give both names.

Let us give two examples showing how to read the following tables: 
\begin{example}\

    \noindent In Table 15, the following semantic clause gives the meaning of the converse/contrapositive of $\neg$: \[
    \Biggl\langle \begin{prooftree} \hypo{A\vdash C} \infer1{\# A \vdash C} \end{prooftree}\, , \quad \begin{prooftree}\hypo{C\vdash A} \infer1{C \vdash \# A} \end{prooftree} \Biggl\rangle \, .
    \] 
\end{example}
\begin{example}\

    \noindent In Table 16, the following semantic clause defines the meaning of $\otimes$, i.e. group conjunction: \[\Biggl\langle \begin{prooftree} \hypo{A, B \vdash C} \infer1{ A \# B \vdash C} \end{prooftree}\, , \quad \begin{prooftree}\hypo{C\vdash A} \hypo{D \vdash \# A} \infer2{C, D \vdash A \# B} \end{prooftree} \Biggl\rangle\, .\]
\end{example}
\subsection{Result tables}
\begin{threeparttable}
\begin{longtable}{|c|c|c|}
\caption{minimal connective rules, no active formulae}
\endfirsthead
\endhead
\hline
 &&\\
L$\#$ $\backslash$ R$\#$ & $\emptyset$ & $\vdash$  \\
&&\\ \hline
&\cellcolor{gray!25}&\\
$\emptyset$ & \cellcolor{gray!25} & $ \{ 1\} : \{0,  1 \}$ \\
& \cellcolor{gray!25} & $(\bot=\top^{cv/cp}=\bot^{iv})$  \\ \hline
 &  &\cellcolor{gray!25} \\
 $\vdash$ & $\{ 0, 1 \} : \{ 1\}$ & \cellcolor{gray!25}\\
 & $(\top = \bot^{cv/cp}=\top^{iv})$ & \cellcolor{gray!25} \\ \hline
\end{longtable}
\end{threeparttable}
\begin{threeparttable}
\begin{longtable}{|c|c|c|}
\caption{minimal connective rules, one active formula}
\endfirsthead
\endhead
\hline&&\\
L$\#$ $\backslash$ R$\#$ & $F \vdash $ & $\vdash F$ \\ &&\\ \hline
&\cellcolor{gray!25}&\\
$F \vdash $ & \cellcolor{gray!25} & $ \{ 1\} : \{  1 \}$\\ & \cellcolor{gray!25} & $(\neg^{cv}/\sim^{cv}=\neg^{cp}/\sim^{cp})$  \\ \hline
 & $\{ 0, 1 \} : \{ 1, 2\}$\tnote{\dag} &\cellcolor{gray!25} \\ 
 & $(\sim\ =\ \sim^{iv})$  & \cellcolor{gray!25} \\ 
 $\vdash F$ & or & \cellcolor{gray!25} \\
 & $ \{ 1, 2\} : \{ 0, 1 \}$\tnote{\dag} & \cellcolor{gray!25} \\ & $(\neg = \neg^{iv})$ & \cellcolor{gray!25} \\ \hline
\end{longtable}
\begin{tablenotes}
\centering
\item $F\in \{A,B\}$
\item[\dag] First proven in \cite{dicher2016proof}.
\end{tablenotes}
\end{threeparttable}
\begin{sidewaystable}[htbp]
\caption{minimal connective rules, two active formulae}
\begin{threeparttable}
{\tiny
\begin{longtable}{|c|c|c|c|c|c|c|c|c|c|c|c|c|}
\endfirsthead
\endhead
\hline  &&&&&&&&&&&& \\ L$\#$ $\backslash$ R$\#$ & $F, G \vdash$ & $F \vdash G$ & $F \vdash G$ & $ \vdash F, G $ & $F \vdash$ $G \vdash$ & $F \vdash$ $ \vdash G$ & $ \vdash F$ $G \vdash$ & $ \vdash F$ $ \vdash G$ & $F \vdash$ or $G \vdash$ & $F \vdash$ or $ \vdash G$ & $ \vdash F$ or $G \vdash$ & $ \vdash F$ or $ \vdash G$ \\   &&&&&&&&&&&& \\ \hline
& \cellcolor{gray!25} \cellcolor{gray!25}  & \cellcolor{gray!25}  &  \cellcolor{gray!25}  & \cellcolor{gray!25}  & \cellcolor{gray!25}  & \cellcolor{gray!25}  & \cellcolor{gray!25}  & $cd$ & \cellcolor{gray!25}  & \cellcolor{gray!25}  & \cellcolor{gray!25}  & \cellcolor{gray!25} \\
$F, G \vdash$ & \cellcolor{gray!25}  & \cellcolor{gray!25}  & \cellcolor{gray!25}  & \cellcolor{gray!25}  &  \cellcolor{gray!25}  & \cellcolor{gray!25}  & \cellcolor{gray!25}  &  $ \{ 1, 2\} : \{ 1 \}$ & \cellcolor{gray!25} & \cellcolor{gray!25}  & \cellcolor{gray!25}  & \cellcolor{gray!25} \\
& \cellcolor{gray!25} \cellcolor{gray!25}  & \cellcolor{gray!25}  &  \cellcolor{gray!25}  & \cellcolor{gray!25}  & \cellcolor{gray!25}  & \cellcolor{gray!25}  & \cellcolor{gray!25}  & $(\otimes = \oplus^{iv})$ & \cellcolor{gray!25}  & \cellcolor{gray!25}  & \cellcolor{gray!25}  & \cellcolor{gray!25} \\  \hline
&  \cellcolor{gray!25} &  \cellcolor{gray!25} &  \cellcolor{gray!25} &  \cellcolor{gray!25} &  \cellcolor{gray!25} &  \cellcolor{gray!25} & $cd$ & \cellcolor{gray!25} &\cellcolor{gray!25}  & \cellcolor{gray!25} & \cellcolor{gray!25} & \cellcolor{gray!25} \\
$F \vdash G$ & \cellcolor{gray!25} & \cellcolor{gray!25} & \cellcolor{gray!25} & \cellcolor{gray!25} & \cellcolor{gray!25} & \cellcolor{gray!25} & $\{ 1 \} : \{ 1, 2\}$  & \cellcolor{gray!25} & \cellcolor{gray!25} & \cellcolor{gray!25} & \cellcolor{gray!25} & \cellcolor{gray!25} \\  
&  \cellcolor{gray!25} &  \cellcolor{gray!25} &  \cellcolor{gray!25} &  \cellcolor{gray!25} &  \cellcolor{gray!25} &  \cellcolor{gray!25} & ($\rightarrow^{cp}= \leftarrow^{iv}$) & \cellcolor{gray!25} &\cellcolor{gray!25}  & \cellcolor{gray!25} & \cellcolor{gray!25} & \cellcolor{gray!25} \\  \hline
 & \cellcolor{gray!25} & \cellcolor{gray!25} & \cellcolor{gray!25} & \cellcolor{gray!25} & \cellcolor{gray!25} & $cd$ & \cellcolor{gray!25} & \cellcolor{gray!25} & \cellcolor{gray!25} & \cellcolor{gray!25} & \cellcolor{gray!25} & \cellcolor{gray!25} \\ 
$F \vdash G$ & \cellcolor{gray!25} & \cellcolor{gray!25} & \cellcolor{gray!25} & \cellcolor{gray!25} & \cellcolor{gray!25} &  $\{ 1 \} : \{ 1, 2\}$ & \cellcolor{gray!25} &  \cellcolor{gray!25} & \cellcolor{gray!25} & \cellcolor{gray!25} & \cellcolor{gray!25} & \cellcolor{gray!25}   \\ 
 & \cellcolor{gray!25} & \cellcolor{gray!25} & \cellcolor{gray!25} & \cellcolor{gray!25} & \cellcolor{gray!25} & $(\rightarrow^{iv} = \leftarrow^{cp})$ & \cellcolor{gray!25} & \cellcolor{gray!25} & \cellcolor{gray!25} & \cellcolor{gray!25} & \cellcolor{gray!25} & \cellcolor{gray!25} \\ \hline
& \cellcolor{gray!25}  & \cellcolor{gray!25} \cellcolor{gray!25}  &  \cellcolor{gray!25}    &\cellcolor{gray!25}    & $cd$  & \cellcolor{gray!25}  & \cellcolor{gray!25}  & \cellcolor{gray!25} & \cellcolor{gray!25}  & \cellcolor{gray!25}  & \cellcolor{gray!25}  & \cellcolor{gray!25} \\
$ \vdash F, G $ & \cellcolor{gray!25} & \cellcolor{gray!25}  &  \cellcolor{gray!25}    & \cellcolor{gray!25} & $ \{0,  1\} : \{  1, 3 \}$  & \cellcolor{gray!25} & \cellcolor{gray!25} & \cellcolor{gray!25} & \cellcolor{gray!25}   & \cellcolor{gray!25}  & \cellcolor{gray!25}  & \cellcolor{gray!25}   \\  
& \cellcolor{gray!25}  & \cellcolor{gray!25} \cellcolor{gray!25}  &  \cellcolor{gray!25}    &\cellcolor{gray!25}    & $(\otimes^{cv} = \oplus^{cp})$  & \cellcolor{gray!25}  & \cellcolor{gray!25}  & \cellcolor{gray!25} & \cellcolor{gray!25}  & \cellcolor{gray!25}  & \cellcolor{gray!25}  & \cellcolor{gray!25} \\  \hline
& \cellcolor{gray!25} & \cellcolor{gray!25} & \cellcolor{gray!25} & $cd$ & \cellcolor{gray!25} & \cellcolor{gray!25} & \cellcolor{gray!25} & \cellcolor{gray!25} & \cellcolor{gray!25} & \cellcolor{gray!25} & \cellcolor{gray!25} & $cs$ \\
$F \vdash$ $G \vdash$ & \cellcolor{gray!25} & \cellcolor{gray!25} & \cellcolor{gray!25} & $\{ 1 \} : \{ 1, 2\}$\tnote{\dag} & \cellcolor{gray!25} & \cellcolor{gray!25} & \cellcolor{gray!25} & \cellcolor{gray!25}  & \cellcolor{gray!25} & \cellcolor{gray!25} & \cellcolor{gray!25} & $ \{ 1\} : \{  1 \}$\tnote{\dag} \\  
& \cellcolor{gray!25} & \cellcolor{gray!25} & \cellcolor{gray!25} & $(\otimes^{iv}=\oplus)$ & \cellcolor{gray!25} & \cellcolor{gray!25} & \cellcolor{gray!25} & \cellcolor{gray!25} & \cellcolor{gray!25} & \cellcolor{gray!25} & \cellcolor{gray!25} & $(\sqcap^{iv} = \sqcup)$  \\  \hline
  & \cellcolor{gray!25} & \cellcolor{gray!25} & $cd$ & \cellcolor{gray!25} & \cellcolor{gray!25} & \cellcolor{gray!25} & \cellcolor{gray!25} & \cellcolor{gray!25}  & \cellcolor{gray!25} &\cellcolor{gray!25}  & $cs$ & \cellcolor{gray!25} \\
 $F \vdash$ $ \vdash G$  &\cellcolor{gray!25}  & \cellcolor{gray!25} & $ \{ 1, 2\} : \{ 1 \}$ & \cellcolor{gray!25} & \cellcolor{gray!25} & \cellcolor{gray!25} & \cellcolor{gray!25} & \cellcolor{gray!25}  & \cellcolor{gray!25} & \cellcolor{gray!25} & $\{ 0, 1 \} : \{ 1, 2\}$ & \cellcolor{gray!25} \\  
 & \cellcolor{gray!25} & \cellcolor{gray!25} & $(\rightarrow^{cv} = \leftarrow)$ & \cellcolor{gray!25} & \cellcolor{gray!25} & \cellcolor{gray!25} & \cellcolor{gray!25} & \cellcolor{gray!25}  & \cellcolor{gray!25} & \cellcolor{gray!25} & $(\rightsquigarrow^{cv} = \leftsquigarrow)$ & \cellcolor{gray!25} \\  \hline
  & \cellcolor{gray!25} & $cd$ & \cellcolor{gray!25} & \cellcolor{gray!25} & \cellcolor{gray!25} & \cellcolor{gray!25} & \cellcolor{gray!25} & \cellcolor{gray!25}  & \cellcolor{gray!25} & $cs$ & \cellcolor{gray!25}  & \cellcolor{gray!25}  \\
 $ \vdash F$ $G \vdash$ & \cellcolor{gray!25} & $ \{ 1, 2\} : \{ 1 \}$\tnote{\ddag} & \cellcolor{gray!25} & \cellcolor{gray!25} & \cellcolor{gray!25} & \cellcolor{gray!25} & \cellcolor{gray!25} & \cellcolor{gray!25}  & \cellcolor{gray!25}  & $\{ 0, 1 \} : \{ 1, 2\}$  &  \cellcolor{gray!25} &  \cellcolor{gray!25} \\ 
 & \cellcolor{gray!25} & $(\rightarrow = \leftarrow^{cv})$ & \cellcolor{gray!25} & \cellcolor{gray!25} & \cellcolor{gray!25} & \cellcolor{gray!25} & \cellcolor{gray!25} & \cellcolor{gray!25}  &  \cellcolor{gray!25} & $(\rightsquigarrow = \leftsquigarrow^{cv})$ &   \cellcolor{gray!25}&  \cellcolor{gray!25} \\  \hline
  & $cd$ & \cellcolor{gray!25} & \cellcolor{gray!25} & \cellcolor{gray!25} & \cellcolor{gray!25} & \cellcolor{gray!25} & \cellcolor{gray!25} & \cellcolor{gray!25}  & $cs$ & \cellcolor{gray!25} & \cellcolor{gray!25} & \cellcolor{gray!25} \\
 $ \vdash F$ $ \vdash G$ & $ \{  1, 3 \} : \{ 0, 1\}$ & \cellcolor{gray!25} & \cellcolor{gray!25} & \cellcolor{gray!25} & \cellcolor{gray!25} & \cellcolor{gray!25} & \cellcolor{gray!25} & \cellcolor{gray!25}  & $\{ 0, 1 \} : \{ 1, 2\}$  & \cellcolor{gray!25} &\cellcolor{gray!25}  & \cellcolor{gray!25} \\  
 & $(\otimes^{cp} = \oplus^{cv})$ & \cellcolor{gray!25} & \cellcolor{gray!25} & \cellcolor{gray!25} & \cellcolor{gray!25} & \cellcolor{gray!25} & \cellcolor{gray!25} & \cellcolor{gray!25}  & $(\sqcap^{cp} = \sqcup^{cv})$ & \cellcolor{gray!25} & \cellcolor{gray!25} & \cellcolor{gray!25} \\  \hline
  & \cellcolor{gray!25} & \cellcolor{gray!25} & \cellcolor{gray!25} & \cellcolor{gray!25} & \cellcolor{gray!25} & \cellcolor{gray!25} & \cellcolor{gray!25} & $cs$  & \cellcolor{gray!25} & \cellcolor{gray!25}  & \cellcolor{gray!25} & \cellcolor{gray!25} \\
 $F \vdash$ or $G \vdash$ & \cellcolor{gray!25} & \cellcolor{gray!25} & \cellcolor{gray!25} & \cellcolor{gray!25} & \cellcolor{gray!25} & \cellcolor{gray!25} & \cellcolor{gray!25} & $ \{ 1\} : \{  1 \}$\tnote{\dag} & \cellcolor{gray!25} & \cellcolor{gray!25} & \cellcolor{gray!25} & \cellcolor{gray!25} \\  
& \cellcolor{gray!25} & \cellcolor{gray!25} & \cellcolor{gray!25} & \cellcolor{gray!25}  & \cellcolor{gray!25} & \cellcolor{gray!25} &\cellcolor{gray!25}  & $(\sqcap = \sqcup^{iv})$ & \cellcolor{gray!25} & \cellcolor{gray!25} & \cellcolor{gray!25} & \cellcolor{gray!25} \\  \hline
  & \cellcolor{gray!25} & \cellcolor{gray!25} & \cellcolor{gray!25} & \cellcolor{gray!25} & \cellcolor{gray!25} & \cellcolor{gray!25} & $cs$ & \cellcolor{gray!25}  & \cellcolor{gray!25} &\cellcolor{gray!25}  &\cellcolor{gray!25}  & \cellcolor{gray!25} \\
 $F \vdash$ or $ \vdash G$ & \cellcolor{gray!25} & \cellcolor{gray!25} & \cellcolor{gray!25} & \cellcolor{gray!25}  & \cellcolor{gray!25} & \cellcolor{gray!25} & $ \{ 1, 2\} : \{ 0, 1 \}$ & \cellcolor{gray!25}  & \cellcolor{gray!25} &\cellcolor{gray!25}  &\cellcolor{gray!25}  & \cellcolor{gray!25}  \\  
& \cellcolor{gray!25} & \cellcolor{gray!25} & \cellcolor{gray!25} & \cellcolor{gray!25}  & \cellcolor{gray!25}  & \cellcolor{gray!25} &  $(\rightsquigarrow^{cp} = \leftsquigarrow^{iv})$ & \cellcolor{gray!25} & \cellcolor{gray!25}  & \cellcolor{gray!25} &\cellcolor{gray!25}  &\cellcolor{gray!25}   \\  \hline
  & \cellcolor{gray!25} & \cellcolor{gray!25} & \cellcolor{gray!25} & \cellcolor{gray!25} & \cellcolor{gray!25} & $cs$ & \cellcolor{gray!25} & \cellcolor{gray!25}  & \cellcolor{gray!25} & \cellcolor{gray!25} & \cellcolor{gray!25} & \cellcolor{gray!25} \\
 $ \vdash F$ or $G \vdash$ & \cellcolor{gray!25} & \cellcolor{gray!25} & \cellcolor{gray!25} & \cellcolor{gray!25}  & \cellcolor{gray!25} & $ \{ 1, 2\} : \{ 0, 1 \}$ & \cellcolor{gray!25} & \cellcolor{gray!25} & \cellcolor{gray!25} & \cellcolor{gray!25} & \cellcolor{gray!25} & \cellcolor{gray!25} \\  
& \cellcolor{gray!25} & \cellcolor{gray!25} & \cellcolor{gray!25} & \cellcolor{gray!25}  & \cellcolor{gray!25} & $(\rightsquigarrow^{iv} = \leftsquigarrow^{cp})$ & \cellcolor{gray!25} & \cellcolor{gray!25} & \cellcolor{gray!25} & \cellcolor{gray!25} & \cellcolor{gray!25} & \cellcolor{gray!25}  \\  \hline
  & \cellcolor{gray!25} & \cellcolor{gray!25} & \cellcolor{gray!25} & \cellcolor{gray!25} & $cs$ & \cellcolor{gray!25} & \cellcolor{gray!25} & \cellcolor{gray!25}  &  \cellcolor{gray!25} & \cellcolor{gray!25} & \cellcolor{gray!25} & \cellcolor{gray!25} \\
 $ \vdash F$ or $ \vdash G$ & \cellcolor{gray!25} & \cellcolor{gray!25} & \cellcolor{gray!25} & \cellcolor{gray!25} & $ \{ 1, 2\} : \{ 0, 1 \}$ & \cellcolor{gray!25} & \cellcolor{gray!25} & \cellcolor{gray!25}  &  \cellcolor{gray!25} & \cellcolor{gray!25} & \cellcolor{gray!25} & \cellcolor{gray!25}  \\ 
& \cellcolor{gray!25} & \cellcolor{gray!25} & \cellcolor{gray!25} & \cellcolor{gray!25} & $(\sqcap^{cv} = \sqcup^{cp})$ & \cellcolor{gray!25} & \cellcolor{gray!25} & \cellcolor{gray!25}  &  \cellcolor{gray!25} & \cellcolor{gray!25} & \cellcolor{gray!25} & \cellcolor{gray!25}  \\  \hline
\end{longtable}}
\begin{tablenotes}
{\tiny
\item $F, G \in \{A,B\}$
\item[\dag] First proven in \cite{dicher2016proof}.
\item[\ddag] First proven in \cite{nagler2026measuring}.}
\end{tablenotes}
\end{threeparttable}
\end{sidewaystable}

\section{Discussion: meaning-equivalence across logics and its limits}

In this section, we briefly discuss the philosophical implications of our results by relating them to the overarching inferentialist-minimalist story of I-bS. We recall from \S 2 that the \textit{meaning} of a connective \# is its \textit{use}. Specifically, it is the use of \# required to demonstrate that \# \textit{can be used} in a \textit{minimal context of inference}. The \textit{semantic clause} of \# formally encodes this \textit{meaning}-defining use in terms of \#'s \textit{inference behaviour} relative to \#'s \textit{definability} in a \textit{minimal derivability relation}.

\subsection{Interrelations of connective meanings} 

Let us first illustrate the meaning-theoretic implications of our positive results by giving inference-behaviour semantics for four selected calculi. Afterwards, we will discuss what it means that \textit{only these} connectives have semantic clauses.
\begin{example}[\textbf{LL}]\

    \noindent Consider \cite{paoli2007implicational}'s fragment of (classical sentential) linear logic, the sequent calculus \textbf{LL}, consisting of the structural rules $\textsc{Id}, \textsc{Cut}cd$ and operational rules for the connectives:
    \begin{itemize}
        \item $\sqcap=\langle\epsilon 1A-1B, \delta cs 2A-2B, 2, \mathbb{N}^0:\mathbb{N}^0\rangle$,
        \item $\otimes=\langle\gamma 1A1B, \delta cd 2A-2B, 2, \mathbb{N}^0:\mathbb{N}^0\rangle$,
        \item $\sqcup=\langle \delta cs 1A-1B, \epsilon 2A-2B, 2, \mathbb{N}^0:\mathbb{N}^0\rangle$,
        \item $\oplus=\langle \delta cd 1A-1B, \gamma 2A2B, 2, \mathbb{N}^0:\mathbb{N}^0\rangle$,
        \item $\rightsquigarrow=\langle \delta cs 2A-1B, \epsilon 1A-2B, 2, \mathbb{N}^0:\mathbb{N}^0\rangle$,
        \item $\rightarrow=\langle \delta cd 2A-1B, \gamma 1A2B, 2, \mathbb{N}^0:\mathbb{N}^0\rangle$,
        \item $\bot=\langle \alpha, \gamma 0, 0, \mathbb{N}^0:\mathbb{N}^0\rangle$,
        \item $\top=\langle \gamma 0, \alpha, 0, \mathbb{N}^0:\mathbb{N}^0\rangle$.
    \end{itemize}
    We now obtain I-bS \textit{semantic clauses} for the \textbf{LL}-connectives, relative to a Belnap-style notion of definability and a minimal derivability relation of $\{\textsc{Id}, \textsc{Cut}cd\}$.\footnote{\cite{paoli2007implicational}'s \textbf{LL} does not contain exponentials. However, even if it did, we would not consider them connectives given our definition of operational rules (Definition 3).} As \textbf{LL}-connectives are formulable in two-dimensional sequent calculi ($\leq$ 2 premiss sequents, $\leq$ 2 active formulae), we have already measured their minimal inference behaviour in this study, yielding the following semantic clauses (for simplicity, referred to in terms of their special versions): $\sqcap_{min}$, $\otimes_{min}$, $\sqcup_{min}$, $\oplus_{min}$, $\rightsquigarrow_{min}$, $\rightarrow_{min}$, $\bot_{min}$, $\top_{min}$.

    This measurement constitutes what we called a \textit{semantic reduction} in \S2: using I-bS, we have isolated the \textbf{LL}-connectives in a partitioned substructural \textit{kernel} within \textbf{LL}. Inside the kernel, the minimal derivability relation of $\{\textsc{Id}, \textsc{Cut}cd\}$ is available \textit{globally}, i.e. across partitions. However, each \textbf{LL}-connective is isolated in a \textit{partition} of the kernel in the form of its \textit{semantic clause}. Hence, semantic clauses do not \textit{interact} in the \textbf{LL}-kernel. Instead, it only encompasses theorems of $\{\textsc{Id}, \textsc{Cut}cd, 1\#_{min}, 2\#_{min}\}$, for exactly one \textbf{LL}-connective \# per partition. No theorem in the kernel contains more than one connective. The \textbf{LL}-kernel, therefore, yields a different and weaker derivability relation than full \textbf{LL}, even though both have the same structural rules. In this way, I-bS avoids inter-dependent connective meanings which could not be determined without reference to the full logic, as is the case for B-eS and P-tV.

    The reverse of reduction is \textit{semantic emergence}. To obtain the full derivability relation of \textbf{LL} from its semantic kernel, we departition it to introduce the usual connective interactions. Additionally, we unconstrain frameworks by applying the \textsc{Max}(imise) rule to any semantic clause $\#_{min}=\langle\texttt{SP}(\#_{min}), a, L:R \rangle$:
    \[
    \begin{prooftree}
        \hypo{\langle\texttt{SP}(\#_{min}), a, L:R \rangle}
        \infer1[\textsc{Max}]{\langle\texttt{SP}(\#_{min}), a, \mathbb{N}^0:\mathbb{N}^0 \rangle}
    \end{prooftree}\, . 
    \]
    Note that \textsc{Max} is a \textit{meta-structural} rule. It alters the framework of operational rule pairs, hence transforming the (set-theoretic) structure of the calculus. In contrast to (object-)structural rules such as \textsc{Cut}, which encode structural transformations on \textit{sequents}, meta-structural rules transform the structure of \textit{rules}.

    In short, semantic reduction \textit{partitions} connectives and \textit{minimises} their frameworks, while semantic emergence \textit{departitions} connectives and \textit{maximises} their frameworks (Figure 7). One can think of these meta-structural and, hence, meta-logical validation processes in I-bS as rough equivalents of soundness and completeness proofs in orthodox semantic practice, such as M-tS, P-tV or B-eS. In this traditional picture, semantic justification is statically grounded in the sound and complete match between the proof-theoretic base calculus and the interpreting model- or proof-theoretic structure. In contrast, I-bS's epistemic assurance is \textit{procedural}: the semantic justification lies not in the format of the interpreting structure---the semantic kernel---but in the process of reduction, i.e. the \textit{measurement of minimal inference behaviour}. In this way, I-bS tries to live up to its inferentialist promise of connecting meaning to actual connective use.\begin{figure}
    \centering
    \caption{I-bS for \textbf{LL}}
    \begin{tikzpicture}
        \draw (6.5,0) node[anchor=north]{\textbf{LL}};
        \draw (2,0) -- (11,0) -- (11,4) -- (2,4) -- (2,0);
        \draw[<-] (6.25,1) -- (6.25,3);
        \draw (6.2,2.5) node[anchor=east, fill=white]{reduction: };
        \draw (6.2,2) node[anchor=east, fill=white]{\small $\bullet$ partition};
        \draw (6.25,1.5) node[anchor=east, fill=white]{\small $\bullet$ minimise};
        \draw[->] (6.75,1) -- (6.75,3);
        \draw (6.8,2.5) node[anchor=west, fill=white]{emergence:};
        \draw (6.8,2) node[anchor=west, fill=white]{\small $\bullet$ departition};
        \draw (6.8,1.5) node[anchor=west, fill=white]{\small $\bullet$ maximise (\textsc{Max})};
        \draw (6.5,3) node[anchor=south]{$\vdash_{\textbf{LL}}$};
        \draw (6.5,0.5) ellipse (4.45cm and 0.45cm);
        \draw (6.5,0.5) node{$\sqcap_{min} | \otimes_{min} | \sqcup_{min} | \oplus_{min} | \rightsquigarrow_{min} | \rightarrow_{min} | \bot_{min} | \top_{min}$};
    \end{tikzpicture}
\end{figure}
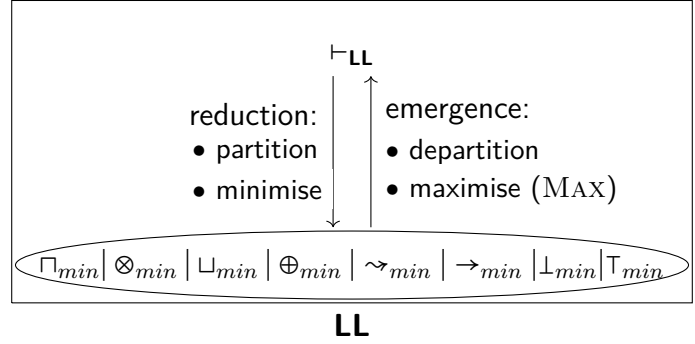
\end{example}
\begin{example}[\textbf{LK}$_{full}$]\

    \noindent Now, consider the sentential fragment of Gentzen's classical \textbf{LK} \cite{gentzen1935auntersuchungen, gentzen1935buntersuchungen}, with both lattice (`additive') and group (`multiplicative') versions of the connectives and added rules for $\bot$ and $\top$. We call this calculus $\textbf{LK}_{full}$. In addition to the structural rules of $\textsc{Id}, \textsc{Cut}cd$, $\textbf{LK}_{full}$ contains \textsc{W(eakening)}  and \textsc{C(ontraction)},
    \[
    \begin{prooftree}
        \hypo{s}
        \infer1[iW]{[i:A]s}
    \end{prooftree}\, , \qquad
   \begin{prooftree}
        \hypo{[i:A,A]s}
        \infer1[iC]{[i:A]s}
    \end{prooftree}\, .
    \]
    for $i \in \{1,2\}$. We use the same operational rules as in Example 14 and add negation:
    \begin{itemize}
        \item  $-=\langle\gamma 2A, \gamma 1A, 1, \mathbb{N}^0:\mathbb{N}^0\rangle$.
    \end{itemize}
    Famously, the lattice and group versions of conjunction, disjunction and implication are each equivalent in the fully Tarskian derivability relation of \textbf{LK}. However, if we measure their minimal inference behaviour, we obtain the same differentiated semantic clauses for them as we did in Example 14.

    Therefore, $\sqcap, \otimes, \sqcup, \oplus, \rightsquigarrow, \rightarrow, \bot, \top$ have the same meaning in \textbf{LL} as they have in $\textbf{LK}_{full}$. This reflects the core feature of I-bS, as discussed in \S2: I-bS allows connectives to share meanings across calculi. `$\otimes$' means the same in classical logic ($\textbf{LK}_{full}$) as it does in linear logic (\textbf{LL}). When one uses I-bS, changing logic does not necessarily mean changing meaning.

    We also noted in \S2 that I-bS enables us to differentiate between the use of a connective stemming from its \textit{meaning} and its use arising from interactions with other connectives or structural features of the derivability relation. Our examples demonstrate this feature: in the present example, the equivalence of group and lattice connectives in $\textbf{LK}_{full}$ is not a consequence of their meaning but rather of the fully Tarskian \textit{derivability relation} of classical logic. In Example 14, the added instances of connective use in \textbf{LL}-proofs with more than one connective do not follow from the connective's meaning but of the \textit{interaction} of different connectives with distinct meanings in linear logic.
    
    \textbf{LK}-negation $-$ provides another interesting case study for our results. Measuring its minimal inference behaviour, we find \textit{two} semantic clauses, $\neg_{min}$ and $\sim_{min}$ (see also Example 11). Thus, I-bS shows that there are two different negations in \textbf{LK} with distinct meanings. The fact that \textbf{LK} conflates the two connectives in the form of $-$ is a consequence of the classical derivability relation (where $[i:\neg A, j:\sim A]$), and not a semantic feature of negation. We will come back to negation in the following example.

    To re-emerge the semantic kernel of $\textbf{LK}_{full}$, we again departition, add \textsc{Max} and add remaining full structural rules ($i\textsc{W}/i\textsc{C}$) (Figure 8).
    \begin{figure}
    \centering
    \caption{I-bS for \textbf{LK}$_{full}$}
    \begin{tikzpicture}
        \draw (6.5,0) node[anchor=north]{\textbf{LK}$_{full}$};
        \draw (1,0) -- (12,0) -- (12,4) -- (1,4) -- (1,0);
        \draw[<-] (6.25,1) -- (6.25,3);
        \draw (6.2,3) node[anchor=east, fill=white]{reduction: };
        \draw (6.2,2.5) node[anchor=east, fill=white]{\small $\bullet$ partition};
        \draw (6.25,2) node[anchor=east, fill=white]{\small $\bullet$ minimise};
        \draw (6.25,1.5) node[anchor=east, fill=white]{\small $\bullet$---$i\textsc{W}/i\textsc{C}$};
        \draw[->] (6.75,1) -- (6.75,3);
        \draw (6.8,3) node[anchor=west, fill=white]{emergence:};
        \draw (6.8,2.5) node[anchor=west, fill=white]{\small $\bullet$ departition};
        \draw (6.8,2) node[anchor=west, fill=white]{\small $\bullet$ maximise (\textsc{Max})};
        \draw (6.8,1.5) node[anchor=west, fill=white]{\small $\bullet$ + $i\textsc{W}/i\textsc{C}$};
        \draw (6.5,3) node[anchor=south]{$\vdash_{\textbf{LK}_{full}}$};
        \draw (6.5,0.5) ellipse (5.45cm and 0.45cm);
        \draw (6.5,0.5) node{$\sqcap_{min} | \otimes_{min} | \sqcup_{min} | \oplus_{min} | \rightsquigarrow_{min} | \rightarrow_{min} | \bot_{min} | \top_{min} | \neg_{min} | \sim_{min}$};
    \end{tikzpicture}
\end{figure}
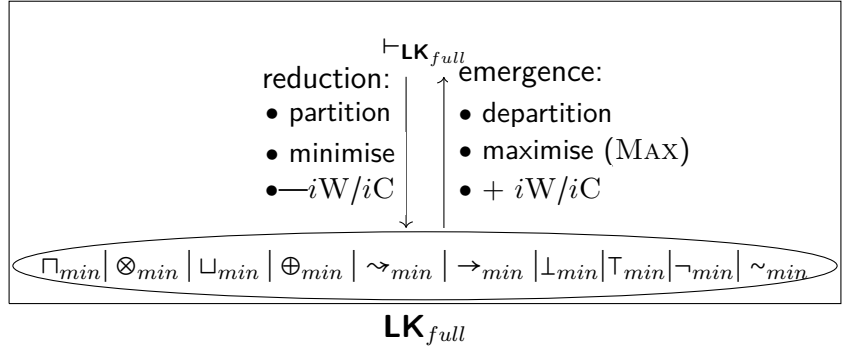
\end{example}
\begin{example}[\textbf{LJ}$_{full}$/\textbf{LDJ}$_{full}$]\

    \noindent If we restrict all \textbf{LK}$_{full}$-rules to the framework $\mathbb{N}^0:\{0,1\}$ and remove the rules for $\oplus$, $\rightsquigarrow$, and $2\textsc{C}$, we obtain full sentential intuitionistic logic \textbf{LJ}$_{full}$. Once again, group and lattice conjunction are equivalent in the full derivability relation. Applying I-bS, we find the same semantic clauses for \textbf{LJ}$_{full}$ as we did for \textbf{LL} (except that there are no clauses for $\oplus$, $\rightsquigarrow$). Therefore, \textbf{LJ}$_{full}$-connectives have the same meaning as their \textbf{LL}-counterparts and, \textit{a fortiori}, their \textbf{LK}$_{full}$-equivalents.
    
    However, there is one exception: $\neg_{min}$ is the \textit{only} connective corresponding to $-$ in \textbf{LJ}$_{full}$. This is because $\sim_{min}$'s minimal framework is incompatible with that of \textbf{LJ}$_{full}$. I-bS gives a semantic explanation for this: $\neg_{min}$ means the same in \textbf{LK}$_{full}$ and \textbf{LJ}$_{full}$. Although \textbf{LK}$_{full}$'s permissive framework conflates $\neg_{min}$ and $\sim_{min}$, they are semantically \textit{distinct} connectives. As a consequence, `negation' ($\neg_{min}$) in \textbf{LJ} only captures half of the meaning of `negation' ($\neg_{min}$/$\sim_{min}$) in \textbf{LK}.

    Intriguingly, \textbf{LDJ}-negation captures the other half. We obtain sentential dual-intuitionistic logic $\textbf{LDJ}_{full}$  \cite{urbas1996dual} from \textbf{LK}$_{full}$ if we restrict all the rules of the latter to the framework $\{0,1\}:\mathbb{N}^0$ and remove the rules for $\otimes$, $\rightarrow$, and $1\textsc{C}$. Here, group and lattice disjunction are equivalent. Computing the I-bS semantic clauses, we find that each connective in $\textbf{LDJ}_{full}$ shares semantic clause and, thus, meaning with its respective counterpart for $\textbf{LK}_{full}$ and \textbf{LL}, insofar that there is one. However, for negation, we now obtain $\sim_{min}$---the other half of the semantics of $\textbf{LK}_{full}$-negation. 

    A similar taxonomy emerges for the other connectives. In $\textbf{LK}$, `conjunction', `disjunction' and `implication' are each usually presented using a single operational rule pair. However, I-bS finds that each $\textbf{LK}$-rule corresponds to \textit{two} meaningful connectives:  a \textit{group} (`multiplicative') and a \textit{lattice} (`additive') version. The case of implication is fully parallel to that of negation. While only $\rightarrow_{min}$ has meaning in \textbf{LJ},  $\rightsquigarrow_{min}$ only has a semantic clause in \textbf{LDJ}. Hence, intuitionistic `if' and dual-intuitionistic `if' do not share meaning, but both carry half of the meaning of classical `if'. Disjunction and conjunction behave somewhat differently. While there are semantic clauses for both $\oplus_{min}$ and $\sqcup_{min}$ in \textbf{LJ}, only $\sqcup_{min}$ is meaningful in \textbf{LDJ}. Dually, $\otimes_{min}$ and $\sqcap_{min}$ have meaning in \textbf{LDJ}, but only $\sqcap_{min}$ does in \textbf{LJ}.

    In terms of this analysis, I-bS can give a semantic response to the question: `Does classical negation (conjunction, disjunction, implication) mean the same as intuitionistic negation (conjunction, disjunction, implication)?' Only \textit{half} of the classical disjunctions, negations and implications share meaning with the corresponding intuitionistic disjunction, negation and implication, namely $\sqcup_{min}$, $\neg_{min}$ and $\rightarrow_{min}$, respectively. Similarly, only \textit{half} of the conjunctions, negations and implications in \textbf{LK} share meaning with dual-intuitionistic conjunction, negation and implication, namely $\sqcap_{min}$, $\sim_{min}$ and $\rightsquigarrow_{min}$, respectively. At the same time, both intuitionistic conjunctions and both dual-intuitionistic disjunctions share meaning with their classical counterparts.

    Thus, we have validated I-bS's main promise: a precise analysis of the differences and interrelations of connective meanings across different logics. In Table 17, we illustrate this picture by showing which connectives are I-bS compatible with selected calculi, namely: classical  \textbf{LK} (framework: $\mathbb{N}^{0}:\mathbb{N}^{0}$), intuitionistic \textbf{LJ} (framework: $\mathbb{N}^{0}:\{0,1\}$), dual-intuitionistic \textbf{LDJ} (framework: $\{0,1\}:\mathbb{N}^{0}$), minimal \textbf{LM} (framework: $\mathbb{N}^{0}:\{1\}$) \cite{johansson1937minimalkalkul} and dual-minimal logic \textbf{LDM} (framework: $\{1\}:\mathbb{N}^{0}$),\footnote{\textbf{LDM} is defined as the dual of \textbf{LM}, in the way that \textbf{LDJ} is the dual of \textbf{LJ}.} as well as lattice logic \textbf{LaL} \cite{restall2005geometry} (framework: $\{0,1\}:\{0,1\}$) and the minimal version thereof, minimal lattice logic \textbf{LaLM} (framework: $\{1\}:\{1\}$).\footnote{\textbf{LaLM} has the same rules as \textbf{LaL}, modulo the displayed framework restriction.}

\begin{table}
    \centering
    \begin{tabular}{|c|c|c|c|c|c|c|c|}
  \hline
  & \textbf{LK} & \textbf{LJ} & \textbf{LDJ} & \textbf{LM} & \textbf{LDM} & \textbf{LaL} & \textbf{LaLM}
         \\ \hline
       $\bot_{min}$  & $\checkmark$ & $\checkmark$ & $\checkmark$ & \cellcolor{gray!25} & $\checkmark$ & $\checkmark$ & \cellcolor{gray!25} \\ \hline
       $\top_{min}$ & $\checkmark$ & $\checkmark$ & $\checkmark$ & $\checkmark$ & \cellcolor{gray!25} & $\checkmark$ & \cellcolor{gray!25} \\ \hline

       $\neg_{min}$ & $\checkmark$ & $\checkmark$ & \cellcolor{gray!25} & \cellcolor{gray!25} & \cellcolor{gray!25} & \cellcolor{gray!25} & \cellcolor{gray!25} \\ \hline
       $\sim_{min}$ & $\checkmark$ & \cellcolor{gray!25} & $\checkmark$ & \cellcolor{gray!25} & \cellcolor{gray!25} & \cellcolor{gray!25} & \cellcolor{gray!25} \\ \hline
       $\otimes_{min}$ & $\checkmark$ & $\checkmark$ & \cellcolor{gray!25} & $\checkmark$ & \cellcolor{gray!25} & \cellcolor{gray!25} & \cellcolor{gray!25} \\ \hline
        $\neg^{cv}_{min}$ & $\checkmark$ & $\checkmark$ & $\checkmark$ & $\checkmark$ & $\checkmark$ & $\checkmark$ & $\checkmark$ \\ \hline
        $\sqcap_{min}$ & $\checkmark$ & $\checkmark$ & $\checkmark$ & $\checkmark$ & $\checkmark$ & $\checkmark$ & $\checkmark$ \\\hline
        $\oplus_{min}$ & $\checkmark$ & \cellcolor{gray!25} & $\checkmark$ & \cellcolor{gray!25} & $\checkmark$ & \cellcolor{gray!25} & \cellcolor{gray!25} \\ \hline
        $\sqcup_{min}$ & $\checkmark$ & $\checkmark$ & $\checkmark$ & $\checkmark$ & $\checkmark$ & $\checkmark$ & $\checkmark$ \\ \hline 
        $\rightarrow_{min}$ & $\checkmark$ & $\checkmark$ & \cellcolor{gray!25} & $\checkmark$ & \cellcolor{gray!25} & \cellcolor{gray!25} & \cellcolor{gray!25} \\ \hline
        $\rightsquigarrow_{min}$ & $\checkmark$ & \cellcolor{gray!25} & $\checkmark$ & \cellcolor{gray!25} & \cellcolor{gray!25} & \cellcolor{gray!25} & \cellcolor{gray!25} \\ \hline
        $\leftarrow_{min}$ & $\checkmark$ & $\checkmark$ & \cellcolor{gray!25} & $\checkmark$ & \cellcolor{gray!25} & \cellcolor{gray!25} & \cellcolor{gray!25} \\ \hline
        $\leftsquigarrow_{min}$ & $\checkmark$ & \cellcolor{gray!25} & $\checkmark$ & \cellcolor{gray!25} & \cellcolor{gray!25} & \cellcolor{gray!25} & \cellcolor{gray!25} \\ \hline
        
       $\otimes_{min}^{cv}$ & $\checkmark$ & \cellcolor{gray!25} & $\checkmark$ & \cellcolor{gray!25} & \cellcolor{gray!25} & \cellcolor{gray!25} & \cellcolor{gray!25} \\ \hline
    $\sqcap^{cv}_{min}$ & $\checkmark$ & $\checkmark$ & \cellcolor{gray!25} & \cellcolor{gray!25} & \cellcolor{gray!25} & \cellcolor{gray!25} & \cellcolor{gray!25} \\ \hline
       $\oplus_{min}^{cv}$ & $\checkmark$ & $\checkmark$ & \cellcolor{gray!25} & \cellcolor{gray!25} & \cellcolor{gray!25} & \cellcolor{gray!25} & \cellcolor{gray!25} \\ \hline
       $\sqcup^{cv}_{min}$ & $\checkmark$ & \cellcolor{gray!25} & $\checkmark$ & \cellcolor{gray!25} & \cellcolor{gray!25} & \cellcolor{gray!25} & \cellcolor{gray!25} \\ \hline
    $\rightarrow_{min}^{iv}$ & $\checkmark$ & \cellcolor{gray!25} & $\checkmark$ & \cellcolor{gray!25} & $\checkmark$ & \cellcolor{gray!25} & \cellcolor{gray!25} \\ \hline
       $\rightsquigarrow^{iv}_{min}$ & $\checkmark$ & $\checkmark$ & \cellcolor{gray!25} & \cellcolor{gray!25} & \cellcolor{gray!25} & \cellcolor{gray!25} & \cellcolor{gray!25} \\ \hline
       $\leftarrow_{min}^{iv}$ & $\checkmark$ & \cellcolor{gray!25} & $\checkmark$ & \cellcolor{gray!25} & $\checkmark$ & \cellcolor{gray!25} & \cellcolor{gray!25} \\ \hline
       $\leftsquigarrow^{iv}_{min}$ & $\checkmark$ & $\checkmark$ & \cellcolor{gray!25} & \cellcolor{gray!25} & \cellcolor{gray!25} & \cellcolor{gray!25} & \cellcolor{gray!25} \\ \hline
    \end{tabular}
    \caption{Overview of Semantic Compatibilities among Selected Calculi}
\end{table}

    Before reflecting on the \textit{exhaustivity} of our results, let us finish giving I-bS for $\textbf{LJ}_{full}$ and  $\textbf{LDJ}_{full}$. To re-emerge the semantic kernel, we again departition, add structural rules, and apply suitable maximisation rules. To complete the semantic picture for $\textbf{LJ}_{full}$ (Figure 9), we apply the rule \textsc{Maxi} to any semantic clause $\#_{min}=\langle\texttt{SP}(\#_{min}), a, L:R \rangle$. To get the scheme for $\textbf{LDJ}_{full}$ (Figure 10), we apply \textsc{Maxdi} instead.
    \[
    \begin{prooftree}
        \hypo{\langle\texttt{SP}(\#_{min}), a, L:R \rangle}
        \infer1[\textsc{Maxi}]{\langle\texttt{SP}(\#_{min}), a, \mathbb{N}^0:\{0,1\}\rangle}
    \end{prooftree}\, ,
    \qquad 
    \begin{prooftree}
        \hypo{\langle\texttt{SP}(\#_{min}), a, L:R \rangle}
        \infer1[\textsc{Maxdi}]{\langle\texttt{SP}(\#_{min}), a, \{0,1\}:\mathbb{N}^0\rangle}
    \end{prooftree}\, .
    \]
       \begin{figure}
    \centering
    \caption{I-bS for \textbf{LJ}$_{full}$}
    \begin{tikzpicture}
        \draw (6.5,0) node[anchor=north]{\textbf{LJ}$_{full}$};
        \draw (2.5,0) -- (10.5,0) -- (10.5,4) -- (2.5,4) -- (2.5,0);
        \draw[<-] (6.25,1) -- (6.25,3);
        \draw (6.2,3) node[anchor=east, fill=white]{reduction: };
        \draw (6.2,2.5) node[anchor=east, fill=white]{\small $\bullet$ partition};
        \draw (6.25,2) node[anchor=east, fill=white]{\small $\bullet$ minimise};
        \draw (6.25,1.5) node[anchor=east, fill=white]{\small $\bullet$ - $i\textsc{W}/2\textsc{C}$};
        \draw[->] (6.75,1) -- (6.75,3);
        \draw (6.8,3) node[anchor=west, fill=white]{emergence:};
        \draw (6.8,2.5) node[anchor=west, fill=white]{\small $\bullet$ departition};
        \draw (6.8,2) node[anchor=west, fill=white]{\small $\bullet$ maximise (\textsc{Maxi})};
        \draw (6.8,1.5) node[anchor=west, fill=white]{\small $\bullet$ + $i\textsc{W}/2\textsc{C}$};
        \draw (6.5,3) node[anchor=south]{$\vdash_{\textbf{LJ}_{full}}$};
        \draw (6.5,0.5) ellipse (3.75cm and 0.45cm);
        \draw (6.5,0.5) node{$\sqcap_{min} | \otimes_{min} | \sqcup_{min} | \rightarrow_{min} | \bot_{min} | \top_{min} | \neg_{min}$};
    \end{tikzpicture}
    \end{figure}
    
    \begin{figure}
    \centering
    \caption{I-bS for \textbf{LDJ}$_{full}$}
    \begin{tikzpicture}
        \draw (6.5,0) node[anchor=north]{\textbf{LDJ}$_{full}$};
        \draw (2.4,0) -- (10.6,0) -- (10.6,4) -- (2.4,4) -- (2.4,0);
        \draw[<-] (6.25,1) -- (6.25,3);
        \draw (6.2,3) node[anchor=east, fill=white]{reduction: };
        \draw (6.2,2.5) node[anchor=east, fill=white]{\small $\bullet$ partition};
        \draw (6.25,2) node[anchor=east, fill=white]{\small $\bullet$ minimise};
        \draw (6.25,1.5) node[anchor=east, fill=white]{\small $\bullet$ - $i\textsc{W}/1\textsc{C}$};
        \draw[->] (6.75,1) -- (6.75,3);
        \draw (6.8,3) node[anchor=west, fill=white]{emergence:};
        \draw (6.8,2.5) node[anchor=west, fill=white]{\small $\bullet$ departition};
        \draw (6.8,2) node[anchor=west, fill=white]{\small $\bullet$ maximise (\textsc{Maxdi})};
        \draw (6.8,1.5) node[anchor=west, fill=white]{\small $\bullet$ + $i\textsc{W}/1\textsc{C}$};
        \draw (6.5,3) node[anchor=south]{$\vdash_{\textbf{LDJ}_{full}}$};
        \draw (6.5,0.5) ellipse (3.75cm and 0.45cm);
        \draw (6.5,0.5) node{$\sqcap_{min} | \sqcup_{min} | \oplus_{min} | \rightsquigarrow_{min} | \bot_{min}| \top_{min} | \sim_{min}$};
    \end{tikzpicture}
\end{figure}
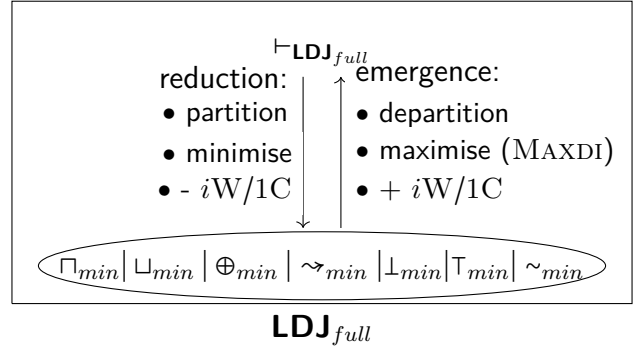
\end{example}

\subsection{Exclusivity of meaning}

We have seen how the positive results of this study deliver on some of the key promises of I-bS. We have provided a frugal ontology of $21$ minimal connectives meanings. Using the corresponding semantic clauses, we have given I-bS for a range of calculi, including ones for linear, classical, (dual-)intuitionistic, minimal, and lattice logics. We have also shown how I-bS can be conceptualised as a reductive substructural kernel within the full derivability relation of a logic. Based on our results, we have provided semantic explanations for the use differences of connectives across derivability relations. Most importantly, we have demonstrated how I-bS provides connective meanings without recourse to a specific calculus.

Thus, we have \textit{validated} that the positive results do what we wanted our semantics to do. Let us now argue that, moreover, the negative part of our findings \textit{verifies} I-bS: starting from a vast dataset of $10,816$ connectives, I-bS finds exactly and only $21$ fixed points, which correspond exactly and only to the connectives in the standard calculi discussed in the previous section (and their converses/inverses) and provide exactly and only the meaning-theoretic explanations we were seeking. In this sense, we have proven the \textit{correctness} of I-bS as a theory of meaning.

We designed this study actively so to minimise bias in favour of any particular logic or semantic theory. The data for our measurement of inference behaviour are generated from only on four assumptions:
\begin{itemize}
    \item \textit{Operationality:} separate, symmetrical and explicit (Definition 3) operational rules adequately specify a connective.
    \item \textit{Formulability:} Table 3 covers all operational rules that can be formulated (in two-dimensional sequent calculi using $\leq 2$ premiss sequents and $\leq 2$ active formulae),
    \item \textit{Definability:} conservativity and uniqueness (Definitions 14-16) capture the conditions under which an operational rule can be defined.
    \item \textit{Substructurality:} $\textsc{Id}$ and $\textsc{Cut} cd$ form an adequate substructural derivability relation.
\end{itemize}
We argue that none of these notions introduces an impermissible bias in favour of I-bS. \textit{Operationality} captures how connectives are usually defined in sequent calculi, including in all calculi mentioned thus far. Even when some of the norms in Definition 3 are violated---e.g. explicitness for implication in \textbf{G3ip} \cite{negri2001structural} or \textbf{BiInt} \cite{pinto2011relating}---the resulting `deviance' of the connective is typically assessed against a canonical benchmark of separate, explicit, and symmetrical rules. The definitional character of such rules is tightly linked to the origins of contemporary Gentzen-style calculi \cite{gentzen1935auntersuchungen, gentzen1935buntersuchungen} and forms a focal point of the literature on rule-based meaning-determination \cite{wansing2000idea, zucker1978adequacy, prawitz1979proofs}.

We take \textit{formulability} to be justified \textit{eo ipso}. In \S3, we have already discussed the pragmatic restrictions on premiss sequents and active formulae, as has our use of standard two-dimensional sequent calculi. The latter design choice, someone might object, runs the risk of introducing some latent bias in favour of classical logic. However, no such bias is reflected in our results as we account for the meaning of classical, intuitionistic and other connectives on equal grounds. I-bS alone gives no reason to favour one logic over another. In contrast, alternative forms of P-tS such as B-eS or P-tV have traditionally been argued to favour intuitionistic logic \cite{sep-proof-theoretic-semantics, barroso2025sequent}.

While the semantic significance of \textit{definability} might be debatable, conservative extensions and uniqueness up to isomorphism are established assurance concepts in mathematics. As such, they are standard sanity checks for the non-expansive deductive power and determinacy of deductive systems, including a formal semantic taxonomy.

Although we justified the minimality of $\{\textsc{Id}, \textsc{Cut}\}$ on inner-theoretic, exceptionalist grounds (\S 4.3), the \textit{substructurality} of the rule pair itself is not inherently biased in favour of I-bS. Not only are $\textsc{Id}$ and $\textsc{Cut}$ part of most standard sequent calculi, including all calculi discussed thus far, but they also form the derivability relation of linear logic, the weakest element of the standard substructural hierarchy (in the sense of the least structural resources and the most distinctions made) \cite{restall2002introduction}.\footnote{Disregarding \textsc{Exchange}-free ordered logic (see also \S 4.3).}

Therefore, we did not introduce any unwanted bias in favour of I-bS when we generated the the $10,816$ candidate connectives and their corresponding definability proofs within which we measured inference behaviour. 

Given this background, we find it \textit{genuinely surprising} that as we apply I-bS to these data, we find exactly and only $21$ semantic clauses. Not only do all of these correspond to the connectives of standard logics such as classical, intuitionistic and dual-intuitionistic logic (or their converses/inverses), but I-bS classifies them as the \textit{minimally meaningful} connectives of sentential logics in two-dimensional derivability relations.\footnote{The dimensionality of a calculus may be associated with the number of truth values. For a compact overview, see \cite{zach2005manysided}.} 

We believe that this result shows that I-bS \textit{adequately} and \textit{correctly} captures at least \textit{some} core feature of the connectives. If one trusts the foundational ideas of P-tS which I-bS seeks to implement (\S\S2, 4.1) as well as its technical toolbox (\S4) and semantic narrative (\S 8.1, \cite{nagler2026measuring}), then this study validates and verifies I-bS as a proof-theoretic semantic approach. It also motivates further research into the power of the I-bS approach, for instance through multi-dimensional extensions or a theory for quantification or modal operators. Moreover, some conjectures in this study are yet to be verified, and replication studies with varied parameters would be worthwhile, such as using natural deduction formalisations, an alternative minimal derivability (e.g. context-sharing \textsc{Cut}) relation or another notion of harmony (e.g. invertibility).

\section*{Summary}
In this paper, we have formalised inference-behaviour semantics (I-bS) and validated and verified it in a broad proof-theoretic setting.

We applied I-bS to all connectives whose operational rules are explicit, separate and symmetrical, and that use at most two premiss sequents and at most two active formulae in a two-dimensional sequent calculus. As a result, we identified $21$ minimally meaningful connectives: bottom and top, intuitionistic and dual-intuitionistic negation, and the additive and multiplicative versions of conjunction, disjunction, right- and left-implication, together with their converses and inverses.

We then explored semantic compatibilities across logics. All $21$ connectives are compatible with classical logic and the connective-only fragment of linear logic. By contrast, intuitionistic and dual-intuitionistic logic are each compatible with only one of the two forms of negation and implication available in classical logic. Moreover, intuitionistic logic is compatible with both conjunctions but only one disjunction, whereas dual-intuitionistic logic is compatible with both disjunctions but only one conjunction.
\subsection*{Acknowledgements}
I sincerely thank Greg Restall, Luca Incurvati, and Franz Berto for their invaluable mentorship. Special thanks to Sara Ayhan, Victor Barroso-Nascimento, Alexander Gheorghiu, Timo Lang, Dale Miller, Thomas Piecha, Elaine Pimentel, David Pym, Sebastian Speitel, Richard Zach, the \textit{Proof Society} workshop 2024, the \textit{Proof-Theoretic Semantics Seminar Series}, the \textit{Logic Colloquium} 2025, the \textit{5th Symposium on Proof-theoretic Semantics}, the Amsterdam \textit{Meaning, Language and Cognition Seminar}, and the St Andrews \textit{Metaphysics and Logic Group} for their helpful comments and feedback on earlier versions of this material. This research was supported by the Linda and Gordon Bonnyman Charitable Trust via the University of St Andrews.
\printbibliography
\end{document}